\documentclass[11pt, letterpaper]{article}
\usepackage[utf8]{inputenc}
\usepackage{fullpage,times}

\usepackage{amsmath,amsfonts,amsthm,amssymb,multirow}
\usepackage{times}
\usepackage{graphicx}
\usepackage{floatpag}
\usepackage{algorithm}
\usepackage[noend]{algpseudocode}
\usepackage{bbold}
\usepackage{enumitem}
\usepackage{subcaption} 
\usepackage{calc}
\usepackage{tikz}
\usetikzlibrary{decorations.markings}
\tikzstyle{vertex}=[circle, draw, inner sep=0pt, minimum size=4pt, fill = black]

\usepackage{graphicx}
\usepackage{tabularx}

\usepackage[ocgcolorlinks]{hyperref} % Option ocgcolorlinks 
\usepackage{xcolor}
\usepackage{xspace}
\usepackage{xfrac}

\colorlet{DarkRed}{red!50!black}
\colorlet{DarkGreen}{green!50!black}
\colorlet{DarkBlue}{blue!50!black}

\hypersetup{
	linkcolor = DarkRed,
	citecolor = DarkGreen,
	urlcolor = DarkBlue,
	bookmarksnumbered = true,
	linktocpage = true
}

\newcommand{\dist}{\mathbf{dist}}
\usepackage[normalem]{ulem}
\usepackage{cleveref}

\makeatletter
\newcommand{\multiline}[1]{%
  \begin{tabularx}{\dimexpr\linewidth-\ALG@thistlm}[t]{@{}X@{}}
    #1
  \end{tabularx}
}
\makeatother
\usepackage{mathtools}
\usepackage{thm-restate}

\makeatletter
\def\BState{\State\hskip-\ALG@thistlm}
\makeatother

\newcommand{\base}{b}
\newcommand{\sumC}[3]{c_{#1, #3}}

\usepackage[compact]{titlesec}
\titlespacing{\section}{0pt}{3ex}{2ex}
\titlespacing{\subsection}{0pt}{2ex}{1ex}
\titlespacing{\subsubsection}{0pt}{0.5ex}{0ex}

\newtheorem{theorem}{Theorem}[section]
\newtheorem{fact}[theorem]{Fact}

\newtheorem{corollary}[theorem]{Corollary}

\newtheorem{definition}[theorem]{Definition}
\newtheorem{lemma}[theorem]{Lemma}
\newtheorem{claim}[theorem]{Claim}
\newtheorem{hypothesis}[theorem]{Hypothesis}
\newtheorem{conjecture}[theorem]{Conjecture}

\makeatletter
\let\c@fconjecture\c@conjecture
\makeatother

\makeatletter
\let\c@fconj\c@conj
\makeatother

\def \eps {\kappa}

\newcommand{\ignore}[1]{}

\title{New Techniques and Fine-Grained Hardness for Dynamic Near-Additive Spanners}
\author{Thiago Bergamaschi \\ MIT \and Monika Henzinger\\University of Vienna \and Maximilian Probst Gutenberg\\ETH Zurich \and Virginia Vassilevska Williams\\MIT\and Nicole Wein\\MIT}

\date{}

\begin{document}
\maketitle
\thispagestyle{empty}
\begin{abstract}
Maintaining and updating shortest paths information in a graph is a fundamental problem with many applications. As computations on dense graphs can be prohibitively expensive, and it is preferable to perform the computations on a sparse skeleton of the given graph that roughly preserves the shortest paths information. Spanners and emulators serve this purpose. Unfortunately, very little is known about dynamically maintaining sparse spanners and emulators as the graph is modified by a sequence of  edge insertions and deletions.  This paper develops fast dynamic algorithms for spanner and emulator maintenance and provides evidence from fine-grained complexity that these algorithms are tight. For unweighted undirected $m$-edge $n$-node graphs we obtain the following results.

Under the popular OMv conjecture, there can be no decremental or incremental algorithm that maintains an $n^{1+o(1)}$ edge (purely additive) $+n^{\delta}$-emulator for any $\delta<1/2$ with arbitrary polynomial preprocessing time and total update time $m^{1+o(1)}$. Also, under the Combinatorial $k$-Clique hypothesis, any fully dynamic combinatorial algorithm  that maintains an $n^{1+o(1)}$ edge $(1+\epsilon,n^{o(1)})$-spanner or emulator for small $\epsilon$ must either have preprocessing time $mn^{1-o(1)}$ or amortized update time $m^{1-o(1)}$. Both of our conditional lower bounds are tight.

As the above fully dynamic lower bound only applies to combinatorial algorithms, we also develop an algebraic spanner algorithm that improves over the $m^{1-o(1)}$ update time for dense graphs. For any constant $\epsilon\in (0,1]$, there is a fully dynamic algorithm with worst-case update time $O(n^{1.529})$ that whp maintains an $n^{1+o(1)}$ edge $(1+\epsilon,n^{o(1)})$-spanner.

Our new algebraic techniques allow us to also obtain a new fully dynamic algorithm for All-Pairs Shortest Paths (APSP) that can perform both edge updates and can report shortest paths in worst-case time $O(n^{1.9})$, which are correct whp. This is the first \emph{path-reporting} fully dynamic APSP algorithm  with a truly subquadratic query time that beats $O(n^{2.5})$ update time. It works against an oblivious adversary.

Finally, we give two applications of our new dynamic spanner algorithms: (1) a fully dynamic $(1+\epsilon)$-approximate APSP algorithm with update time $O(n^{1.529})$ that can report approximate shortest paths in $n^{1+o(1)}$ time per query; previous subquadratic update/query algorithms could only report the distance, but not obtain the paths; (2) a fully dynamic algorithm for near-$2$-approximate Steiner tree maintenance with both terminal and edge updates.

% ***
% worst case against oblivious adversary,
% tight spanners, dynamic path-reporting APSP, Steiner tree with edge/terminal updates
\end{abstract}
\clearpage
\pagenumbering{arabic}

%overview
%1-2 page fully dyn spanner, comparison with other techniques
%1-2 page path reporting algebraic
%simple lower bound or overview

\section{Introduction}

Computing shortest paths in a graph is a fundamental problem with many applications. 
However, as on dense graphs the running time can be prohibitively expensive, it is preferable to perform the computation on a sparser representation of the given graph that approximately preserves the shortest paths distances. Such representations are called \emph{spanners} and \emph{emulators}.
Given an undirected, unweighted graph $G=(V,E)$, a subgraph $H$ of $G$ is defined to be an \emph{$(\alpha, \beta)$-spanner} if for every pair of vertices $x,y \in V$, we have that 
\[
     \mathbf{dist}_G(x,y) \leq \mathbf{dist}_H(x,y) \leq \alpha \cdot \mathbf{dist}_G(x,y) + \beta.
\]
A graph $H=(V,E')$ is defined to be an \emph{$(\alpha, \beta)$-emulator} if it fulfills the above constraint, is possibly edge-weighted, but  not necessarily a subgraph of $G$. 
Thus every spanner is also an emulator but not vice versa. When evaluating the quality of a spanner or emulator $H$
three parameters are of interest: the \emph{multiplicative approximation} $\alpha$, the \emph{additive approximation} $\beta$ and the \emph{sparsity} of $H$ that is the number of edges in $H$. 

Spanners and emulators have a variety of applications, ranging from efficient routing to parallel and distributed algorithms to efficient distance oracles, i.e. data structures that answer shortest-path queries. Thus, there exists a large body of work on computing spanners (see below). 
As real-world graphs are often dynamic, it raises the question whether spanners and emulators can be maintained efficiently when the graph is modified by edge updates. 
Unfortunately, very little is known about this question. 
In this article, we are concerned with the design of efficient algorithms to {\em dynamically} maintain an $(1+\epsilon, \beta)$-spanner on an undirected unweighted graph that is undergoing edge insertions and deletions.
% $H$ on 
% a dynamic undirected and unweighted graph $G=(V,E)$, that is a graph undergoing edge insertion and deletions. 

If the update sequence is restricted to consist exclusively of insertions, we say that the graph is \emph{incremental} and if it only consists of edge deletions we say that it is \emph{decremental}. 
If the graph is either incremental or decremental, we also say it is \emph{partially dynamic} and otherwise we say it is \emph{fully dynamic}.

Apart from giving new decremental and fully dynamic deterministic and randomized algorithms that maintain spanners and emulators we also provide evidence from fine-grained complexity that these algorithms are tight. We then use our new algorithms and techniques 
to give novel fully dynamic approximate and exact all-pairs shortest paths (APSP) algorithms {\em that can report the corresponding shortest path,} addressing an open question raised by~\cite{BrandN19}. We further provide applications for other problems such as the maintenance of an approximate Steiner tree of a graph.

\paragraph{Prior Work.}
In this section, we discuss work that is directly related to the results in our paper. We use $\tilde{O}$-notation to suppress logarithmic factors, let $n$ and $m$ be the maximum number of vertices and edges respectively in any version of the graph under consideration. Unless otherwise specified, all graphs are undirected and unweighted. To ease the discussion, we assume for the rest of the section that $\epsilon$ is a constant.

\paragraph{Spanners and Emulators.} Spanners and static algorithms to construct them have been studied in great detail for \emph{multiplicative} approximation \cite{awerbuch1985complexity, peleg1989graph,althofer1993sparse, cohen1993fast,awerbuch1998near,baswana2003simple, roditty2005deterministic} culminating in near-optimal algorithms to construct $(2k-1, 0)$-spanners of sparsity $\tilde{O}(n^{1+1/k})$. There has also been an extensive line of research on \emph{purely additive} spanners \cite{AingworthCM96,AingworthCIM99,dor2000all, bollobas2005sparse,baswana2005new,chechik2013new} where $(1,2), (1, 4)$ and $(1,6)$-spanners are known of sparsity $\tilde{O}(n^{3/2}), \tilde{O}(n^{7/5})$ and $\tilde{O}(n^{4/3})$. While algorithms for fast constructions have been studied (e.g.  \cite{woodruff2010additive,knudsen2014additive,knudsen2017additive}), no near-optimal algorithm for the construction of any of the above additive spanners is known. For example, the fastest algorithm for constructing a $O(1)$-additive spanner with $O(n^{4/3})$ edges is $O(n^2)$. Further, Abboud and Bodwin \cite{abboud20174} proved that any purely additive spanner of sparsity $\tilde{O}(n^{4/3 - \epsilon})$, for any constant $\epsilon > 0$, has at least polynomial in $n$ additive error. 
Constructions by Bodwin and Vassilevska Williams \cite{bodwin2015very, bodwin2016better} are known giving 
sparsity $\tilde{O}(n)$ and 
additive error $\tilde{O}(n^{3/7+\epsilon})$. Following \cite{abboud20174}, Huang and Pettie \cite{HuangP18} constructed a family of graphs such that any $\tilde{O}(n)$-sized spanner for an $n$-node graph in the family must have $\Omega(n^{1/13})$ additive stretch.

Mixed-error $(\alpha, \beta)$-spanners were studied in \cite{elkin20041, baswana2005new, DBLP:conf/soda/ThorupZ06, pettie2009low, elkin2018efficient, ben2020new}. Most of these results focus on the setting of near-additive spanners, that is $(\alpha, \beta)$-spanners where $\alpha = 1 +\epsilon$ for some arbitrarily small constant $\epsilon > 0$. The goal of this setting is to obtain extremely sparse spanners $\tilde{O}(n)$. The best results obtain $(1+\epsilon, n^{o(1)})$-spanners with $\tilde{O}(n)$ edges. Abboud et al. \cite{abboud2018hierarchy} further developed a fine-grained hierarchy to give lower bounds for trade-offs between $\epsilon$, additive error and the sparsity of emulators. These lower bounds also apply to the setting of $(\alpha, \beta)$-emulators.

Finally, we point out that a related notion to emulators are hopsets:  given a graph $G$, we say that a graph $H$ is a $(\alpha, \beta, h)$-hopset if for every two vertices $x,y \in V$, there is a path $\pi_{x,y}$ from $x$ to $y$ in the graph $G \cup H$ consisting of at most $h$ edges, such that $\dist_G(x,y) \leq w(\pi_{x,y}) \leq (1+\epsilon) \dist_G(x,y) + \beta$ where $w(\pi_{x,y})$ denotes the weight of the path $\pi_{x,y}$. There is a lot of recent work on hopsets, especially $(1+\epsilon, \beta, h)$-hopsets, for which there are efficient algorithms  \cite{DBLP:conf/spaa/ElkinN19, huang2019thorup, elkin2019hopsets} with $\beta = 0$ and $h = n^{o(1)}$. The hopset literature builds heavily on previous clustering techniques from near-additive spanners. An excellent survey that highlights this connection was recently given by Elkin and Neiman \cite{elkin2020near}. 

\paragraph{Spanners and Emulators in Dynamic Graphs.} Spanners have also been extensively studied in the dynamic graph setting, where near-optimal algorithms for multiplicative spanners in fully dynamic graphs exist \cite{ausiello2005small, elkin2011streaming, baswana2008fully, baswana2006dynamic, BodwinK16, BernsteinFH19,forster2019dynamic, bernstein2020fully}. For hopsets, the dynamic graph literature has been mainly concerned with maintaining $(1+\epsilon, n^{o(1)}, n^{o(1)})$-hopsets in partially dynamic graphs \cite{henzinger2014decremental, bernstein2016deterministic, gutenberg2020deterministic} where they were used to derive fast algorithms for the partially dynamic Single-Source Shortest Paths problem.
As was observed in~\cite{HenzingerKN14} (Lemma 4.2) combining~\cite{DBLP:journals/siamcomp/RodittyZ12}
with~\cite{DBLP:conf/soda/ThorupZ06} leads to a $(1 + \epsilon, 2(1 + 2/\epsilon)^{k-2})$-approximate decremental emulator with total time $O((1 + 2/\epsilon)^{k-2} m n^{1/k})$.
To our knowledge, additive and near-additive spanners have not been studied in the dynamic graph literature. Also, there are no known conditional lower bounds for dynamic algorithms for maintaining a spanner.

\paragraph{Fully Dynamic Shortest Paths with Worst-Case Update Time.} Closely related to maintaining a spanner/emulator is the problem of maintaining shortest paths. There are three problems of focal interest:
%\begin{itemize}
 %   \item 
 
 (1) The $s$-$t$ Shortest Path ($st$-SP) problem asks for the shortest path between two fixed vertices $s,t \in V$.
 
 (2) The Single-Source Shortest Paths (SSSP) problem asks for the shortest path tree from a fixed vertex.
%$s \in V$

(3) The All-Pairs Shortest Paths (APSP) problem asks for the shortest path between every vertex pair.% in the graph.
%\end{itemize}

For each of these three problems, there is the \emph{distance reporting} and the \emph{path reporting} variant, where the former requires to only return the length of the shortest path, while the latter needs to return the actual shortest path.
There is an enormous line of research on these three problems in various settings. Since our fully dynamic algorithms have worst-case guarantees on update time, we focus this discussion on prior work on fully dynamic algorithms with worst-case update time.

For the $st$-SP problem and the SSSP problem the lower bounds in \cite{abboud2014popular,henzinger2015unifying} suggest that the essentially best solution to these problems is to rerun Dijkstra's algorithm after every update (even when the updates are not required to be worst-case). However, these conditional lower bounds are based on the BMM conjecture and therefore hold only for ``combinatorial'' algorithms. Indeed, Sankowski \cite{sankowski2005subquadratic} has shown that a worst-case update time of $O(n^{1.932})$ and query time of $O(n^{1.288})$ to obtain the \emph{distance} between a pair of vertices is possible and therefore has given the first subquadratic algorithm for the distance-reporting version of the $st$-SP problem. Recently, this result was further improved to worst-case update time $\tilde{O}(n^{1.863})$ and query time $\tilde{O}(n^{0.45})$ in \cite{BrandN19} where distance reporting queries are only required to return a $(1+\epsilon)$-approximate distance estimate. Rebalancing their trade-off terms, the authors also obtain an algorithm that maintains $(1+\epsilon)$-approximate SSSP with worst-case update time $\tilde{O}(n^{1.823})$ and $(1+\epsilon)$-approximate APSP in worst-case update time $\tilde{O}(n^2)$. 

A \emph{major drawback} of both approaches is that they cannot answer path reporting queries. The algorithm with fastest worst-case update time that can return \emph{actual} (approximate) shortest-paths for $st$-SP and SSSP remains to rerun Dijkstra's algorithm and for the APSP problem to use a combinatorial data structure where the currently best worst-case update is $\tilde{O}(n^{2+2/3})$ for weighted graphs and $\tilde{O}(n^{2.5})$ time for unweighted graphs (see \cite{thorup2005worst, abraham2017fully, gutenberg2020fully}). 

For approximate distance oracles with amortized running time, Abraham et al.~\cite{AbrahamCT14} achieved an $2^{O(k)}$-approximation with $O(\sqrt{m}n^{1/k})$ amortized update time  for constant $k$, and
 Forster et al.~\cite{Forster2021} gave an $(O(\log n))^{3k-2}$-approximation with $O(k \log^2 n)$ query time and $m^{1/k + o(1)}(O(\log n))^{4i-3}$ update time for any $k \geq 2$, being the first to break the $\sqrt m$ update time barrier.
 
\paragraph{Partially Dynamic Shortest Paths.}
% Fine-grained complexity has provided a variety of conditional lower bounds for dynamic graph problems such as dynamic shortest paths in a variety of settings (see e.g. \cite{abboud2014popular,henzinger2015unifying}). There are no known conditional lower bounds for dynamic algorithms for maintaining a spanner, however.

The classic ES-tree data structure \cite{shiloach1981line} initiated the field with a deterministic total time $O(mn)$ algorithm for partially dynamic exact SSSP. In the setting where a $(1+\epsilon)$-multiplicative approximation is allowed, Bernstein and Roditty \cite{bernstein2011improved} gave the first improvement over the ES-tree for an approximation algorithm with an algorithm for decremental $(1+\epsilon)$-approximate SSSP with total time $n^2 2^{O(\sqrt{\log (n)})}$. Subsequently, Henzinger et al. gave an algorithm \cite{henzinger2014decremental} with total update time $m^{1+o(1)}$. These algorithms are all randomized and against an oblivious adversary. 

Bernstein and Chechik gave the first deterministic partially dynamic $(1+\epsilon)$-approximate algorithm that improves upon the ES-tree data structure; it runs in total time $\tilde{O}(n^2)$ \cite{bernstein2016deterministic} and does not report paths, only distances. Chuzhoy and Khanna \cite{Chuzhoy:2019:NAD:3313276.3316320} gave an algorithm with total time $n^{2+o(1)}$ that works against an adaptive adversary and returns paths with $n^{1+o(1)}$ query time. Chuzhoy and Saranurak \cite{chuzhoy2020deterministic} recently further improved the running time of the path query to $|P|n^{o(1)}$ for an approximate shortest path $P$. Further, Bernstein and Chechik recently gave an algorithm with total update time $\tilde{O}(mn^{3/4})$ \cite{bernstein2017deterministic}, which was then improved to $O(mn^{0.5+o(1)})$ by Probst Gutenberg and Wulff-Nilsen \cite{gutenberg2020deterministic}. Neither of these data structures can answer path queries which was recently addressed in \cite{bernstein2020fully}.

For decremental APSP, Henzinger et al.~\cite{henzinger2014decremental} presented an approximation algorithm with stretch $((2+\epsilon)^k - 1)$ and total update time $m^{1+1/k+o(1)}$ for any positive integer $k$. They also gave an algorithm with stretch $(2+\epsilon)$ or $(1+\epsilon, 2)$ with total update time $\tilde{O}(n^{2.5})$ in \cite{henzinger2016dynamic} which was recently derandomized by Chuzhoy and Saranurak  \cite{chuzhoy2020deterministic}. Finally Henzinger et al.~\cite{henzinger2016dynamic}  presented a $(1+\epsilon)$-approximate deterministic algorithm with $\tilde{O}(mn/\epsilon)$ update time which derandomized the construction by Roditty and Zwick \cite{DBLP:journals/siamcomp/RodittyZ12} with matching running time. Later on, Chechik  \cite{chechik2018near} presented a $(2+\epsilon)k-1$-approximate algorithm with update time $mn^{1/k + o(1)}$ for any positive integer $k$ and constant $\epsilon$, whose total update time matches the preprocessing time of static distance oracles \cite{thorup2005approximate} with corresponding stretch.
Recently, Chen et al.~\cite{Chen20} gave an incremental $(2k-1)$-approximate algorithm with $O(m^{1/2}n^{1/k})$ worst-case time per operation.
We point out that there is an extensive line of work on the decremental APSP problem \cite{king1999fully, baswana2002improved, demetrescu2004new, RodittyZ11,thorup2005worst, bernstein2011improved,  DBLP:journals/siamcomp/RodittyZ12, abraham2013dynamic, henzinger2014decremental, henzinger2016dynamic,HenzingerKN17,abraham2017fully, chechik2018near, evald2020decremental} that is beyond the scope of this overview.

From the lower bounds side, Roditty and Zwick \cite{RodittyZ11} showed that any incremental or decremental algorithm for SSSP in weighted graphs with preprocessing time $p(n)$, query time $q(n)$ and update time $u(n)$ must satisfy $p(n)+n\cdot u(n)+n^2 \cdot q(n)\geq n^{3-o(1)}$ unless APSP has a truly subcubic time algorithm. Similarly, for unweighted graphs, they showed that any {\em combinatorial} incremental or decremental algorithm must satisfy that equation unless Boolean matrix multiplication (BMM) has a truly subcubic time combinatorial algorithm. Abboud and Vassilevska Williams \cite{abboud2014popular} extended these lower bounds to also hold for $st$-SP, where now the algorithms must satisfy $p(n)+n\cdot (u(n)+q(n))\geq n^{3-o(1)}$, for weighted graphs under the APSP conjecture, and for unweighted graphs under the combinatorial BMM conjecture.  %That is, either the preprocessing time is $n^{3-o(1)}$, or the update or query time is $n^{2-o(1)}$, which is essentially the running time for recomputation from scratch. 

\paragraph{Hypotheses for Fine-Grained Complexity} Our conditional lower bounds rely on two popular hypotheses: the Online Boolean Matrix-Vector Multiplication (OMv) conjecture and the Combinatorial $k$-Clique hypothesis. In the OMv problem we are given an $n\times n$ matrix $M$ that can be preprocessed. Then, an online sequence of vectors $v^1,\dots ,v^n$ is presented and the goal is to compute each $Mv^i$ before seeing the next vector $v^{i+1}$. The OMv conjecture was first defined in~\cite{henzinger2015unifying}, and has been used many times since.

\begin{conjecture}[OMv]\label{conj:omv}
For any constant $\epsilon > 0$, there is no $O(n^{3-\epsilon})$-time
algorithm that solves OMv with error probability at most $1/3$ in the word-RAM model with $O(\log n)$ bit words.
\end{conjecture}

The Combinatorial $k$-Clique hypothesis is defined as follows and has been used a number of times (e.g.~\cite{lincoln2018tight,abboud2018if,bringmann2019fine}).

\begin{hypothesis}[Combinatorial $k$-Clique]
For any constant $\epsilon>0$, for an $n$-node graph there is no $O(n^{k-\epsilon})$ time combinatorial algorithm for $k$-clique detection with error probability at most $1/3$ in the word-RAM model with $O(\log n)$ bit words.
\end{hypothesis}

For the special case of $k$-Clique detection where $k=3$ (i.e. triangle detection), we also consider non-combinatorial algorithms. Triangle detection can easily be solved using matrix multiplication but it is a big open question whether triangle detection admits a $O(n^{\omega-\epsilon})$ time algorithm, where $\omega<2.373$ is the matrix multiplication exponent (e.g.~\cite{williams2010subcubic},~\cite{woeginger2008open} Open Problem 4.3(c), and~\cite{spinrad2003efficient} Open Problem 8.1). It is generally believed that such an algorithm does not exist, and our reductions from $k$-Clique also imply hardness under this hypothesis.

%***patrascu,abboud-vw,omv,jan+danupon

\paragraph{Our results.}
We present novel algorithms and conditional lower bounds for $(\alpha,\beta)$-spanners and emulators as well as faster fully dynamic APSP algorithms. We prove the following for undirected unweighted graphs.

%\begin{enumerate}
%\item
\medskip
\noindent 1. {\bf \emph{Conditional lower bounds for partially dynamic spanners/emulators.}} Under the OMv conjecture, there can be no decremental or incremental algorithm that maintains a  $(1,n^{o(1)})$-emulator (and thus spanner) with $O(m^{1-\epsilon})$ edges for any constant $\epsilon>0$ with arbitrary polynomial preprocessing time and total update time $O(mn^{1-\epsilon})$. The same result also holds for \emph{all} sparsities $m$ for combinatorial algorithms under the Combinatorial $k$-Clique hypothesis.
%Since under the Girth conjecture, it is folklore that there can be no $(O(1), 0)$-spanner or emulator with $n^{1+o(1)}$ edges, a non-zero additive error is necessary for any spanner with constant multiplicative error. 

For completeness, we also present algorithms that are tight with our conditional lower bounds. Note that mixed additive/multiplicative error is necessary for these algorithms since there can be no $(O(1), 0)$-spanner or emulator with $n^{1+o(1)}$ edges (e.g.~\cite{margulis1982explicit}).
%We present the first mixed-error partially dynamic algorithm for this problem. 
Our algorithms rely heavily on prior work. Using techniques similar to~\cite{chechik2018near}, for any constant $\epsilon\in (0,1]$, we maintain whp against an oblivious adversary a partially dynamic $n^{1+o(1)}$ edge $(1+\epsilon,n^{o(1)})$-spanner in
total update time $m^{1+o(1)}$ time. %Note that our lower bounds show that 
%our algorithm is conditionally tight. 
Using techniques from~\cite{gutenberg2020deterministic}, we also give a deterministic partially dynamic algorithm that maintains a $(1+\epsilon, n^{\alpha + o(1)})$-emulator of a graph in total time $O(m n^{1-\alpha + o(1)})$ for any $\alpha > 0$. Using a result in \cite{bernstein2020fully}, we can further turn the above algorithm into a randomized algorithm that maintains a $(1+\epsilon, n^{\alpha + o(1)})$-\emph{spanner} in expected total time $O(m n^{1-\alpha + o(1)})$ for any $\alpha > 0$ and that works against an adaptive adversary.

\medskip
\noindent 2. {\bf\emph{Conditional lower bounds for combinatorial fully dynamic spanners.}} Under the Combinatorial $k$-Clique hypothesis, for a graph of any sparsity $m$, for any constant $\epsilon>0$, there can be no fully dynamic combinatorial algorithm  that maintains an $O(m^{1-\epsilon})$-edge $(1+\alpha,n^{o(1)})$-emulator for small $\alpha$ with preprocessing time $mn^{1-\epsilon}$ and amortized update time $m^{1-\epsilon}$. This conditional lower bound also extends to incremental and decremental algorithms but only for worst-case update times. 

For completeness, we also present an algorithm that is tight with our conditional lower bound. This algorithm follows from rerunning a known static algorithm after every update. For any constant $\epsilon\in (0,1]$, we give  a deterministic fully dynamic algorithm with preprocessing time $m^{1+o(1)}$ and worst-case update time $m^{1+o(1)}$ time that maintains an $n^{1+o(1)}$ edge $(1+\epsilon,n^{o(1)})$-spanner.

%\item 
\medskip
\noindent 3. {\bf\emph{Algebraic fully dynamic spanner algorithms.}} The above fully dynamic lower bound only applies to combinatorial algorithms, and we show that this is inherent; we develop an algebraic spanner algorithm that beats our combinatorial lower bound. For any constant $\epsilon\in (0,1]$, there is a fully dynamic algorithm 
with preprocessing time $\tilde{O}(n^2)$ in an  initially empty graph and $O(n^{2.373})$ in an initially non-empty graph and
worst-case update time $O(n^{1.529})$ 
that whp maintains an $n^{1+o(1)}$-edge $(1+\epsilon,n^{o(1)})$-spanner and works against an oblivious adversary.~\footnote{Both the bound on the time per update as well as the correctness hold with high probability. If we rebuild the data structure from scratch every polynomially many updates, we can instead achieve an {\em expected  amortized} time bound of $O(n^{1.529})$ per update.}

The construction from our above lower bound from the Combinatorial $k$-Clique hypothesis with $k=3$ (i.e. triangle detection) also gives a conditional lower bound for non-combinatorial algorithms. Unless there is a breakthrough in non-combinatorial algorithms for triangle detection algorithms, there can be no fully dynamic algorithm that maintains an $O(n^{\omega-1-\epsilon})$-edge $(1+\alpha,n^{o(1)})$-emulator for constant $\alpha<2/3$ with preprocessing time $O(n^{\omega-\epsilon})$ and amortized update time $O(n^{\omega-1-\epsilon})$, where $\omega<2.373$ is the matrix multiplication exponent. Thus $O(n^{1.372})$ update time and $O(n^{2.372})$ preprocessing time is not possible with current techniques.
%The preprocessing time is the best known running time for multiplying two $n\times n$ matrices, which is the best known running time for Triangle Detection in an $n$ node graph.

We also give a conditional lower bound from the OMv conjecture that precludes algorithms for emulators with more edges and higher preprocessing time than the above lower bound from triangle detection, but at the cost of a lower update time. Under the OMv conjecture, for any constant $\epsilon>0$, there can be no fully dynamic algorithm that maintains an $O(m^{1-\epsilon})$-edge $(1+\alpha,n^{o(1)})$-emulator for constant $\alpha<2/3$ with arbitrary polynomial preprocessing time  and amortized update time $O(n^{1-\epsilon})$.

Both of these conditional lower bound also extend to incremental and decremental algorithms but only for worst-case update times.

%(4) \emph{Decremental deterministic emulator algorithm.}
%\item

\medskip
    \noindent 4. {\bf \emph{Fully dynamic exact path-reporting APSP.}} To achieve the above results we develop the first  fully dynamic APSP data structure that supports distance queries, path reporting queries, and edge updates in {\em subquadratic} time per operation. It uses algebraic techniques, is randomized, and works against an oblivious adversary. Specifically we show the following result, where $\omega(a,b,c)$ is the exponent for multiplying an $n^a\times n^b$ matrix by an $n^b\times n^c$, and $\eps_*$ is the solution to $\omega(1,1,\eps) = 1 + 2 \eps$. With the current bounds for rectangular matrix multiplication, $\eps_* \approx 0.529.$

%\medskip\noindent
\begin{theorem}\label{thm:intro}
Let $\eps$ be such that $0<\eps\leq \eps_{*}$, and let $D$ be a distance parameter between $1$ and $n$. There is a randomized fully dynamic data structure that can maintain an unweighted directed graph $G=(V,E)$ supporting the following operations with preprocessing time $\tilde{O}(n^2)$ in an empty initial graph and $\tilde{O}(Dn^\omega)$ in an non-empty initial graph:
(a) {\bf edge updates}  in worst-case time $\tilde{O}(Dn^{\omega(1,1, \eps)-\eps})$ time; 
(b) {\bf [distance reporting]:} on query $i,j\in V$, return $\dist(i,j)$ if $\dist(i,j)\leq D$, or answer that $\dist(i,j)>D$ otherwise, in worst-case $\tilde{O}(Dn^{\eps})$ time, where the answer is correct whp;
(c) {\bf [path reporting]:} on query $i,j\in V$, if $\dist(i,j)\leq D$, return a shortest path from $i$ to $j$, in $\tilde{O}(D^2n^\kappa)$ time, where the answer is correct whp.

\end{theorem}
%\medskip

We believe that this result is of independent interest.

Based on it we build a fully dynamic exact APSP data structure that 
with preprocessing time $\tilde O(n^2)$ on an initially empty graph
achieves worst-case time  $O(n^{1.9})$ per edge update, $O(n^{1.529})$ per distance query and $O(n^{1.9})$ per path reporting query.
This is a significant improvement over the $O(n^{2.5})$ worst-case update time of~\cite{abraham2017fully,gutenberg2020fully} and closer to the $O(n^2)$ time bound which is achieved by the data structure of~\cite{demetrescu2004new} which can only support distance reporting queries, but no path reporting queries.

%***RANDOMIZATION***
The algorithms in Theorem~\ref{thm:intro} are all Monte Carlo-- they are correct with high probability and always run in the desired running time. If they could be made into Las Vegas algorithms (ones that are always correct but have expected running time), our applications of Theorem~\ref{thm:intro}, such as our algebraic spanners, would also be Las Vegas, which is a more desirable guarantee.% for a data structure problem. 
However, there are significant hurdles to overcome in order to make Theorem~\ref{thm:intro} Las Vegas. Like Sankowski's original data structure \cite{sankowski-thesis}, Theorem~\ref{thm:intro} heavily relies on the use of polynomial identity testing (PIT), namely on the fact that PIT is in co-RP and hence has a fast Monte Carlo algorithm. To obtain a Las Vegas algorithm using a similar approach, one would need a ZPP algorithm for PIT. 
However, obtaining such an algorithm seems extremely difficult, and in fact Impagliazzo and Kabanets
\cite{KabanetsI03,KabanetsI04}
 showed that such an algorithm would imply strong circuit lower bounds. Thus although the rest of our techniques can be made Las Vegas, making Theorem~\ref{thm:intro} Las Vegas as well is far from possible with current techniques.

%\item 
\medskip
\noindent 5. {\bf \emph{Applications.}} We present two applications of our above results: fully dynamic approximate path-reporting APSP, and fully dynamic Steiner tree. Using the above theorem and the above algebraic spanner, we give the first subquadratic 
fully dynamic $(1+\epsilon)$-approximate APSP algorithm. It needs $\tilde O(n^2)$ preprocessing time on an empty graph and  achieves worst-case time $n^{1+\eps^*+o(1)}=O(n^{1.529})$ for updates, $n^{1+o(1)}$ 
for approximate distance reporting and approximate shortest path reporting,
 whp against an oblivious adversary.
Note that all previous subquadratic update/query algorithms could only report distances, not paths.

%\item
%\medskip
%\noindent 6. {\bf \emph{Fully dynamic Steiner tree.}}
Our second application of our above results is a fully dynamic algorithm for $(2+\epsilon)$-approximate Steiner tree, which can be used, for example, for routing in dynamic networks. Specifically we give the first subquadratic algorithm that maintains a $(2+\epsilon)$-approximate Steiner tree for a set $S$ of terminals  with both terminal and edge updates. Specifically, it has preprocessing time $\tilde{O}(n^2)$ on an empty initial graph and $n^{\omega+o(1)}$ on a non-empty initial graph, and worst-case time $n^{1+\eps^*+o(1)}
  + |S|^2 \cdot n^{1 + o(1)}$ per edge update, $|S| n^{1+o(1)}$ per node addition to $S$, and
  $|S|n^{o(1)}$ per node removal from $S$, giving subquadratic update time when $|S|\leq n^{1/2-o(1)}$ whp against an oblivious adversary. 
  By increasing the processing time to $O(n^{2.621})$ using the data-structure of \cite{BrandN19}, the time for edge updates can be made $O(n^{1.843}+ |S|^2 \cdot n^{0.45} + |S| \cdot n^{1 + o(1)})$, allowing for more leverage over the size of the terminal set $S$.
 The only prior work in general graphs maintains a $(6+\epsilon)$-approximate Steiner Tree under changes to $S$ only (no edge updates) and has preprocessing time $\tilde O(m \sqrt n )$ and update time $\tilde O(\sqrt{n})$~\cite{LackiOPSZ15}. 
%  \end{enumerate}

\paragraph{Organization} In Section~\ref{sec:over} we give a technical overview of a selection of our results. Section~\ref{sec:pre} is the preliminaries. In Section~\ref{sec:lb}, we present our conditional lower bounds.  In Section~\ref{sec:algebraic}, we present our data structure for dynamic APSP with path reporting. In Section~\ref{subsec:AlgebraicSpanner}, we present our algebraic spanner algorithm, which uses the data structure from Section~\ref{sec:algebraic}. In Section~\ref{sec:app} we present two additional applications of the data structure from Section~\ref{sec:algebraic}: our dynamic algorithm for approximate APSP with path reporting and our dynamic Steiner tree algorithm. Finally, in Section~\ref{sec:span}, we present our combinatorial dynamic spanner and emulator algorithms.

\section{Technical overview}\label{sec:over}

\paragraph{Conditional lower bounds.} We first outline our OMv-based constructions. Instead of reducing from the OMv problem, we reduce from the related OuMv problem, which is defined as follows. We are given an $n\times n$ matrix $M$ that can be preprocessed. Then, an online sequence of vector pairs $(u^1,v^1),\dots,(u^n,v^n)$ is presented and the goal is to compute each $(u^i)^\intercal Mv^i$ before seeing the next pair. A reduction from OMv to OuMv is known~\cite{henzinger2015unifying}.

For both our fully dynamic and incremental/decremental lower bounds from OMv we begin with the same basic gadget.  Given the matrix $M$ from the OuMv instance, we construct a bipartite graph $A,B$ where $A=\{a_1,\dots a_n\}$, $B=\{b_1,\dots b_n\}$, and the edge $(a_i,b_j)$ is present if and only if $M_{i,j}=1$. 

The fully dynamic construction is shown in Figure~\ref{fig:omv_fully_intro}. We begin by taking a number $c$ of disjoint copies $G_1,\dots,G_c$ of the basic gadget and an additional set of $c+1$ isolated vertices $w_0,\dots w_c$. Each basic gadget will introduce error to the approximation, so larger $c$ means that we are showing a lower bound for algorithms with higher approximation factors but faster running times. 

After constructing this initial graph, we start $n$ dynamic phases. In phase $i$, we are given the vectors $u^i$ and $v^i$ of the OuMv instance. For each $1\leq j\leq c$ and each $k$ with $u^i_k=1$, insert an edge between $w_{j-1}$ and $a_k\in G_j$. Similarly, for each $1\leq j\leq c$ and each $k$ with $v^i_k=1$, insert an edge between $w_{j}$ and $b_k\in G_j$. We remove these edges after the phase is over.

\begin{figure}[ht]
  \centering
  \includegraphics[width=.8\linewidth]{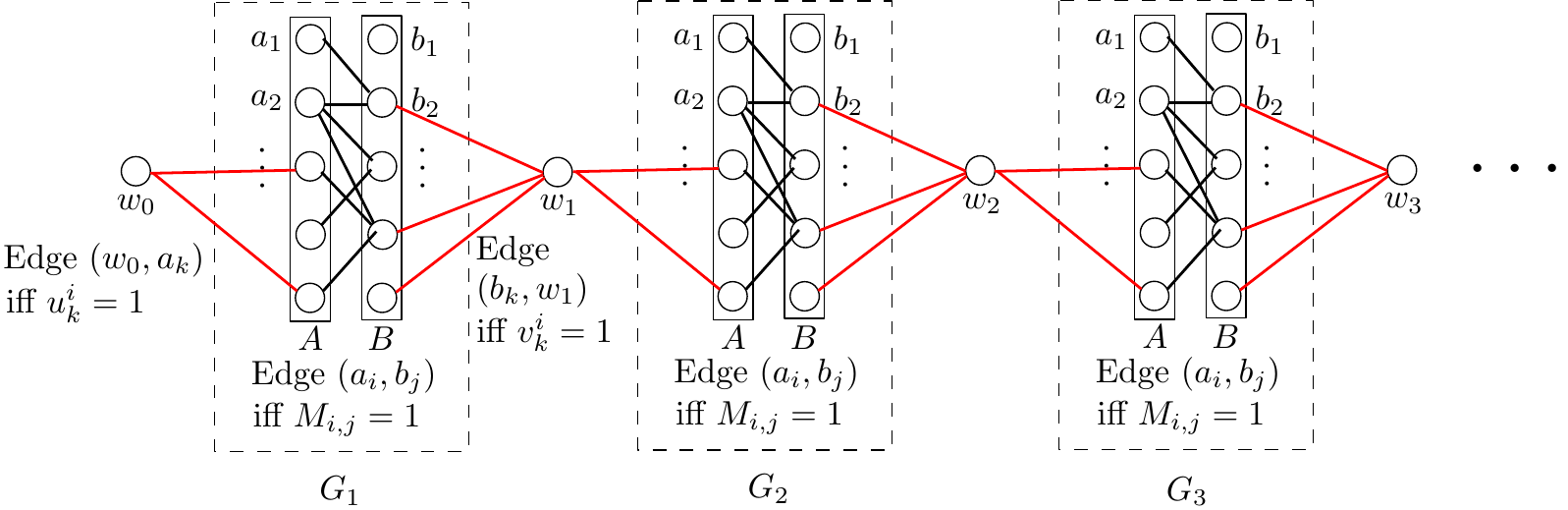}
  \caption{The construction for fully dynamic algorithms. The red edges are dynamically added in phase $i$.}
  \label{fig:omv_fully_intro}
\end{figure}

 At the end of each phase, we run Breadth-First Search (BFS) on the dynamic emulator to estimate the distance between $w_0$ and $w_c$, which we claim provides the answer to this phase of the OuMv instance. In particular, note that for any $i$ the distance between $w_i$ and $w_{i+1}$  is 3 if and only if $(u^i)^\intercal Mv^i=1$, and otherwise this distance is at least 5.
 Also, since the emulator has $O(n^{2-\epsilon})$ edges, the resulting algorithm would solve OuMv in $O(n^{3-\epsilon})$ time.
 
The incremental construction is similar, however we cannot remove  edges at the end of each phase. To get around this, we replace each $w_i$ with a path and insert edges incident to a different vertex in the path at each phase. The resulting construction is shown in Figure~\ref{fig:omv_inc_intro}.

\begin{figure}[ht]
  \centering
  \includegraphics[width=\linewidth]{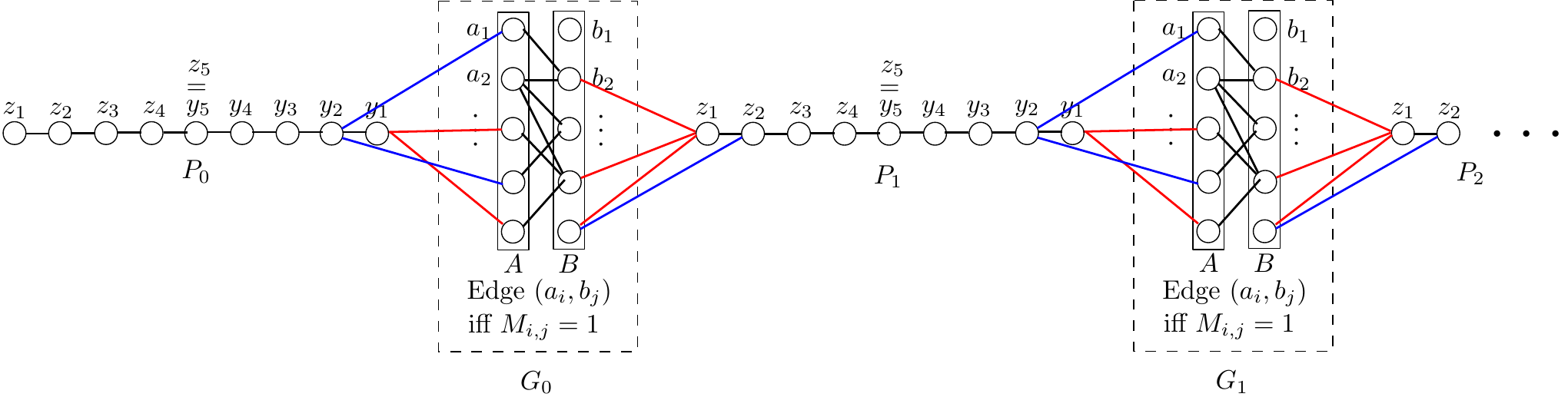}
  \caption{The construction for incremental algorithms. The red edges are dynamically added in phase 1 and the blue edges are dynamically added in phase 2.}
  \label{fig:omv_inc_intro}
\end{figure}

The decremental construction is roughly the reverse of the incremental construction.

For the $k$-clique-based constructions, we instead reduce from the $k$-cycle problem (a reduction from $k$-clique to $k$-cycle is known~\cite{lincoln2018tight}). The $k$-cycle constructions for both the fully and partially dynamic settings follow a similar structure to the OuMv constructions but use a different basic gadget. The basic gadget is built by using color coding and taking a layered version of the graph where each color is a layer and only edges between vertices of adjacent colors are present.

\paragraph{Algebraic Fully Dynamic Spanner Algorithm.} Let us next present an algorithm to maintain a $(1+\epsilon,n^{o(1)})$-spanner on a fully dynamic graph with worst-case update time $O(n^{1.529})$. 

Let $k = \sqrt{\log n}$. We sample sets $V = A_0 \supseteq A_1 , \dots \supseteq A_k \supseteq A_{k+1} = \emptyset$ where $A_i$ for $i \in [1,k + 1]$ is obtained by sampling each vertex in $V$ with probability $n^{-i/k} \log n$ (and to make the sets nesting add it to all $A_j$ where $j \leq i$, we assume that $A_{k+1}$ is empty which can be achieved by resampling a constant number of times in expectation). Given these sets, we say that each $a \in A_{\ell} \setminus A_{\ell+1}$  is \emph{active} if for no $j > \ell$ there exists a vertex $a' \in A_j \setminus A_{j+1}$ with  $\dist_{\tilde{G}}(a,a') \lesssim \left(\frac{1}{\epsilon}\right)^{j}$. Using this definition of activeness, we can show by a simple hitting set argument that each active vertex $a \in A_{\ell} \setminus A_{\ell+1}$ has in its ball to radius $\sim \left(\frac{1}{\epsilon}\right)^{\ell+1}$ at most $\tilde{O}(n^{(\ell+1)/k})$ vertices w.h.p..

Given this set-up, a natural way to construct a spanner $H$, would be to find for each level $\ell \in [0,k]$, the active vertices in $A_{\ell} \setminus A_{\ell+1}$ and to include their shortest path trees truncated at radius $\sim \left(\frac{1}{\epsilon}\right)^{\ell+1}$. For the number of edges of $H$, it is not hard to see that each for each $\ell \in [0,k)$, there are at most $\tilde{O}(n^{1-\ell/k})$ vertices that are active in $A_{\ell} \setminus A_{\ell+1}$. Each of these vertices contributes a single edge for each vertex in its truncated ball (except for itself), and as discussed above we have that each ball is of size at most $\tilde{O}(n^{(\ell+1)/k})$. Thus, we have that $H$ has at most $\tilde{O}(n^{1+1/k}) = n^{1+ o(1)}$ edges.

For the stretch factor, observe that for any vertices $s,t \in V$, with shortest path $\pi_{s,t}$, we have that for $s \in A_{\ell} \setminus A_{\ell+1}$ for some level $\ell$, if $s$ is active, we can simply travel along the $\pi_{s,t}$ to some vertex $s'$ that is closer to $t$ (since the truncated shortest path tree of $s$ is included in $H$) and then expose the shortest path $\pi_{s', t}$ inductively. Or, we have that there is a vertex $a' \in A_{j} \setminus A_{j+1}$ at distance $\lesssim \left(\frac{1}{\epsilon}\right)^{j}$ to vertex $s$ (with $j > \ell$. Choosing $a'$ to be the vertex that is at this distance to $a$ with the largest possible $j$, we will be able to argue that $a'$ is active. Thus, we can travel from $s$ to $a'$ to a vertex $s'$ on $\pi_{s,t}$ at distance roughly $\left(\frac{1}{\epsilon}\right)^{j+1}$ along the shortest path tree at $a'$ truncated at depth $ \sim \left(\frac{1}{\epsilon}\right)^{j+1}$ that was included in $H$. It is not hard to see that the error induced for visiting $a'$ can be subsumed in a multiplicative $(1+O(\epsilon))$-approximation. However, this only works if $s$ and $t$ are at distance $ \gtrsim \left(\frac{1}{\epsilon}\right)^{j+1}$, otherwise it induces an additive error of $n^{o(1)}$. This explains why we obtain a $(1+\epsilon, n^{o(1)})$-approximation.

Unfortunately, while this set-up of $H$ is sensible, consider the example of the complete graph. Then, there would be a vertex $a \in A_{k} \setminus A_{k+1}$ (which is active since $A_{k+1}$ is empty) where visiting the truncated ball at $a$ would take time $\tilde{O}(n^2)$, which is by far too expensive for our algorithm. 

Too overcome this issue, instead of inserting truncated shortest path trees to $H$, we only insert for any active vertex $a \in A_{\ell} \setminus A_{\ell+1}$, the shortest paths to other vertices in $A_{\ell} \setminus A_{\ell+1}$ in the truncated shortest path tree of $a$. We then fix a threshold $\gamma \approx \lfloor 0.529 \cdot k\rfloor$, and can use the algebraic data structure from \Cref{thm:intro} to maintain the distances of vertices in  $A_{\gamma}$ (and thereby $A_{\gamma + 1}, A_{\gamma +2}, \dots, A_k, A_{k+1}$) without explicitly maintaining the balls of the active vertices. For active vertices in some set $A_{\ell} \setminus A_{\ell+1}$ for $\ell < \gamma$, we can compute the balls explicitly after every update. This is because each such ball only contains $\tilde{O}(n^{(\ell+1)/k})$ vertices, and therefore the induced graph can contain at most $\tilde{O}(n^{2(\ell+1)/k})$ edges, which implies that we overall, spend at most time $\tilde{O}(n^{1-\ell/k}) \cdot n^{2(\ell+1)/k} = \tilde{O}(n^{1 + (\ell+2)/k})$ time for computing all such balls. We refer the reader to section \Cref{subsec:AlgebraicSpanner} for a proper analysis of the running time.

Finally, we point out that so far $H$ only contains shortest paths between active vertices in the same set $A_{\ell} \setminus A_{\ell+1}$ (if they are reasonably close). However, to have a path between vertices on different levels, we also add a $O(\log n)$ spanner $\tilde{G}$ of $G$ to $H$. Such a spanner is simple to maintain, for example \cite{forster2019dynamic} shows how to maintain such a spanner with $\tilde{O}(n)$ edges and $\tilde{O}(1)$ amortized update time. 

The idea of the approximation proof then becomes the following for some path $\pi_{s,t}$: Let $i$ be the largest index such that an active vertex $a \in A_i$ is at distance at most $\sim \epsilon^{-i}$ to $s$. Let $w$ be the farthest vertex from $s$ on $\pi_{s,t}$ such that (1) the distance from $s$ to $w$ is at most $\sim \epsilon^{-(i+1)}$, and (2) $w$ has distance at most $\sim \epsilon^{-i}$ to an active vertex $a' \in A_i$. Such a vertex $w$ exists since we could have $w = s$ and $a' = a$. It is then apparent that the distance between $a$ and $a'$ is $\lesssim \epsilon^{-(i+1)}$ and since $a'$ is active, we can ensure that the shortest path from $a$ to $a'$ is in $H$. Further, we can use the paths in the spanner $\tilde{G}$ (which belongs to $H$) to get from $s$ to $a$ and from $a'$ to $w$; since these two distances are small, it suffices to have an $O(\log n)$ multiplicative error for them. Now, observe that this induces additive error along the path segment from $s$ to $w$ of $\tilde{O}(\epsilon^{-i})$ (by the triangle inequality). We either have that $a$ and $a'$ are roughly at distance $\sim \epsilon^{-(i+1)}$ (which suffices to subsume the additive error in the multiplicative error), or we have that for the next path segment of length $\sim \epsilon^{-(i+1)}$, no vertex is close to any active vertex in $A_i$. Thus, when we repeat the whole argument for the  next path segment, we get that vertices on lower levels are active, which means that they induce less additive error. This allows us to subsume the additive error from higher levels into  multiplicative error for a series of segments of lower levels.

We refer the reader to \Cref{subsec:AlgebraicSpanner} for the full details of the algorithm.

\paragraph{Fully dynamic APSP with path reporting.} Our data-structure of theorem \ref{thm:intro} is an augmentation of Sankowski's \cite{sankowski-thesis} data structure to support fast successor queries. Essentially, Sankowski showed how to reduce the problem of maintaining the short distances in a dynamic unweighted graph to the dynamic matrix inverse problem, by representing the path lengths as degrees of the adjoint of a polynomial matrix. He then showed how to efficiently maintain the inverse of a matrix subject to entry updates, allowing for fast distance queries. We extend his techniques to maintain products of matrices and the inverse, and show how to use these products to extract successor information similarly to Seidel's path reporting algorithm for static APSP \cite{seidelapsp}. Let us begin here by sketching Sankowski's data-structure \cite{sankowski-thesis}, formally reviewed in Section \ref{subsec::sankreview}, and then present the high level of our augmentation. \\

\noindent\textit{Short Distances to Dynamic Matrix Inverse.} \cite{sankowski-thesis} showed how to encode path lengths of an unweighed graph in the adjoint of a matrix, that is, given a adjacency matrix $A_{ij}$ with a random integer entry if $(i, j)\in E$, then the lowest degree non-zero term of the polynomial adj$(\mathbb{I}-u A)_{ij}$ over the variable $u$ is the distance $d_{ij}$ whp (Lemma \ref{lemma::sankadjoint}). In this manner maintaining adj$(\mathbb{I}-uA)_{ij}$ mod $u^{D+1}$, for some distance parameter $D$, allows us to query a distance $i\rightarrow j$ correctly whp if $d_{ij}\leq D$. Note that adj$M$ = det$M\times M^{-1}$, s.t. it suffices just to maintain det $M$ and $M^{-1}$ mod $u^{D+1}$. \\

\noindent \textit{Dynamic Matrix Inverse.} We detail the algebraic tools developed by Sankowski \cite{sankowski-thesis} to maintain the inverse of a matrix $M$ dynamically and over a ring in Section \ref{subsec::sankreview}. The main idea is to maintain explicitly (i.e. all $n^2$ entries) two matrices $T, N$, s.t. we maintain the invariant
\begin{equation}
    M^{-1} = T(\mathbb{I}+N)
\end{equation}
where, initially, $T=M^{-1}$ and $N=0$, and as later shown each single entry update to $M$ corresponds to a single row update to $N$ (and no modifications to $T$!). After $m$ updates, $N$ has at most $m$ non-zero rows, and we can exploit this \textit{sparsity} of $N$ to quickly compute its row-updates in $O(mn)$, and every $m=n^\eps$ updates we reset $T\leftarrow T+TN$, $N\leftarrow 0$, in $O(n^{\omega(1, 1, \eps)-\eps})$ time on average. In this manner, we guarantee that $N$ is always sparse, and updates take amortized time 
\begin{equation}
    O(n^{\omega(1, 1, \eps)-\eps}+n^{1+\eps})
\end{equation}
for some parameter $\eps\in (0, 1)$ which we can later optimize over. Entry queries $(i, j)$ are now straightforward, as it suffices to compute the dot product 
\begin{equation}
    M^{-1}_{ij} = e_i^TT(\mathbb{I}+N)e_j = T_{ij} + (e_i^TT)\cdot (Ne_j)
\end{equation}
which can be done in $O(n^\eps)$ time since a given column of $N$ has at most $O(n^\eps)$ non-zero entries. In Corollary \ref{cor::polyminverse} \cite{sankowski-thesis} showed that we can maintain these matrices over a ring mod $u^{D+1}$ by introducing a multiplicative factor of $\tilde{O}(D)$ to the runtimes described above, s.t. now if $M=\mathbb{I}-uA$ we can query any distance $d\leq D$ in the graph in time $\tilde{O}(Dn^\eps)$.\\

\noindent\textit{Successor Queries to Product Maintenance.} There are two main ingredients to our augmentation the data-structure of \cite{sankowski-thesis}. The first is to reduce the successor query of a pair $(i, j)$ of an unknown number of distinct successors to that of a single successor, using a known sparsification trick used first in Seidel's algorithm for static, undirected and unweighted APSP \cite{seidelapsp}. This only introduces a $\log^2 n$ multiplicative factor to the runtime and we defer the formal argument to Lemmas \ref{lemma::singlesuccessor} and \ref{lemma::multiplesuccessor}. The second ingredient is to show how to find a single successor by finding a witness of the product  $(A\cdot $adj$(\mathbb{I}-uA))_{ij}$. The key new insight is that if the distance $1<d_{ij}\leq D$, then adj$(\mathbb{I}-uA)_{ij}$ has minimum degree $d_{ij}$, and thereby the product $(A\cdot $adj$(\mathbb{I}-uA))_{ij}$ must have minimum degree $d_{ij}-1$. This is since there must exist a \textit{unique} witness $s$ (the single successor!) s.t. $A_{is}$ is non-zero, corresponding to an edge, and adj$(\mathbb{I}-uA)_{sj}$ has minimum degree $d_{sj} = d_{ij}-1$, corresponding to the length of the shortest path from $s$ to $j$.

We can find this single witness by computing its bitwise description, that is, defining $O(\log n)$ versions of the adjacency matrix $A$, $A^{(l)}$ for $l\in [O(\log n)]$, where we null the $p$th column of $A^{(l)}$ if the $l$th bit of $p$ is $0$. As there is only a single witness $s$, the minimum degree of the product $(A^{(l)}\cdot $adj$(\mathbb{I}-uA))_{ij}$ is $d_{sj}=d_{ij}-1$ only if the column of $s$ is selected, that is, the $l$th bit $s_l=1$. In this manner, if we query the $O(\log n)$ products  $(A^{(l)}\cdot $adj$(\mathbb{I}-uA))_{ij}$, the 1-bits $s_l=1$ are exactly the products $l$ s.t. the minimum degree is correct. This allows us to extract the successor description in a polylog number of queries to products $(E\cdot $adj$(\mathbb{I}-uA))_{ij}$ for given matrices $E$. \\

\noindent \textit{Product Maintenance.} The last detail in our successor query augmentation is to show how to maintain products $(E\cdot $adj$(\mathbb{I}-uA))$, where we can modify entries of $E$ and $A$, and query entries $i,j$ of the result. Note again that it suffices to maintain $(E\cdot (\mathbb{I}-uA)^{-1})$, as opposed to the adjoint, just by multiplying by the determinant. We do so by following the inverse maintenance algorithm and explicitly maintaining the matrices $T, N$ and $V\equiv ET$, s.t. we maintain the invariant
\begin{equation}
    EM^{-1} = ET(\mathbb{I}+N) = V(\mathbb{I}+N)
\end{equation}
We address updates to $E$ and to $A$ completely differently. Updates to $A$ follow the original lazy construction, where we simply perform a row-update to $N$, and every $n^\eps$ updates we "reset" $V\leftarrow V(\mathbb{I}+N), T\leftarrow T(\mathbb{I}+N), N\leftarrow 0$. We note that correctness follows by associativity, s.t. although matrices $V$ and $T$ are dense we can still exploit the sparsity of $N$ to compute their updates independently in time $O(Dn^{1+\eps}+Dn^{\omega(1, 1, \eps)-\eps})$ on average. Entry-Updates to $E$, $E\leftarrow E+e_{ij}$, are much simpler. We once again use associativity to compute the row update $V\leftarrow V+e_{ij}T$ in $\tilde{O}(Dn)$ time. Finally, to query an entry of the product $(EM^{-1})_{ij}$, we follow analogously to \cite{sankowski-thesis} and compute the dot product

\begin{equation}
    (EM^{-1})_{ij} = e_i^TV(\mathbb{I}+N)e_j = V_{ij} + (e_i^TV)\cdot (Ne_j)
\end{equation}

\noindent and since $V$ and $N$ are maintained explicitly, this takes time $\tilde{O}(Dn^\eps)$, the exact same as the distance queries. Overall, this allows for $\tilde{O}(Dn^\eps)$ time successor queries, and by iterating, short path queries of length $\leq D$ in time $\tilde{O}(D^2n^\eps)$.

\section{Preliminaries}\label{sec:pre}

We let $G=(V,E)$ denote an undirected unweighted dynamic input graph, where $n = |V|$ and $m = |E|$. For any graph $H$, and two vertices $a,b \in V(H)$, we denote by $\mathbf{dist}_H(a,b)$ the distance between the two vertices in $G$ and let  $\pi_{a,b, H}$ denote a corresponding shortest path between $a$ and $b$. If the graph $H$, especially when we use the input graph $G$, is clear from the context, we simply use $\pi_{a,b}$. We define ${B}_H(s, r)$ in the graph $H$, to be the ball rooted at $s$ with radius $r$, i.e. the set of vertices ${B}_H(s, r) = \{ w \in V(H) | \dist_H(s,w) \leq r\}$. Throughout the article, we often use the data structure stated below that is sometimes referred to as the Even-Shiloach (ES) tree.

\begin{lemma} [c.f.  \cite{shiloach1981line}]
\label{lma:maintainBalls}
For any vertex $s \in V$, radius $r$, there is a deterministic data structure on a partially dynamic graph $G$ that reports for every $w \in {B}(s,r)$, the distance $\dist(s,w)$. In fact, the data structure maintains explicitly the shortest path tree in $G[{B}(s,r)]$. The total update time of the data structure is $O(m r)$ time where $m$ is the maximum number of edges ever in $G[{B}(s,r)]$.
\end{lemma}

Let $\omega$ be the infimum over all reals such that $n\times n$ matrices can be multiplied in $O(n^{\omega+\epsilon})$ time for all $\epsilon>0$. It is known that $\omega\in [2,2.373)$ \cite{vstoc12,legallmult}. More generally, let $\omega(a,b,c)$ be the 
infimum over all reals such that an $n^a\times n^b$ matrix can be multiplied by an $n^b\times n^c$ matrix in $O(n^{\omega(a,b,c)+\epsilon})$ time for all $\epsilon>0$.
A notable result is that for $b\leq 0.313$, $\omega(1,b,1)=2$ \cite{GallU18}.

\section{Conditional lower bounds}\label{sec:lb}

We present conditional lower bounds for amortized algorithms in the fully dynamic, incremental, and decremental settings. Our constructions for amortized algorithms in the fully dynamic setting also imply lower bounds for the incremental and decremental settings, but only for worst-case update times. We present separate constructions for the \emph{amortized} incremental and decremental settings.

%the following is copied from dynamic diameter paper
We first describe why our fully dynamic conditional lower bounds also apply to the incremental and decremental
settings for worst-case update times. This is due to the nature of our reductions:
all of our reductions produce an initial graph on which we perform update stages that only insert or only delete (we can
choose which) a batch of edges, ask a query and undo the changes just made, returning to the initial
graph. An incremental (resp. decremental) algorithm can be used for this type of dynamic graph by performing the insertions (resp. deletions) and then rolling back the data structure to the initial graph and repeating.

\subsection{Conditional lower bounds from the OMv conjecture}

\subsubsection{Statement of results}
We prove conditional lower bounds from the OMv conjecture for dynamic emulator maintenance in the fully dynamic, incremental, and decremental settings. We prove the following theorem for the fully dynamic setting, which also extends to the incremental and decremental settings but only for worst-case update times.

\begin{theorem}\label{thm:oumv-fully}
Under the OMv conjecture, for an $n$-vertex fully dynamic graph with at most $m$ edges at all times, for every constant $\epsilon>0$, there is no algorithm for maintaining a $(1+\alpha,\beta)$-emulator for any $\alpha\in[0,2/3)$ and integer $\beta\geq 0$, with $O(m^{1-\epsilon}(\frac{2-3\alpha}{\beta}))$ edges, polynomial preprocessing time, and  amortized update time $O(n^{1-\epsilon}(\frac{2-3\alpha}{\beta})^2)$ such that over a polynomial number of edge updates the error probability is at most $1/3$ in the word-RAM model with $O(\log n)$ bit words.

The same result holds for incremental and decremental algorithms, but for worst-case update time.
\end{theorem}

In particular, for the natural setting where $\alpha$ is constant and $\beta=n^{o(1)}$ we have the following corollary:

\begin{corollary}
Under the OMv conjecture, for an $n$-vertex fully dynamic graph with at most $m$ edges at all times, for every constant $\epsilon>0$, there is no algorithm for maintaining a $(1+\alpha,n^{o(1)})$-emulator for any constant $\alpha\in[0,2/3)$ with $O(m^{1-\epsilon})$ edges, polynomial preprocessing time, and  amortized update time $O(n^{1-\epsilon})$ such that over a polynomial number of edge updates the error probability is at most $1/3$ in the word-RAM model with $O(\log n)$ bit words.

The same result holds for incremental and decremental algorithms, but for worst-case update time.
\end{corollary}

For the incremental and decremental settings, we prove the following theorem.

\begin{theorem}\label{thm:oumv-partial}
Under the OMv conjecture, for an $n$-vertex incremental or decremental graph with $m$ edge insertions or deletions, for every constant $\epsilon>0$, there is no algorithm for maintaining a $\beta$-additive emulator for any integer $\beta\geq 0$, with $O(m^{1-\epsilon}/\beta)$ edges, polynomial preprocessing time, and  total update time $O(mn^{1-\epsilon}/\beta^2)$ with error probability at most $1/3$ in the word-RAM model with $O(\log n)$ bit words.
\end{theorem}

In particular, for the natural setting where $\beta=n^{o(1)}$ we have the following corollary:

\begin{corollary}
Under the OMv conjecture, for an $n$-vertex incremental or decremental graph with $m$ edge insertions or deletions, for every constant $\epsilon>0$, there is no  algorithm for maintaining a $n^{o(1)}$-additive emulator with $O(m^{1-\epsilon})$ edges, polynomial preprocessing time, and total update time $O(mn^{1-\epsilon})$ with error probability at most $1/3$ in the word-RAM model with $O(\log n)$ bit words.
\end{corollary}

%\noindent We will reduce from OuMv and then apply Theorem~\ref{thm:oumv} to show hardness under the OMv conjecture.

\subsubsection{Preliminaries}

Our reductions are from the Online Vector-Matrix-Vector Multiplication problem (OuMv). A reduction from OMv to OuMv is known:

\begin{theorem}[OuMv: Theorem 2.7 from~\cite{henzinger2015unifying}]\label{thm:oumv} The OMv conjecture implies that for any constant $\epsilon > 0$, there is no algorithm
for OuMv with with polynomial preprocessing time and computation time $O(n^{3-\epsilon})$ with error probability at most $1/3$ in the word-RAM model with $O(\log n)$ bit words.
\end{theorem}

Given an instance of OuMv, we introduce a basic gadget that we will use in all of our constructions.

\paragraph{The basic gadget} Let $M$ be the $n\times n$ input matrix for the OuMv instance. We construct a gadget as shown in Figure~\ref{fig:omv_basic}. The gadget consists of a bipartite graph $A,B$ where $A=\{a_1,\dots a_n\}$, $B=\{b_1,\dots b_n\}$, and the edge $(a_i,b_j)$ is present if and only if $M_{i,j}=1$. 
%Additionally, there is a vertex $u$ and a vertex $v$ whose edges to $A$ and $B$ respectively will represent the online vectors $u$ and $v$.

\begin{figure}[ht]
  \centering
  \includegraphics[width=.15\linewidth]{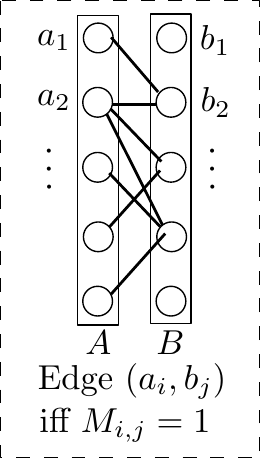}
  \caption{The basic gadget for reductions from OMv.}
  \label{fig:omv_basic}
\end{figure}

\subsubsection{Reduction for fully dynamic algorithms}

In this section we prove Theorem~\ref{thm:oumv-fully}.

%Let us begin with a warm-up reduction which only works for fully-dynamic algorithms. (This reduction was obtained together with Ivan Mikhailin in the Summer of 2016.)

% Suppose that there exists a constant $\epsilon'>0$ so that for an $n'$-vertex $m'$-edge graph there is a fully dynamic algorithm that maintains a $(1+\alpha,\beta)$-emulator for any $\alpha\in[0,2/3)$ and integer $\beta\geq 0$, with $O((m'\frac{2-3\alpha}{\beta})^{2-\epsilon'})$ edges, polynomial preprocessing time, and  amortized update time $O(\frac{2-3\alpha}{\beta}n'^{1-\epsilon'})$. We will show that this implies an algorithm for OuMv with polynomial preprocessing time and computation time $O(n^{3-\epsilon})$ for some constant $\epsilon>0$, which contradicts conjecture~\ref{conj:omv} by Theorem~\ref{thm:oumv}. 

\paragraph{Construction} 
Let $c=\lceil\frac{\beta}{2-3\alpha}\rceil+1$ and take $c$ disjoint copies $G_1,\dots,G_c$ of the basic gadget (from Figure~\ref{fig:omv_basic}). Let $w_0,\dots w_c$ be an additional set of $c+1$ isolated vertices. 

Now, we start $n$ dynamic phases. In phase $i$, we are given the vectors $u^i$ and $v^i$ of the OuMv instance. For each $1\leq j\leq c$ and each $k$ with $u^i_k=1$, insert an edge between $w_{j-1}$ and $a_k\in G_j$. Similarly, for each $1\leq j\leq c$ and each $k$ with $v^i_k=1$, insert an edge between $w_{j}$ and $b_k\in G_j$.  See Figure~\ref{fig:omv_fully}. 

%insert edges between every two consecutive triangle gadgets $a, a+1$, from $x^4_i$ in the $a$th gadget to $x^1_i$ in the $(a+1)$st gadget.
%Call the graph after these $c-1$ insertions, $G_i$. See Figure~\ref{fig:trianglegadget}.

\begin{figure}[ht]
  \centering
  \includegraphics[width=.8\linewidth]{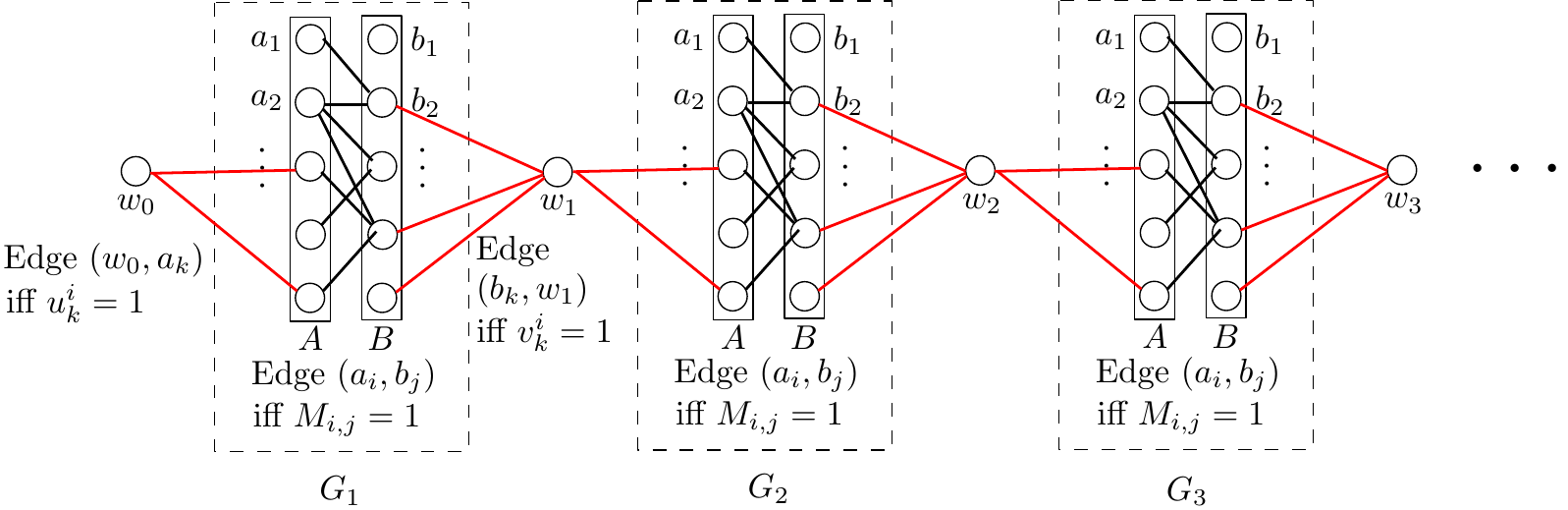}
  \caption{The construction for fully dynamic algorithms. The red edges are dynamically added in phase $i$.}
  \label{fig:omv_fully}
\end{figure}

Throughout all of the edge updates, we maintain our dynamic emulator. At the end of each phase, we run a single source shortest paths computation via  
%Dijkstra's algorithm 
Breadth-First Search (BFS) on the emulator to estimate the distance between $w_0$ and $w_c$. 

If the estimated distance between $w_0$ and $w_c$ is less than $5c$, we return 1 for this phase of the OuMv instance, and otherwise we return 0.

Following the end of each phase, we remove all of the edges added during that phase.

\paragraph{Correctness}

First we will show that if $u^iMv^i=1$ then our algorithm returns 1 for the $i^{th}$ phase. If $u^iMv^i=1$ then there exists $j,k$ such that $u^i_j=M_{j,k}=v^i_k=1$. Thus, the basic gadget contains the edge $(a_j,b_k)$. Also, for all $1\leq \ell\leq c$, in the $i^{th}$ phase we add an edge between $w_{\ell-1}$ and $a_j\in G_\ell$ and an edge between $w_\ell$ and $b_k\in G_\ell$. Thus, there is a path of length 3 from $w_{\ell-1}$ to $w_\ell$ through $a_j$ and $b_k$. Therefore, $\dist(w_0,w_c)\leq 3c$. 
Thus, the estimate of $\dist(w_0,w_c)$ returned by our $(\alpha+1,\beta)$-emulator is at most $3c(\alpha+1)+\beta$, which is less than $5c$ since $c>\frac{\beta}{2-3\alpha}$. Thus, our algorithm returns 1 for the $i^{th}$ phase.

Now we will show that if our algorithm returns 1 in the $i^{th}$ phase then $u^iMv^i=1$. If our algorithm returns 1, then the estimate of  $\dist(w_0,w_c)$ returned by our $(1+\alpha,\beta)$-emulator is less than $5c$, so the true distance $\dist(w_0,w_c)$ is also less than $5c$. 

First, we observe that the layered structure of the graph ensures that 
every path between $w_0$ and $w_c$ must contain every $w_\ell$ in chronological order. That is, every shortest path between $w_0$ and $w_c$ must contain as a subpath a shortest path from $w_{\ell-1}$ to $w_\ell$ for all $1\leq \ell\leq c$. Then since the graph is $c$ identical copies of a gadget, we have that $\dist(w_0,w_c)=c\cdot \dist(w_0,w_1)$.

Since $\dist(w_0,w_c)<5c$, we know that $\dist(w_0,w_1)<5$. Furthermore, since the graph is bipartite and $w_0$ and $w_1$ are on opposite sides of the bipartition, $\dist(w_0,w_1)$ must be odd so $\dist(w_0,w_1)\leq 3$. Observe that the only possible paths of length 3 between $w_{\ell-1}$ and $w_\ell$ contain a vertex $a_j\in G_\ell$ followed by a vertex $b_k\in G_\ell$. If such a path exists in the $i^{th}$ phase, then the basic construction ensures that $M_{j,k}=1$ and the dynamic phase ensures that $u^i_j=1$ and $v^i_k=1$. Thus, $u^iMv^i=1$.

\paragraph{Running time}

Let $n'$ be the number of vertices in the dynamic graph and let $m'$ be the maximum number of edges ever in the dynamic graph. We first calculate $n'$ and $m'$. Each basic gadget contains $2n$ vertices and at most $n^2$ edges. Thus, $c$ copies of the basic gadget contain $2cn$ vertices and $cn^2$ edges. There are also an additional $c+1$ vertices $w_\ell$. During each of the $n$ phases we add at most $2nc$ edges. Thus, the total number of vertices is $n'=O(nc)$ and the total number of edge updates over the entire sequence is $O(cn^2)$, so $m'=O(cn^2)$.

Suppose that the emulator has $O(m'^{1-\epsilon}(2-3\alpha)/\beta)$ edges for $\epsilon>0$ and has polynomial preprocessing time and amortized update time $O(n'^{1-\epsilon}((2-3\alpha)/\beta)^2)$.

Then, our dynamic emulator algorithm has amortized update time $O(n'^{1-\epsilon}(\frac{2-3\alpha}{\beta})^2)=O((nc)^{1-\epsilon}(\frac{2-3\alpha}{\beta})^2)$.  Since there are $O(cn^2)$ edge updates, the total update time of the dynamic emulator algorithm is $O(n^{3-\epsilon})$ since $c=\lceil\frac{\beta}{2-3\alpha}\rceil+1$.

Additionally, $n$ times during the algorithm, we run a single call of BFS
%Dijkstra's algorithm 
on the emulator. The number of edges in the emulator is  $O(m'^{1-\epsilon}(\frac{2-3\alpha}{\beta}))=O((cn^2)^{1-\epsilon}(\frac{2-3\alpha}{\beta}))=O(n^{2-\epsilon})$ since $c=\lceil\frac{\beta}{2-3\alpha}\rceil+1$. Thus, running %Dijkstra's algorithm 
BFS takes total time $O(n^{3-\epsilon})$.
% $\tilde{O}(n^{3-\epsilon})$, which is $O(n^{3-\epsilon'})$ for any constant $\epsilon'<\epsilon$.

Putting everything together, our dynamic emulator algorithm implies an algorithm for OuMv with polynomial preprocessing time and computation time $O(n^{3-\epsilon})$ for $\epsilon>0$, contadicting the OMv conjecture.

\subsubsection{Reduction for incremental and decremental algorithms.}
In this section we prove Theorem~\ref{thm:oumv-partial}.

%Suppose there exists a constant $\epsilon>0$ such that for an $N$-vertex graph there is an incremental or decremental dynamic algorithm for maintaining a $\beta$-additive emulator for any integer $\beta\geq 0$, with $O(N^{2-\epsilon}/\beta^2)$ edges, polynomial preprocessing time, and  total update time $O(N^{3-\epsilon}/\beta^2)$. We will show that this implies an algorithm for OuMv with polynomial preprocessing time and computation time $O(n^{3-\epsilon'})$ for some constant $\epsilon'>0$, which contradicts the OMv conjecture by Theorem~\ref{thm:oumv}. 

\paragraph{Construction} 

The construction will be similar to the fully dynamic construction, but with different interactions between consecutive copies of the basic gadget. We first describe the incremental construction.

Starting with an empty graph, we perform edge insertions to construct the following graph. Take $\beta+1$ disjoint copies $G_1,\dots,G_{\beta+1}$ of the basic gadget (from Figure~\ref{fig:omv_basic}). Then, add $\beta+2$ paths $P_0,\dots,P_{\beta+1}$ each on $2n-1$ new vertices. Call the vertices of each path $z_1,z_2,\ldots,z_n=y_n,y_{n-1},\ldots,y_1$. In other words, the middle node of each path has two names, $z_n$ and $y_n$. %This concludes the description of the initial graph. 

Now, we start $n$ phases. In phase $i$, we are given the vectors $u^i$ and $v^i$ of the OuMv instance. For each $1\leq j\leq \beta+1$ and each $k$ with $u^i_k=1$, insert an edge between $y_i\in P_{j-1}$ and $a_k\in G_j$. Similarly, for each $1\leq j\leq \beta+1$ and each $k$ with $v^i_k=1$, insert an edge between $z_i\in P_j$ and $b_k\in G_j$. See Figure~\ref{fig:omv_inc}.

\begin{figure}[ht]
  \centering
  \includegraphics[width=\linewidth]{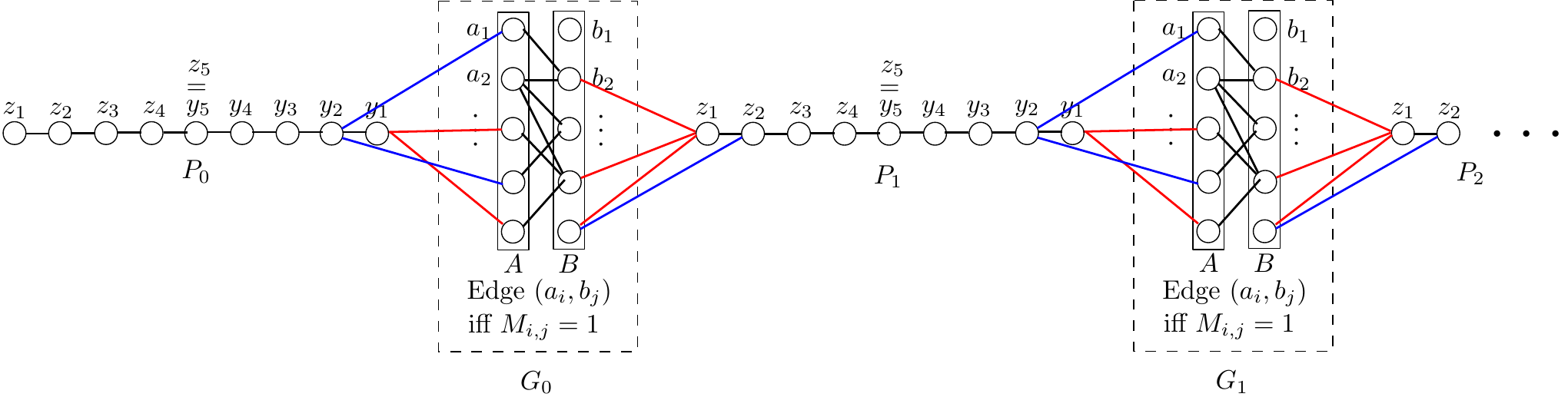}
  \caption{The construction for incremental algorithms. The red edges are dynamically added in phase 1 and the blue edges are dynamically added in phase 2.}
  \label{fig:omv_inc}
\end{figure}

Throughout all of the edge updates, we maintain our incremental emulator. At the end of each phase, we run a single call of BFS
%Dijkstra's algorithm 
on the emulator to estimate the distance between $z_1\in P_0$ and $y_1\in P_{\beta+1}$. 

If the estimated distance between $z_1\in P_0$ and $y_1\in P_{\beta+1}$ at the end of phase $i$ is less than\\ $4(\beta+1)+(\beta+2)(2n-2i)+2(i-1)$, we return 1 for phase $i$ of the OuMv instance, and otherwise we return 0.

Now, we describe the decremental construction, which is similar to the incremental construction. The initial graph consists of $G_1,\dots,G_{\beta+1}$ and  $P_0,\dots,P_{\beta+1}$ from the incremental construction, as well as an edge for all $1\leq \ell\leq \beta+1$ from $z_i\in P_\ell$ to every vertex in $B\subseteq G_\ell$, and an edge for all $1\leq \ell\leq \beta+1$ from $y_i\in P_{\ell-1}$ to every vertex in $A\subseteq G_\ell$. %This concludes the description of the initial graph. 

Now, we start $n$ dynamic phases. In phase $i$, we are given the vectors $u^i$ and $v^i$ of the OuMv instance. For each $1\leq j\leq \beta+1$, and each $k$ with $u^i_k=0$, delete the edge between $y_{n-i+1}\in P_{j-1}$ and $a_k\in G_j$. Similarly, for each $1\leq j\leq \beta+1$ and each $k$ with $v^i_k=0$, delete the edge between $z_{n-i+1}\in P_j$ and $b_k\in G_j$. 

Throughout all of the edge updates, we maintain our decremental emulator. At the end of each phase, we run a single call to BFS on the emulator to estimate the distance between $z_1\in P_0$ and $y_1\in P_{\beta+1}$. 

Following the end of each phase $i$, for each $1\leq j\leq \beta+1$, we delete all edges between $y_{n-i+1}\in P_{j-1}$ and $A\subseteq G_j$ and all edges between $z_{n-i+1}\in P_j$ and $B\subseteq G_j$. 

If the estimated distance between $z_1\in P_0$ and $y_1\in P_{\beta+1}$ at the end of phase $i$ is less than\\ $4(\beta+1)+(\beta+2)(2n-2(n-i+1))+2((n-i+1)-1)$, we return 1 for phase $i$ of the OuMv instance, and otherwise we return 0. Note that this threshold is exactly the threshold from the incremental algorithm but with $i$ replaced with $n-i+1$.

\paragraph{Correctness}
The following argument is written for the incremental setting but the same argument applies for the decremental setting.

First we will show that if $u^iMv^i=1$ then our algorithm returns 1 for the $i^{th}$ phase. If $u^iMv^i=1$ then there exists $j,k$ such that $u^i_j=M_{j,k}=v^i_k=1$. Thus, the basic gadget contains the edge $(a_j,b_k)$. Also, for all $1\leq \ell\leq \beta+1$, in the $i^{th}$ phase we add an edge between $y_i\in P_{\ell-1}$ and $a_j\in G_\ell$ and an edge between  $z_i\in P_\ell$ and $b_k\in G_\ell$.

Thus, for all $1\leq \ell \leq \beta+1$, there is a path of length 3 from $y_i\in P_{\ell-1}$ to $z_i\in P_\ell$ through $a_j\in G_\ell$ and $b_k\in G_\ell$. Also, for each $0\leq \ell \leq \beta+2$, there is a path along $P_\ell$ from $z_i\in P_\ell$ to $y_i\in P_\ell$ of length $(2n-2)-2(i-1)=2n-2i$. Finally, there is a path of length $i-1$ from $z_1\in P_0$ to $z_i\in P_0$ and a path of length $i-1$ from $y_i\in P_{\beta+2}$ to $y_1\in P_{\beta+2}$. Concatentating all of these paths, we have that $\dist(z_1\in P_0,y_1\in P_{\beta+1})\leq 3(\beta+1)+(\beta+2)(2n-2i)+2(i-1)$. 
Thus, the estimate of $\dist(z_1\in P_0,y_1\in P_{\beta+1})$ returned by our $\beta$-additive emulator is at most $4(\beta+1)-1+(\beta+2)(2n-2i)+2(i-1)$. Thus, our algorithm returns 1 for the $i^{th}$ phase.

Now we will show that if our algorithm returns 1 in the $i^{th}$ phase then $u^iMv^i=1$. If our algorithm returns 1, then the estimate of  $\dist(z_1\in P_0,y_1\in P_{\beta+1})$ returned by our $\beta$-additive emulator is less than $4(\beta+1)+(\beta+2)(2n-2i)+2(i-1)$, so the true distance $\dist(z_1\in P_0,y_1\in P_{\beta+1})$ is also less than $4(\beta+1)+(\beta+2)(2n-2i)+2(i-1)$. 

First, we observe that the layered structure of the graph ensures that 
every path between $z_1\in P_0$ and $y_1\in P_{\beta+1}$ contains each $z_i$ and $y_i$ in order from $P_0$ to $P_{\beta+1}$. That is, every shortest path between $z_1\in P_0$ and $y_1\in P_{\beta+1}$ is composed of precisely following subpaths:
\begin{itemize}
\item A shortest path from $z_1\in P_0$ to $z_i\in P_0$. The only simple path connecting these vertices is of length $i-1$.
\item A shortest path from $y_i\in P_{\beta+1}$ to $y_1\in P_{\beta+1}$. The only simple path connecting these vertices is of length $i-1$.
\item A shortest path from $z_i\in P_\ell$ to $y_i\in P_\ell$ for all $1\leq \ell\leq \beta+2$. The only simple path connecting these vertices is of length $(2n-2)-2(i-1)=2n-2i$.
\item A shortest path from $y_i\in P_{\ell-1}$ to $z_i\in P_\ell$ for all $1\leq \ell\leq \beta+1$. Since the graph is a series of identical copies of a gadget, we know that $\dist(y_i\in P_{\ell-1},z_i\in P_\ell)$ is the same for all $\ell$. Furthermore, we know the length of each of the previous three types of subpaths and we know that $\dist(z_1\in P_0,y_1\in P_{\beta+1})<4(\beta+1)+(\beta+2)(2n-2i)+2(i-1)$, so we conclude that each $\dist(y_i\in P_{\ell-1},z_i\in P_\ell)<4$.
\end{itemize}

Due to the layering of the graph, for all $1\leq \ell\leq \beta+1$ the shortest path between $y_i\in P_{\ell-1}$ and $z_i\in P_\ell$ must contain vertices $a_j\in G_\ell$ and $b_k\in G_\ell$ for some $j,k$. Since $\dist(y_i\in P_{\ell-1},z_i\in P_\ell)<4$, there are no other vertices on this shortest path. The basic construction ensures that since the edge $(a_j,b_k)$ exists, we have $M_{j,k}=1$, and the dynamic phase ensures that since the edge $(y_i\in P_{\ell-1},a_j\in G_\ell)$ exists, we have $u^i_j=1$ and since the edge $(b_k\in G_\ell,z_i\in P_\ell)$ exists, we have $v^i_k=1$. Thus, $u^iMv^i=1$.

\paragraph{Running time} Let $n'$ be the number of vertices in the dynamic graph and let $m'$ be number of edge insertions or deletions. We first calculate $n'$ and $m'$. Each basic gadget contains $2n$ vertices and at most $n^2$ edges. Thus, $\beta+1$ copies of the basic gadget contain $2(\beta+1)n$ vertices and $(\beta+1)n^2$ edges. Additionally, we have $(\beta+2)$ paths on $(2n-1)$ vertices each, for a total of  $(2n-1)(\beta+2)$ additional vertices. During each of the $n$ phases there are at most $2n(\beta+1)$ edge updates. Thus, the total number of vertices is $n'=O(\beta n)$ and the total number of edges updates is $m'=O(\beta n^2)$.

Now assume that our incremental or decremental emulator has $O(m'^{1-\epsilon}/\beta)$ edges, polynomial preprocessing time and total update time $O(m'n'^{1-\epsilon}/\beta^2)$.

Then the total update time of the emulator is $O(m'n'^{1-\epsilon}/\beta^2) =O(\beta n^2 (\beta n)^{1-\epsilon}/\beta^2)=O(n^{3-\epsilon})$. 

Additionally, $n$ times during the algorithm, we run BFS 
%Dijkstra's algorithm 
on the emulator. The number of edges in the emulator is $O(m'^{1-\epsilon}/\beta)=O((\beta n^2)^{1-\epsilon}/\beta)=O(n^{2-\epsilon})$. Thus, running the BFS calls takes total time $O(n^{3-\epsilon})$.
%, which is $O(n^{3-\epsilon'})$ for any constant $\epsilon'<\epsilon$.

Putting everything together, our incremental or decremental emulator algorithm implies an algorithm for OuMv with polynomial preprocessing time and computation time $O(n^{3-\epsilon})$ for $\epsilon>0$, contradicting the OMv conjecture.

\subsection{Conditional lower bounds from the {\em k}-Clique hypothesis}

\subsubsection{Statement of results}
We prove conditional lower bounds from the $k$-Clique hypothesis for dynamic emulator maintenance in the fully dynamic, incremental, and decremental settings. We prove the following theorem for the fully dynamic setting, which also extends to the incremental and decremental settings but only for worst-case update times.
\begin{theorem}\label{thm:kcycle-fully}
Under the Combinatorial $k$-Clique hypothesis, for every constant $\epsilon>0$ and every constant integer $\ell\geq 1$, for an $n$-vertex fully dynamic graph with at most $m=\Theta(n^{1+1/\ell})$ edges at all times, there is no combinatorial algorithm for maintaining a $(1+\alpha,\beta)$-emulator for any $\alpha\in[0,\frac{2}{2\ell+1})$ and integer $\beta\geq 0$ with $O(m^{1-\epsilon}(\frac{2-(2\ell+1)\alpha}{\beta}))$ edges, preprocessing time $O(mn^{1-\epsilon}(\frac{2-(2\ell+1)\alpha}{\beta})^2)$, and amortized update time $O(m^{1-\epsilon}(\frac{2-(2\ell+1)\alpha}{\beta})^2)$ with error probability at most $1/3$ in the word-RAM model with $O(\log n)$ bit words.

The same result holds for incremental and decremental algorithms, but for worst-case update time.
\end{theorem}

In particular, for the natural setting where $\alpha$ is constant and $\beta=n^{o(1)}$ we have the following corollary:

\begin{corollary}
Under the Combinatorial $k$-Clique hypothesis, for every constant $\epsilon>0$ and every constant integer $\ell\geq 1$, for an $n$-vertex fully dynamic graph with at most $m=\Theta(n^{1+1/\ell})$ edges at all times, there is no combinatorial algorithm for maintaining a $(1+\alpha,n^{o(1)})$-emulator for any constant $\alpha\in[0,\frac{2}{2\ell+1})$ with $O(m^{1-\epsilon})$ edges, preprocessing time $O(mn^{1-\epsilon})$, and amortized update time $O(m^{1-\epsilon})$ with error probability at most $1/3$ in the word-RAM model with $O(\log n)$ bit words.

The same result holds for incremental and decremental algorithms, but for worst-case update time.
\end{corollary}

When $k=3$ (i.e. triangle detection), our construction also implies a conditional lower bound under the hypothesis that triangle detection cannot be done in time $O(n^{\omega-\epsilon})$ for any constant $\epsilon$ even for non-combinatorial algorithms:

\begin{theorem}\label{thm:triangle}
Under the hypothesis that there is no algorithm for triangle detection in $O(n^{\omega-\delta})$ for any constant $\delta$, for every constant $\epsilon>0$ for an $n$-vertex fully dynamic graph with at most $m$ edges at all times, there is no algorithm for maintaining a $(1+\alpha,\beta)$-emulator for any $\alpha\in[0,2/3)$ and integer $\beta\geq 0$ with $O(n^{\omega-1-\epsilon}(\frac{2-3\alpha}{\beta})^{\omega-1})$ edges, preprocessing time $O(n^{\omega-\epsilon}(\frac{2-3\alpha}{\beta})^\omega)$, and amortized update time $O(n^{\omega-1-\epsilon}(\frac{2-3\alpha}{\beta})^\omega)$ with error probability at most $1/3$ in the word-RAM model with $O(\log n)$ bit words.

The same result holds for incremental and decremental algorithms, but for worst-case update time.
\end{theorem}

In particular, for the natural setting where $\alpha$ is constant and $\beta=n^{o(1)}$ we have the following corollary:

\begin{corollary}
Under the hypothesis that there is no algorithm for triangle detection in $O(n^{\omega-\delta})$ for any constant $\delta$, for every constant $\epsilon>0$ for an $n$-vertex fully dynamic graph with at most $m$ edges at all times, there is no algorithm for maintaining a $(1+\alpha,n^{o(1)})$-emulator for any constant $\alpha\in[0,2/3)$  with $O(n^{\omega-1-\epsilon})$ edges, preprocessing time $O(n^{\omega-\epsilon})$, and amortized update time $O(n^{\omega-1-\epsilon})$ with error probability at most $1/3$ in the word-RAM model with $O(\log n)$ bit words.

The same result holds for incremental and decremental algorithms, but for worst-case update time.
\end{corollary}

For the incremental and decremental settings, we prove the following theorem.

\begin{theorem}\label{thm:kcycle-partial}
Under the Combinatorial $k$-Clique hypothesis, for every constant $\epsilon>0$ and every constant integer $\ell\geq 1$, for an $n$-vertex incremental or decremental graph with $m=\Theta(n^{1+1/\ell})$ edge insertions or deletions, there is no combinatorial algorithm for maintaining a $\beta$-additive emulator for any integer $\beta\geq 0$ with $O(m^{1-\epsilon}/\beta)$ edges and total time $O(mn^{1-\epsilon}/\beta^2)$ with error probability at most $1/3$ in the word-RAM model with $O(\log n)$ bit words.
\end{theorem}

In particular, for the natural setting where $\beta=n^{o(1)}$ we have the following corollary:

\begin{corollary}
Under the Combinatorial $k$-Clique hypothesis, for every constant $\epsilon>0$ and every constant integer $\ell\geq 1$, for an $n$-vertex incremental or decremental graph with $m=\Theta(n^{1+1/\ell})$ edge insertions or deletions, there is no combinatorial algorithm for maintaining a $n^{o(1)}$-additive emulator with $O(m^{1-\epsilon})$ edges and total time $O(mn^{1-\epsilon})$ with error probability at most $1/3$ in the word-RAM model with $O(\log n)$ bit words.
\end{corollary}

%\noindent We will reduce from the directed $k$-cycle problem and then apply Theorem~\ref{thm:kcycle} to show hardness under the $k$-Clique hypothesis.

\subsubsection{Preliminaries}
Our reductions are from the $k$-Cycle problem. A reduction from $k$-clique to $k$-cycle in graphs of all sparsities is known:

\begin{theorem}[Combinatorial $k$-Cycle~\cite{lincoln2018tight}]\label{thm:kcycle}
Under the Combinatorial $k$-Clique hypothesis, for every constant $\epsilon>0$ and every integer $\ell\geq 1$, there is no combinatorial algorithm for detecting a $k=2\ell+1$-cycle in a directed graph with $m=\Theta(n^{1+1/\ell})$ edges in time $O(mn^{1-\epsilon})$ with error probability at most $1/3$ in the word-RAM model with $O(\log n)$ bit words.
\end{theorem}

Given an instance of $k$-Cycle, we introduce a basic gadget that we will use in all of our constructions.

\paragraph{The basic gadget} Let $G=(V,E)$ with $n=|V|$ and $m=|E|$ be the graph on which we wish to find a directed $k$-cycle. We use color coding: we color the vertices with colors in $1,2,\dots,k$ uniformly at random. For all $i$, let $V_i$ be the set of vertices of color $i$. We construct a gadget as shown in Figure~\ref{fig:clique_basic}.

The gadget consists of $k+1$ layers of vertices. For all $1\leq i\leq k$, layer $i$ contains a copy $x^i_j$ of every $x_j\in V_j$. Layer $k+1$ contains a copy $x^{k+1}_j$ of every $x_j\in V_1$. For two consecutive layers $i,i+1$, we include the undirected edge $(x^i_a,x^{i+1}_b)$ if and only if the directed edge $(x_a,x_b)$ is in $E$. 

\begin{figure}[ht]
  \centering
  \includegraphics[width=.25\linewidth]{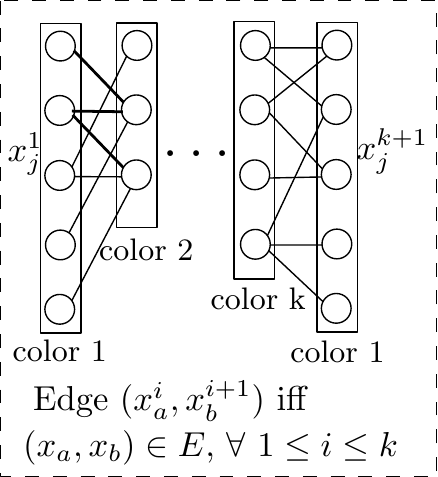}
  \caption{The basic gadget for reductions from $k$-clique.}
  \label{fig:clique_basic}
\end{figure}

Our constructions will use several copies of the basic gadget. Each copy uses the same coloring of $G$ so each copy is identical.

\subsubsection{Reduction for fully dynamic algorithms}
In this section we prove Theorems~\ref{thm:kcycle-fully} and~\ref{thm:triangle}.

\paragraph{Construction} 

We say that a $k$ cycle in $G$ is \emph{colorful} if according to the coloring from the basic gadget the $k$-cycle has exactly one vertex of each color and the vertices are in color order $1,2,\dots k$ around the cycle (i.e. the vertex of color 1 is a adjacent to the vertex of color $k$). We will present an algorithm that detects a colorful $k$-cycle in $G$ if one exists. Any given $k$-cycle is colorful with probability $1/k^{k-1}$. We repeat the entire algorithm, including construction of the basic gadget, $\Theta(k^{k-1})$ times so that if $G$ contains a $k$-cycle, then with probability at least $2/3$, for at least one of the repetitions $G$ contains a colorful $k$-cycle. 

We will construct a dynamic graph $G'$. Let $c=\lceil\frac{\beta}{2-k\alpha}\rceil+1$ and take $c$ disjoint copies $G'_1,\dots,G'_c$ of the basic gadget. This completes the preprocessing phase.

Now, we start the dynamic phases. There is one dynamic phase for each vertex in $G$ of color 1. In each phase $i$, we insert an undirected edge between every pair of consecutive gadgets $G'_j$ and $G'_{j+1}$ from $x^{k+1}_i\in G'_j$ $x^1_i\in G'_{j+1}$. See Figure~\ref{fig:clique_fully}.

\begin{figure}[ht]
  \centering
  \includegraphics[width=.9\linewidth]{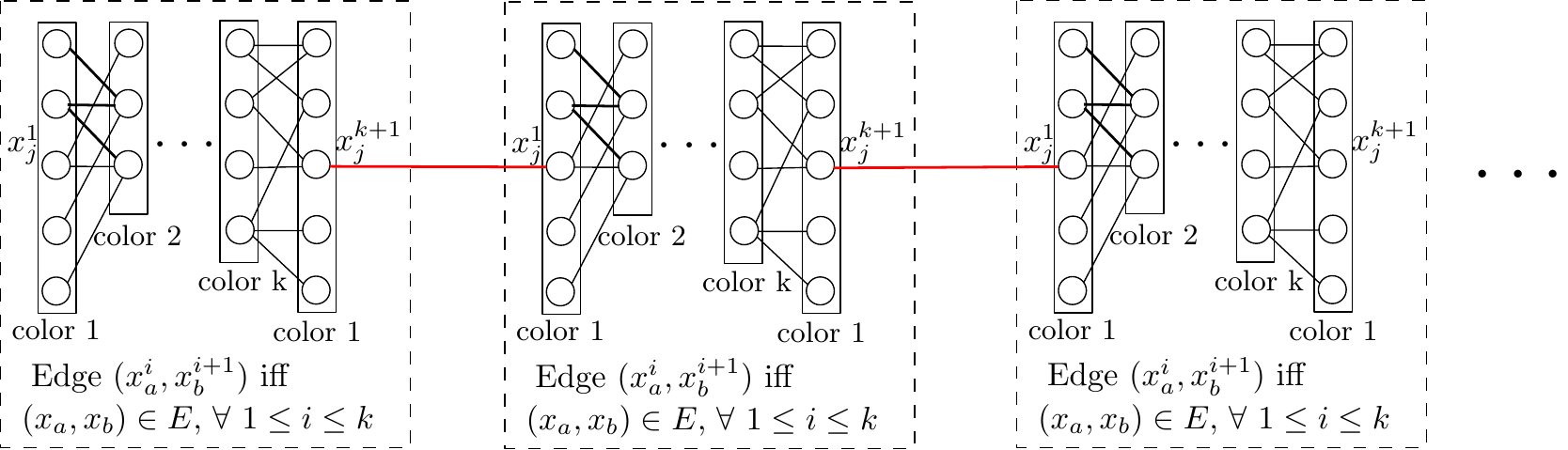}
  \caption{The construction for fully dynamic algorithms. The red edges are dynamically added in phase $j$.}
  \label{fig:clique_fully}
\end{figure}

Throughout all of the edge updates, we maintain our dynamic emulator. At the end of each phase $i$, we run a single call to BFS
%Dijkstra's algorithm 
on the emulator to estimate the distance between $x^1_i\in G'_1$ and $x^{k+1}_i\in G'_c$. If the estimated distance is less than $(k+2)c-1$, we return that we have detected a $k$-cycle. 

Following the end of each phase, we remove all of the edges added during that phase.

If after all phases of all $\Theta(k^{k-1})$ repetitions of the algorithm, we have not detected a $k$-cycle, we return that the graph has no $k$-cycles.

\paragraph{Correctness} First we will show that if the graph $G$ contains a $k$-cycle then our algorithm detects one. Suppose we are in a repetition of the algorithm where this $k$-cycle is colorful. Without loss of generality, let $x_1,\dots, x_k$ be the vertices in a $k$-cycle in $G$. such that each $x_i$ has color $i$.

We claim that our algorithm detects this $k$-cycle at the end of the first phase. The basic construction ensures that for all $1\leq i \leq k-1$, in each gadget $G'_j$, the edge $(x^i_i,x^{i+1}_{i+1})$ exists and the edge $(x^k_k,x^{k+1}_1)$ exists. Thus, there is a path of length $k$ from $x^1_1$ to $x^{k+1}_1$ in each basic gadget.

Furthermore, in the first dynamic phase we insert an edge between every two consecutive gadgets $G'_j$ and $G'_{j+1}$ from $x^{k+1}_1\in G'_j$ $x^1_1\in G'_{j+1}$. Concatenating these paths, we have a path of length $kc-1$ from $x^1_1\in G'_1$ and $x^{k+1}_1\in G'_c$. Thus, the estimate of $\dist(x^1_1\in G'_1,x^{k+1}_1\in G'_c)$ returned by our $(1+\alpha,\beta)$-emulator is at most $(kc-1)(1+\alpha)+\beta<(k+2)c-1$ by choice of $c$. Therefore, our algorithm returns that we have detected a $k$-cycle.

Now we will show that if our algorithm returns that we have detected a $k$-cycle, then $G$ contains a $k$-cycle. Let $i$ be the phase that our algorithm detects a $k$-cycle. Consider $G'$ at the end of phase $i$. Because our algorithm detected a $k$-cycle, we know that the distance between $x^1_i\in G'_1$ and $x^{k+1}_i\in G'_c$ estimated by our emulator is less than $(k+2)c-1$. Therefore, the true distance between $x^1_i\in G'_1$ and $x^{k+1}_i\in G'_c$ is also less than $(k+2)c-1$. 

We note that since there is only one edge between every pair of adjacent gadgets $G'_i$, any shortest path between $x^1_i\in G'_1$ and $x^{k+1}_i\in G'_c$ contains for every edge whose endpoints are in different copies of the basic gadget. That is, this path contains for every $1\leq j\leq c-1$ the edge $(x^{k+1}_i\in G'_j,x^1_i\in G'_{j+1})$.  
Therefore, for every $1\leq j\leq c-1$ this path contains as a subpath a shortest path between $x^1_i\in G'_j$ and $x^{k+1}_i\in G'_j$. Since $G'$ contains $c$ identical copies of the basic gadget, $\dist(x^1_i\in G'_j,x^{k+1}_i\in G'_j)$ is the same for all $j$. 

Since $\dist(x^1_i\in G'_1,x^{k+1}_i\in G'_c)<(k+2)c-1$ the edges between gadgets contribute $c-1$ to this quantity, we know that for all $1\leq j\leq c$, $\dist(x^1_i\in G'_j,x^{k+1}_i\in G'_j)<k+1$. Fix $j$. Since each basic gadget contains $k+1$ layers, we know that $\dist(x^1_i\in G'_j,x^{k+1}_i\in G'_j)\geq k$. Therefore, there is a path from $x^1_i\in G'_j$ to $x^{k+1}_i\in G'_j$ that contains exactly one vertex from each layer. The construction of the basic gadget ensures that there is an edge $(x^i_a\in G'_j,x^{i+1}_b\in G'_j)$ if and only if the directed edge $(x_a,x_b)$ is in $E$. Thus, this path from $x^1_i\in G'_j$ to $x^{k+1}_i\in G'_j$ corresponds to a directed walk of length $k$ in $G$. In particular the first and last vertex on this walk are both $x_i$ so this is a closed walk. Furthermore, every internal layer of $G'_j$ corresponds to a different color, so every vertex of the closed walk in $G$ has a different color. Thus, every vertex of the closed walk is distinct so the closed walk is indeed a directed $k$-cycle. 

\paragraph{Running time} Let $n'$ be the number of vertices in the dynamic graph and let $m'$ be the maximum number of edges ever in the dynamic graph. We first calculate $n'$ and $m'$. Each basic gadget contains $O(n)$ vertices and $O(m)$ edges. Thus, $c$ copies of the basic gadget contain $n'=O(cn)$ vertices and $O(cm)$ edges. During each of the at most $n$ phases we add $c-1$ edges so the number of edge updates after preprocessing is $O(cm)$. Thus, $m'=O(cm)=O(cn^{1+1/\ell})=O(c(n'/c)^{1+1/\ell})=O(n'^{1+1/\ell})$. We repeat the entire algorithm $O(k^{k-1})=O(1)$ times. Thus, the total number of edge updates over the entire sequence is $O(cn)$.

We will now split our running time analysis into two -- one for refuting the combinatorial $k$-Clique conjecture, and one for refuting the hypothesis that triangles cannot be solved faster than matrix multiplication.

\subparagraph{Combinatorial $k$-Clique} 
Let us assume that the dynamic emulator has, for some $\epsilon>0$, preprocessing time $O(m'n'^{1-\epsilon} \left(\frac{2-(2\ell+1)\alpha}{\beta}\right)^2)$, and amortized update time $O(m'^{1-\epsilon} \left(\frac{2-(2\ell+1)\alpha}{\beta}\right)^2)$.

Our dynamic emulator algorithm has preprocessing time\\ $O(m'n'^{1-\epsilon}(\frac{2-(2\ell+1)\alpha}{\beta})^2)=O((cm)(cn)^{1-\epsilon}/c^2)=O(mn^{1-\epsilon})$ and amortized update time\\ $O(m'^{1-\epsilon}(\frac{2-(2\ell+1)\alpha}{\beta})^2)=O((cm)^{1-\epsilon}/c^2)=O(m^{1-\epsilon}/c)$. Since there are $O(cn)$ edge updates, the total running time due to the dynamic emulator algorithm is $O(mn^{1-\epsilon})$.

Additionally, at most $n$ times during each repetition the algorithm, we run BFS on the emulator. The number of edges in the emulator is $O(m'^{1-\epsilon}(\frac{2-(2\ell+1)\alpha}{\beta}))=O((cm)^{1-\epsilon}/c)=O(m^{1-\epsilon})$. Thus, the BFS calls take total time $O(nm^{1-\epsilon})=O(mn^{1-\epsilon})$.
%=O(mn^{1-\epsilon'})$ for any $\epsilon'<\epsilon$.

Putting everything together, our combinatorial dynamic emulator algorithm implies a combinatorial algorithm for directed $k$-cycle detection in time $O(mn^{1-\epsilon})$, thus refuting the combinatorial $k$-Clique hypothesis.

\subparagraph{Non-combinatorial triangle} 
Here we assume that the dynamic emulator has, for some $\epsilon>0$, preprocessing time $O(n'^{\omega-\epsilon}(\frac{2-3\alpha}{\beta})^\omega)$ and amortized update time $O(n'^{\omega-1-\epsilon}(\frac{2-3\alpha}{\beta})^\omega)$.

Thus, our dynamic emulator algorithm has preprocessing time\\ $O(n'^{\omega-\epsilon}(\frac{2-3\alpha}{\beta})^\omega)=O((cn)^{\omega-\epsilon}/c^\omega)=O(n^{\omega-\epsilon})$ and amortized update time\\ $O(n'^{\omega-1-\epsilon}(\frac{2-3\alpha}{\beta})^\omega)=O((cn)^{\omega-1-\epsilon}/c^\omega)=O(n^{\omega-1-\epsilon}/c^{1+\epsilon})$. Since there are $O(cn)$ edge updates, the total running time of the dynamic emulator algorithm is $O(n^{\omega-\epsilon})$.

Additionally, at most $n$ times during each repetition the algorithm, we run 
BFS on the emulator. The number of edges in the emulator is $O(n'^{\omega-1-\epsilon}(\frac{2-3\alpha}{\beta})^{\omega-1})=O((cn)^{\omega-1-\epsilon}/c^{\omega-1})=O(n^{\omega-1-\epsilon})$. Thus, the BFS calls take total time $O(n^{\omega-\epsilon})$.

Putting everything together, our dynamic emulator algorithm implies an algorithm for triangle detection in time $O(n^{\omega-\epsilon})$, refuting the hypothesis that triangle detection needs $n^{\omega-o(1)}$ time.

\subsubsection{Reduction for incremental and decremental algorithms.}

In this section we prove Theorem~\ref{thm:kcycle-partial}.

\paragraph{Construction} 

The construction is similar to the fully dynamic construction, but with different interactions between consecutive gadgets. We first describe the incremental construction.

As in the fully dynamic construction, we say that a $k$ cycle in $G$ is \emph{colorful} if according to the coloring from the basic gadget the $k$-cycle has exactly one vertex of each color and the vertices are in color order $1,2,\dots k$ around the cycle (i.e. the vertex of color 1 is a adjacent to the vertex of color $k$). We will present an algorithm that detects a colorful $k$-cycle in $G$ if one exists. Any given $k$-cycle is colorful with probability $1/k^{k-1}$. We repeat the entire algorithm, including construction of the basic gadget, $\Theta(k^{k-1})$ times so that if $G$ contains a $k$-cycle, then with probability at least $2/3$, for at least one of the repetitions $G$ contains a colorful $k$-cycle. 

We will construct a dynamic graph $G'$. Starting with an empty graph, we perform edge insertions to construct the following graph. Take $\beta+1$ disjoint copies $G'_1,\dots,G'_{\beta+1}$ of the basic gadget. Then, we add $\beta+2$ paths $P_0,\dots,P_{\beta+1}$ each on $2n-1$ new vertices. Call the vertices of each path $z_1,z_2,\ldots,z_n=y_n,y_{n-1},\ldots,y_1$. In other words, the middle node of each path has two names, $z_n$ and $y_n$. 

Now, we start the phases. There is one phase for each vertx in $G$ of color 1. In phase $i$, for each $1\leq j\leq \beta+1$ we insert an edge between $y_i\in P_{j-1}$ and $x^1_i\in G'_j$. Similarly, for each $1\leq j\leq \beta+1$ we insert an edge between $z_i\in P_j$ and $x^{k+1}_i\in G'_j$. See Figure~\ref{fig:clique_inc}.

\begin{figure}[ht]
  \centering
  \includegraphics[width=\linewidth]{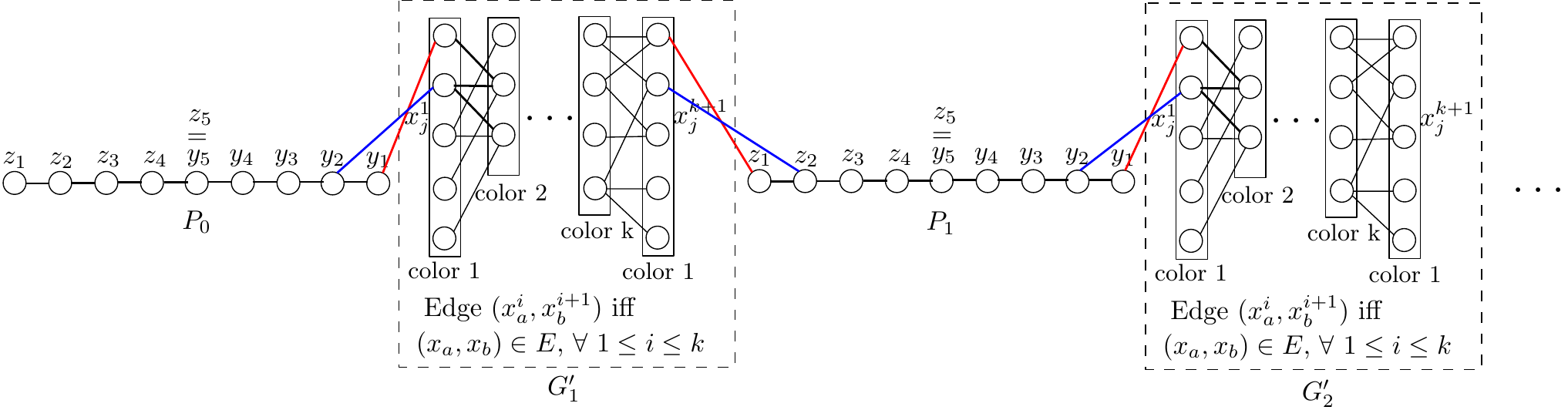}
  \caption{The construction for incremental algorithms. The red edges are dynamically added in phase 1 and the blue edges are dynamically added in phase 2.}
  \label{fig:clique_inc}
\end{figure}

Throughout all of the edge updates, we maintain our incremental emulator. At the end of each phase, we run BFS on the emulator to estimate the distance between $z_1\in P_0$ and $y_1\in P_{\beta+1}$. 

If the estimated distance between $z_1\in P_0$ and $y_1\in P_{\beta+1}$ at the end of phase $i$ is less than\\ $(k+3)(\beta+1)+(\beta+2)(2n-2i)+2(i-1)$, we return that we have detected a $k$-cycle.

If after all phases of all $\Theta(k^{k-1})$ repetitions of the algorithm, we have not detected a $k$-cycle, we return that the graph has no $k$-cycles.

Now, we describe the decremental construction. The edge updates are exactly the reverse of the incremental construction. That is, the initial graph in the decremental construction is the final graph in the incremental construction. Then, in phase $i$ of the decremental construction, for each $1\leq j\leq \beta+1$ we delete the edge between $y_{n-i+1} \in P_{j-1}$ and $x^1_{n-i+1}\in G'_j$. and delete the edge between $z_{n-i+1}\in P_j$ and $x_{{k+1}_{n-i+1}}\in G'_j$. 

Throughout all of the edge updates, we maintain our decremental emulator. At the end of each phase, we run BFS on the emulator to estimate the distance between $z_1\in P_0$ and $y_1\in P_{\beta+1}$. 

If the estimated distance between $z_1\in P_0$ and $y_1\in P_{\beta+1}$ at the end of phase $i$ is less than\\ $(k+3)(\beta+1)+(\beta+2)(2n-2(n-i+1))+2((n-i+1)-1)$, we return that we have detected a $k$-cycle. Note that this threshold is exactly the threshold from the incremental algorithm but with $i$ replaced with $n-i+1$.

If after all phases of all $\Theta(k^{k-1})$ repetitions of the algorithm, we have not detected a $k$-cycle, we return that the graph has no $k$-cycles.

\paragraph{Correctness}
The following argument is written for the incremental setting but the same argument applies for the decremental setting.

First we will show that if the graph $G$ contains a $k$-cycle then our algorithm detects one. Suppose we are in a repetition of the algorithm where this $k$-cycle is colorful. Without loss of generality, let $x_1,\dots, x_k$ be the vertices in a $k$-cycle in $G$ where each $x_i$ is of color $i$.

We claim that our algorithm detects this $k$-cycle at the end of the first phase. The basic construction ensures that for all $1\leq i \leq k-1$, in each gadget $G'_j$, the edge $(x^i_i,x^{i+1}_{i+1})$ exists and the edge $(x^k_k,x^{k+1}_1)$ exists. Thus, there is a path of length $k$ from $x^1_1$ to $x^{k+1}_1$ in each basic gadget.

Also due to the dynamic edge updates, for each $0\leq j \leq \beta+1$, there is an edge between $y_i\in P_{j-1}$ and $x^1_1$ and an edge between $z_i\in P_j$ and $x^{k+1}_1$.
 Additionally, for each $0\leq j \leq \beta+2$, there is a path along $P_j$ from $z_i\in P_j$ to $y_i\in P_j$ of length $(2n-2)-2(i-1)=2n-2i$. Finally, there is a path of length $i-1$ from $z_1\in P_0$ to $z_i\in P_0$ and a path of length $i-1$ from $y_i\in P_{\beta+2}$ to $y_1\in P_{\beta+2}$. Concatenating all of these paths, we have that $\dist(z_1\in P_0,y_1\in P_{\beta+1})\leq (k+2)(\beta+1)+(\beta+2)(2n-2i)+2(i-1)$. 
Thus, the estimate of $\dist(z_1\in P_0,y_1\in P_{\beta+1})$ returned by our $\beta$-additive emulator is at most $(k+2)(\beta+1)+(\beta+2)(2n-2i)+2(i-1)+\beta<(k+3)(\beta+1)+(\beta+2)(2n-2i)+2(i-1)$. 

Now we will show that if our algorithm returns that we have detected a $k$-cycle, then $G$ contains a $k$-cycle. Let $i$ be the phase that our algorithm detects a $k$-cycle. Consider $G'$ at the end of phase $i$. Because our algorithm detected a $k$-cycle, we know that the distance between $z_1\in P_0$ and $y_1\in P_{\beta+1}$ estimated by our emulator is less than $(k+3)(\beta+1)+(\beta+2)(2n-2i)+2(i-1)$. Therefore, the true distance between $z_1\in P_0$ and $y_1\in P_{\beta+1}$ is also less than $(k+3)(\beta+1)+(\beta+2)(2n-2i)+2(i-1)$.

We observe that the layered structure of the graph ensures that 
every path between $z_1\in P_0$ and $y_1\in P_{\beta+1}$ contains each $z_i$ and $y_i$ in order from $P_0$ to $P_{\beta+1}$. That is, every shortest path between $z_1\in P_0$ and $y_1\in P_{\beta+1}$ is composed of precisely following subpaths:
\begin{itemize}
\item A shortest path from $z_1\in P_0$ to $z_i\in P_0$. The only simple path connecting these vertices is of length $i-1$.
\item A shortest path from $y_i\in P_{\beta+1}$ to $y_1\in P_{\beta+1}$. The only simple path connecting these vertices is of length $i-1$.
\item A shortest path from $z_i\in P_j$ to $y_i\in P_j$ for all $1\leq j\leq \beta+2$. The only simple path connecting these vertices is of length $(2n-2)-2(i-1)=2n-2i$.
\item A shortest path from $y_i\in P_{j-1}$ to $z_i\in P_j$ for all $1\leq j\leq \beta+1$. Since the graph is a series of identical copies of a gadget, we know that $\dist(y_i\in P_{j-1},z_i\in P_j)$ is the same for all $j$. Furthermore, we know the length of each of the previous three types of subpaths and we know that $\dist(z_1\in P_0,y_1\in P_{\beta+1})<(k+3)(\beta+1)+(\beta+2)(2n-2i)+2(i-1)$, so we conclude that each $\dist(y_i\in P_{j-1},z_i\in P_j)<k+3$.
\end{itemize}

Due to the layering of the graph, for all $1\leq p\leq \beta+1$ the shortest path between $y_i\in P_{p-1}$ and $z_i\in P_p$ must contain vertices $x^1_j\in G'_p$ and $x^{k+1}_{j'}\in G'_p$ for some $j,j'$. Fix $j,j'$. Since $\dist(y_i\in P_{p-1},z_i\in P_p)<k+3$, we have that $\dist(x^1_j\in G'_p,x^{k+1}_{j'}\in G'_p)<k+1$. Then since each basic gadget contains $k+1$ layers, we know that there is a path from $x^1_j\in G'_p$ to $x^{k+1}_{j'}\in G'_p$ that contains exactly one vertex from each layer. The construction of the basic gadget ensures that there is an edge $(x^i_a\in G'_j,x^{i+1}_b\in G'_j)$ if and only if the directed edge $(x_a,x_b)$ is in $E$. Thus, this path from $x^1_i\in G'_p$ to $x^{k+1}_i\in G'_p$ corresponds to a directed walk of length $k$ in $G$. In particular the first and last vertex on this walk are both $x_i$ so this is a closed walk. Furthermore, every internal layer of $G'_j$ corresponds to a different color, so every vertex of the closed walk in $G$ has a different color. Thus, every vertex of the closed walk is distinct so the closed walk is indeed a directed $k$-cycle.

\paragraph{Running time} Let $n'$ be the number of vertices in $G'$ and let $m'$ be number of edge insertions or deletions. We first calculate $n'$ and $m'$. Each basic gadget contains $O(n)$ vertices and $O(m)$ edges. Thus, $\beta+1$ copies of the basic gadget contain $O(\beta n)$ vertices and $O(\beta m)$ edges. Additionally, we have $(\beta+2)$ paths on $(2n-1)$ vertices each, for a total of  $(2n-1)(\beta+2)$ additional vertices. Thus, $n'=O(\beta n)$. During each of the at most $n$ phases there are at most $2(\beta+1)$ edge updates. Thus, $m'=O(\beta m)$. We repeat the entire algorithm $O(k^{k-1})=O(1)$ times. Thus, the total number of edge updates over all repetitions of the algorithm is $O(\beta m)$.

Let's assume that the incremental or decremental emulator algorithm has total time $O(m'n'^{1-\epsilon}/\beta^2)=O(\beta m (\beta n)^{1-\epsilon}/\beta^2)=O(mn^{1-\epsilon})$.

Additionally, at most $n$ times during the algorithm, we run BFS on the emulator. The number of edges in the emulator is $O(m^{1-\epsilon}/\beta)=O( (\beta m)^{1-\epsilon}/\beta)=O(m^{1-\epsilon})$. Thus, the BFS calls take total time $O(nm^{1-\epsilon})=O(mn^{1-\epsilon})$.

Putting everything together, our incremental or decremental emulator algorithm implies an algorithm for directed $k$-cycle detection in time $O(mn^{1-\epsilon})$, thus refuting the combinatorial $k$-Clique hypothesis.

\section{Algebraic All Pairs Shortest Paths with Path Reporting}\label{sec:algebraic}
The main result of this section is a randomized, fully dynamic algorithm that can maintain and query successors for all pairs shortest paths of up to $D$ edges. Our algorithm is an augmentation of the algebraic all pairs shortest distances algorithm of Sankowski \cite{sankowski-thesis}, who originally posed as an open problem to use his techniques and the construction of \ref{lemma::sankadjoint} to actually report paths. We state our new result formally in the following theorem

\begin{theorem} 
[Successor Queries and Short Paths] \label{thm:DpathReporting} For any parameters $\eps\in (0, 1), D<n$, and an unweighted graph $G=(V, E)$ subject to edge insertions and deletions, there is a dynamic, randomized data-structure $\mathcal{P}^D_\eps$ that supports the following operations: \begin{itemize}
    \item Ins/Delete$(e)$ \textit{Inserts/Deletes edge $e\in E$ in worst case time $O\big(Dn^{\omega(1, 1, \eps)-\eps}+Dn^{1+\eps})$}.
    \item Short Distance/Successor Query$(i, j)$ \textit{Returns the distance $d \leq D$ and a successor on any short, shortest $i\rightarrow j$ path in worst case time $O(Dn^\eps)$} and is correct whp.
    \item Short Path Queries$(i, j)$ \textit{Returns a shortest $i\rightarrow j$ path of length $d \leq D$ by repeatedly finding successors in worst case $O(dDn^\eps)$ time, and is correct whp.}
\end{itemize}
with pre-processing time $O(Dn^2)$ on empty graphs and $O(Dn^\omega)$ otherwise.
\end{theorem}

\noindent as also presented in \ref{thm:restateDpathReporting}. We believe this result is of independent interest. We can minimize the update time by the choice of $\eps = \eps^*$ that balances the exponents in the runtime, that is, $\eps^*$ is the solution to
\begin{equation}
    \omega(1, 1, \eps)-\eps = 1+\eps
\end{equation}
which numerically can be evaluated to $\eps^*\approx .529$. The corresponding update time is $O(Dn^{1+\eps^*}) = O(Dn^{1.529})$

An overview of this section is as follows. We begin by listing a small toolkit of sparse matrix facts, to be used in the following subsections. In subsection \ref{subsec::sankreview}, we detail \cite{sankowski-thesis}'s original all pairs shortest distances construction, his path encoding lemma \ref{lemma::sankadjoint}
 and the dynamic matrix inverse algorithm \ref{thm::matrixinverse}, all central to our augmentations and key to our proof of \ref{thm:DpathReporting}. In subsection \ref{subsec::successor}, we build on the matrix inverse algorithm to show how to maintain the product of the adjacency matrix and the inverse, and then use these products to prove the theorem above. Finally, in subsection \ref{subsec::exactapsp}, we then use this result to construct the first subquadratic time update and path reporting algorithm for fully dynamic, unweighted APSP. 

\subsection{Preliminaries}

\begin{fact}
[Rectangular Matrix Multiplication]\label{fact::rmm} We denote as $O(n^{\omega(a, b, c)})$ the cost of multiplying two rectangular matrices, the first of size $n^a\times n^b$, the second $n^b\times n^c$. The current exponent of square matrix multiplication $\omega(1, 1, 1)\equiv \omega\leq 2.3729$.
\end{fact}

\begin{fact}
[Row Sparse Matrix Multiplication]\label{fact::rowsparse} The product of two matrices $A, B$, the first of $\leq n^\delta$ non-zero rows, and the second of $\leq n^\alpha$ non-zero rows, has at most $n^\delta+n^\alpha$ non-zero rows and takes time $O(n^{\omega(1, \alpha, \delta)})$ to compute.
\end{fact}

\begin{fact}
[Row Sparse Matrix Inverse]\label{fact::rowsparseinverse} The inverse of a matrix $A$ that differs from identity in at most $n^\delta$ rows, differs from identity in at most $n^\delta$ rows. Moreover, it can be computed in $n^{\omega(1, \delta, \delta)}$ time.
\end{fact}

\begin{fact}
[The Hitting Set Lemma]\label{fact::hittingset} Given a set $S$ of $n$ elements, a random sample $H\subset V$ of size $\geq c n\log n/k$ for a given constant $c>1$ hits every subset of size $k$ of $S$ with high probability. 
\end{fact}

\subsection{Algebraic All Pairs Distances} 
\label{subsec::sankreview}
In his PhD thesis, Sankowski \cite{sankowski-thesis} showed the following lemma on how to encode path lengths in a matrix. 
\begin{lemma}[\cite{sankowski-thesis}]\label{lemma::sankadjoint}
 Let $\Tilde{A}$ be the symbolic adjacency matrix of the graph $G=(V, E)$, where each edge $(i, j)\in E$ defines a variable $\Tilde{A}_{ij} = x_{ij}$, and $\Tilde{A}_{ij}=0$ if $(i, j)\notin E$. Consider the adjoint \text{adj}$(I-u\tilde{A})_{ij}$ as a polynomial over an additional variable $u$. The length of the shortest path in $G$ from $i$ to $j$ is equal to the degree of $u$ of the smallest degree non-zero term in \text{adj}$(I-u\tilde{A})_{ij}$. 
\end{lemma} 

Sankowski \cite{sankowski-thesis} used this result to construct algorithms for the All Pairs Shortest Distances problem by sampling a uniformly random integer in a field $\mathbb{F}$ of size $n^{O(1)}$ for each symbolic edge-variable and using the Schwartz-Zippel Lemma to guarantee that with high probability, for all $i,j\in V$ the degree $d_{ij}$ (the distance) term in \text{adj}$(I-u\tilde{A})_{ij}$ is non-zero - and thus whp it suffices to read the polynomial entry at $i, j$ to obtain the distance $d_{ij}$. Sankowski \cite{sankowski-thesis} then showed how to maintain and query this adjoint dynamically, and over a ring mod $u^{D+1}$ in worst case update time $O(D(n^{\omega(1,1,\eps)-\eps}+n^{1+\eps}))$ for a given parameter $\eps\in (0, 1)$. This effectively allowed the \textit{short distance queries}, that is, to return the distance between any pair of vertices $i,j$ correctly if the distance $d_{ij}$ is less than $D$, however introducing a tradeoff in runtime for large $D$. We state and re-prove these theorems here, for concreteness.

\begin{theorem}
[Dynamic Integer Matrix Inverse \cite{sankowski-thesis}]\label{thm::matrixinverse} 
 Given a constant $\eps\in (0, 1)$, there is a deterministic, dynamic data-structure $\mathcal{D}_\eps$ that maintains the inverse and the determinant of an integer matrix $M$ subject to non-singular entry-wise updates, and supports the following operations:
\begin{itemize}
    \item Update$(i, j, v)$ \textit{Updates entry $M_{ij}=v$ and the data-structure $\mathcal{D}_\eps$ in time $O(n^{\omega(1,1,\eps)-\eps}+n^{1+\eps})$}
    \item Query$(i, j)$ \textit{Returns the value of the inverse at entry $(i, j)$, $M^{-1}_{ij}$, in $O(n^\eps)$ time.}
\end{itemize}
\noindent For current $\omega < 2.3729$, the value of $\eps$ that balances the update time is $\eps^*\approx.529$. 
\end{theorem} 

The key idea is to write the inverse as $M^{-1} = T(\mathbb{I}+N)$ as the product of two matrices, where one of them, $N$, is \textit{sparse} and initially null. As we will show, each entry update to $M$ corresponds to a row update to $N$ (without updating $T$), and every $O(n^\eps)$ updates we exploit fast matrix multiplication over this sparse $N$ to "reset" $T\leftarrow T+NT, N\leftarrow 0$, guaranteeing the sparsity of $N$. \\
 
 \begin{proof} Let us denote $v_{ij} e_i e_j^T$ as the matrix corresponding to the \textit{additive update} to entry $ij$ of $M$ during some update. That is, $M'\leftarrow M+v_{ij}e_i e_j^T$. We construct $\mathcal{D}_\eps$ as follows. We maintain explicitly the matrices $T, N$ during the execution, where initially $N=0$ and we precompute $T = M^{-1}$ in $O(n^\omega)$ time. At every non-singular entry update, we first compute the update to $N$ through the following algorithm:
 \begin{itemize}
     \item Compute the row vector $b = v_{ij}e_j^TM^{-1} = v_{ij}(e_j^TT)(\mathbb{I}+N)$
     \item Compute the inverse $B=(\mathbb{I}+e_ib)^{-1}$
     \item Finally, update
     \begin{equation}
         N'\leftarrow NB+B-\mathbb{I}
     \end{equation}
 \end{itemize}
 Correctness of this update follows from plugging in the result into the definition of $M^{-1}$ to guarantee the invariant
 \begin{equation}
     T(\mathbb{I}+N')=T(NB+B) =M^{-1}B = (M+e_ibM)^{-1} =(M+e_iv_{ij}e_j^T)^{-1}  = (M')^{-1}
 \end{equation}
 as intended. Note that by Fact \ref{fact::rowsparseinverse}, as $(\mathbb{I}+e_ib)$ has exactly one non-identity row, so does its inverse $B = (\mathbb{I}+e_ib)^{-1}$. By Fact \ref{fact::rowsparse} this implies that at each update $N$ gains at most one additional non-zero row, and thereby \textit{after $k$ updates $N$ has at most $k$ non-zero rows}. Finally, after $O(n^\eps)$ updates, we reset the inverse by computing 
 
 \begin{itemize}
     \item After every $n^\eps$ updates, $T\leftarrow T(1+N)$, $N\leftarrow 0$
 \end{itemize}
 
 Since $N$ has at most $O(n^\eps)$ non-zero rows, computing this product takes time $O(n^{\omega(1, 1, \eps)})$. Note that it takes time $O(n^{1+\eps})$ to compute the row vector $b$ given the product of the row vector $e_j^T T$ and the row-sparse $N$, s.t. the average update time is 
 \begin{equation}
     O\big(n^{1+\eps}+n^{\omega(1, 1,\eps)-\eps}\big)
 \end{equation}
 which for current $\omega\approx 2.3729$ is minimized at $\eps\approx .529$, for a runtime of $O(n^{1.529})$. Finally, to query any entry $ij$ of $M^{-1}$, we compute the dot product of the row vector $(e_i^TT)$ and the sparse column vector $(Ne_j)$ in time $O(n^\eps)$:
 \begin{equation}
     e_i^TM^{-1}e_j = e_i^T T(\mathbb{I}+N)e_j =  e_i^TTe_j+ (e_i^TT)(Ne_j)
 \end{equation}
 To conclude this proof, we note that the determinant det$M$ is easily maintained in $O(1)$ by definition of the matrix $B$ above. At an update to entry $(i, j)$ of $M$, $M'=B^{-1}M\Rightarrow \text{det}M' = \text{det}M\cdot (1+b_i)$. This allows us to support queries to the adjoint of $M$.
 \end{proof}
%%%%%%%%%%%%%%%%%%%%%%%%%%%%%%%%% Sank proof ends

 Sankowski then showed how to support the above dynamic matrix inverse algorithm over polynomial matrices. We present the result in the corollary below 
 
 \begin{corollary}
 [Dynamic Polynomial Matrix Inverse \cite{sankowski-thesis}]  \label{cor::polyminverse}
 Given $\eps \in (0, 1)$, $D\in [n]$ there is a dynamic data-structure $\mathcal{D}_\eps^D$ that can maintain the inverse $M^{-1}$ of polynomial matrix $M=\mathbb{I}-A$ where $M, A\in  (\mathbb{F}[X]/\langle X^{D+1}\rangle)^{n\times n}$ subject to entry-wise polynomial updates to $M$, incurring a multiplicative cost of $\tilde{O}(D)$ to the runtimes of the theorem above.
 \end{corollary}

The proof of the corollary above arises from Sankowski's extension of Strassen's idea of computing over the formal power series to the dynamic case. We note that over the formal power series mod $X^D$, $M=\mathbb{I}-A\in  (\mathbb{F}[X]/\langle X^{k+1}\rangle)^{n\times n}$ is always invertible as
\begin{equation}
    M^{-1} = \frac{1}{\mathbb{I}-A} = \sum_{i=0}^k A^i = \prod_{i=1}^{\lceil\log D\rceil}(\mathbb{I}+A^{2^i})
\end{equation}
and moreover can be computed/preprocessed in $\tilde{O}(Dn^\omega)$. The details are in \cite{sankowski-thesis}. Corollary \ref{cor::polyminverse} above, together with the path encoding lemma \ref{lemma::sankadjoint} define the fully dynamic data-structure $\mathcal{D}_\eps^D$ that can support edge updates and short distance queries, which we base this section off of and later augment. For concreteness, we state this result in the following theorem:

\begin{theorem}
 [Short Distance Queries \cite{sankowski-thesis}]\label{thm::dqueries}
 \textit{Given a unweighted, dynamic graph $G$ subject to edge insertions and deletions, and parameters $\eps \in (0, 1)$ and $1< D< n$, there exists a dynamic data-structure $\mathcal{D}_\eps^D$ that supports edge updates in worst case $O(D(n^{\omega(1,1,\eps)-\eps}+n^{1+\eps}))$ time and can query any distance $d\leq D$ correctly whp in worst case time $O(Dn^\eps)$. If otherwise $d> D$, the distance query outputs a failure.}
\end{theorem}

Which follows simply by defining $M=\mathbb{I}-uA$ in Corollary \ref{cor::polyminverse}, where $A$ is the symbolic adjacency matrix with random integer values whenever non-zero. Additionally, $\mathcal{D}_\eps^D$ can also maintain the $H\times H$ submatrix of the inverse, corresponding to the up to $D$ distances between all pairs of vertices in a subset of size $|H|=n^q$ vertices in $O(Dn^{2q})$ time, as follows: 

\begin{corollary}
[Submatrix Maintenance \cite{sankowski-thesis}]\label{cor::submaintenance}
 \textit{Given a subset $H\subset [n]$ of size $n^q$ of the column/row indices, $\mathcal{D}_\eps^D$ can maintain the $H\times H$ submatrix $M^{-1}_{H, H}$ of the inverse in additional $O(Dn^{2q})$ update time, allowing queries in $O(1)$ to the submatrix.}
\end{corollary} 

\begin{proof} To explicitly maintain $M^{-1}_{H, H}$, it suffices to show how to efficiently perform updates. In the notation of the proof of the theorem above,

\begin{equation}
    M^{-1}_{H, H} \leftarrow (M^{-1}B)_{H, H} = M^{-1}_{H, H}+M^{-1}_{H, \{i\}}(B-\mathbb{I})_{\{i\}, H}
\end{equation}

\noindent where the update to the inverse is expressed in terms of the matrix $B$,  which only differs from identity at a single row $i$, and thereby can be quickly computed in $O(Dn^{2q})$.
\end{proof}

\subsection{Successor and Short Path Queries} \label{subsec::successor}
We dedicate this section to the proof of Theorem \ref{thm:DpathReporting}. The key idea in maintaining and querying the successors in short, shortest paths is inspired by Seidel's static algorithm for undirected, unweighted APSP \cite{seidelapsp}. In order to augment the data-structure of Theorem \ref{thm::dqueries} \cite{sankowski-thesis}, we additionally maintain the product $A \cdot \text{adj}M$ of the boolean adjacency matrix and the adjoint of the polynomial matrix $M=\mathbb{I}-uA$ as defined in Lemma \ref{lemma::sankadjoint}. The $(i, j)$ entry of the adjoint $(\text{adj}M)_{ij}$ is a polynomial where the degree of the lowest degree non-zero monomial is the length of the shortest path, and if the distance from $i$ to $j\in V$ is $d\leq D$, then we can query it in time $O(Dn^\eps)$ under Theorem 
\ref{thm::dqueries}. Moreover, inspection of the product tells us that $(A\cdot \text{adj}M)_{ij}$ must have smallest degree $d-1$ as there must exist some witness $s$ (the successor!) s.t. $A_{is}=1$ and $M^{-1}_{sj}$ has minimum degree $d-1$. This idea of maintaining successors as witnesses of a polynomial matrix product in addition to a sparsification argument that allows us to reduce the case of a multiple witnesses (successors) to that of a single witness (successor), enables a bitwise selection trick that follows closely to Seidel's successor finding algorithm \cite{seidelapsp} in the static case.

The outline of this section is as follows. First, we describe how to maintain the product $A\cdot \text{adj}M$. Note again that $A\cdot \text{adj}M = \text{det}M \times A\cdot M^{-1}$, and thus we can instead maintain $A\cdot M^{-1}$, and $\text{det}M$ following \ref{thm::matrixinverse}. Then, we show how to find a successor on a given path $i\rightarrow j$ of length $\leq D$, if the successor is unique, using $\log n$ queries to the products described. Finally, we review Seidel's sparsification trick \cite{seidelapsp} to reduce the case of multiple witnesses to a polylog number of single witness queries, and we conclude with our main theorem on the short path finding algorithm $\mathcal{P}_\eps^D$. In the next subsection, we use this path finding black box as a subroutine to construct novel algorithms for exact dynamic APSP with path reporting. We begin by presenting our theorem on maintaining the products. 
\begin{theorem}[Dynamic Product Maintenance]\label{thm::witnessedmm}
Let $\eps\in (0, 1)$ be a constant, and define the dynamic polynomial matrices $E, A, M\in(\mathbb{F}[X]/\langle X^{D+1}\rangle)^{n\times n}$, where $M\equiv\mathbb{I}-A$, and both $E$ and $A$ are subject to entry-wise updates. There is a data-structure that supports the following operations over the product $E\cdot M^{-1}$:
\begin{itemize}
    \item UpdateE$(i, j, v)$ \textit{Updates the entry $E_{ij}\leftarrow v$ for $i, j\in [n]$ and $v\in (\mathbb{F}[X]/\langle X^{D+1}\rangle)$ in worst case time $O(Dn)$.}
    \item UpdateA$(i, j, v)$ \textit{Updates the entry $A_{ij}\leftarrow v$ for $i, j\in [n]$ and $v\in (\mathbb{F}[X]/\langle X^{D+1}\rangle)$ in worst case time $O\big(D(n^{\omega(1,1,\eps)-\eps}+n^{1+\eps})\big)$}
    \item Query$(i, j)$ \textit{Returns the entry $(E\cdot M^{-1})_{ij}\in (\mathbb{F}[X]/\langle X^{D+1}\rangle)$ in time $O(Dn^\eps)$}.
\end{itemize}
\end{theorem} 

\begin{proof} The key idea in this proof is to use Theorem \ref{thm::matrixinverse} and maintain matrices $T, N$ s.t. $M^{-1}=T(\mathbb{I}+N)$, as well as explicitly maintaining a matrix $V\equiv ET$. Using associativity this avoids an intermediary matrix product during queries, and as we will show can be efficiently maintained subject to updates to $E$ and $A$ under slight modifications to the original scheme.\\ 

During the execution we maintain explicitly the polynomial matrices $E, T, N$ and the product $V=ET$, such that we maintain the correctness invariant $EM^{-1} = (ET)(\mathbb{I}+N) = V(\mathbb{I}+N)$ by associativity. At every entry update to $E\leftarrow E+v_{ij}e_ie_j^T$, we simply compute the row vector $v_{ij}e_j^TT$, corresponding to the update to the $i$th row $V\leftarrow V+e_i(v_{ij}e_j^TT)$ in $O(Dn)$ time. Correctness of this update follows from associativity in the expansion
\begin{equation}
   (E+v_{ij}e_ie_j^T)M^{-1} = (ET+v_{ij}e_ie_j^TT)(\mathbb{I}+N) = (V+e_i(v_{ij}e_j^TT))(\mathbb{I}+N)
\end{equation}
showing that it suffices to update the product $V$. Next, whenever we update $A\leftarrow A+v_{ij}e_ie_j^T$, we simply follow Theorem \ref{thm::matrixinverse} and update the sparse matrix $N$, and every $n^\eps$ updates to $A$ we \textit{reset} $V\leftarrow V\cdot(1+N)$, then $T\leftarrow T\cdot(1+N)$, finally $N\leftarrow 0$, in $O(Dn^{\omega(1, 1, \eps)-\eps}+Dn^{1+\eps})$ amortized time. Note that this maintains the invariant that $EM^{-1} = V(\mathbb{I}+N)$ as well as $M^{-1}=T(\mathbb{I}+N)$. To conclude this proof, we support queries to $(E M^{-1})_{ij}$ by computing 
\begin{equation}
    e_i^T E M^{-1} e_j = e_i^T V e_j + (e_i^T V) (Ne_j)
\end{equation}
 in $\tilde{O}(Dn^\eps)$ time due to the dot product of the vectors $(e_i^T V)$ and $(Ne_j)$, where $(Ne_j)$ is a column vector with a sparse number of non-zero rows, as detailed in the proof of \ref{thm::matrixinverse} in the previous subsection. Note that maintaining $V$ explicitly was key to these fast queries.
\end{proof}
%%%%%%%%%% end of successor proof

Effectively, this implies that we can efficiently maintain and query the product above $E\cdot M^{-1}$, for any matrix $E$, in the same time bounds as Sankowski's original dynamic polynomial matrix inverse construction $\mathcal{D}_\eps^D$ of Corollary \ref{cor::polyminverse}. We now show how to use this construction to find single witnesses/single successors.
\begin{lemma}\label{lemma::singlesuccessor}
Let $i, j\in V$ be vertices s.t. their distance is $1<d_{ij}\leq D$. If there is a single successor $s\in V = [n]$ on any shortest path from $i$ to $j$, then we can find it using $O(\log n)$ queries to the data-structure of the theorem above, in $\tilde{O}(Dn^\eps)$ time. 
\end{lemma}

\begin{proof} For each $l\in[O(\log n)]$ we define a subset of the column indices $S_l\subseteq [n]$ where each index $s\in [n]$ is in $S_l$ if the $l$th bit of $s$'s bitwise description is $s_l=1$. Define $O(\log n)$ copies of the adjacency matrix $A$, where $A^{(l)}$ for $l\in[O(\log n)]$ is defined by only picking the subset of columns $S_l$, and all other columns of $A^{(l)}$ are completely null. We initialize and  maintain $O(\log n)$ copies of the data-structure of Theorem \ref{thm::witnessedmm}, where $E^{(l)}=A^{(l)}$. As described, for a given pair of vertices $i, j\in V$ if $1<d_{ij}\leq D$ then there must exist some successor $s\in [n]$ of $i$ s.t. $A_{is}=1$,  adj$M_{sj}$ has minimum degree $d_{sj}=d_{ij}-1$, and thereby the product $(A\cdot\text{adj}M)_{ij}$ has minimum degree $d_{sj} = d_{ij}-1$. However, as by assumption there is a single witness $s$, all other $s'\in [n]\setminus \{s\}$ have minimum degree of $\text{adj}M_{s'j}> d_{sj}$, and thus any copy of the data-structure $(A^{(l)}\cdot \text{adj}M)_{ij}$ where $s$ is not one of the selected columns will have minimum degree $>d_{ij}-1$. This is simply since $A_{is}^{(l)} = 0$ if the $l$th bit $s_l=0$. It follows that the minimum degree of the $l$th product at entry $(i, j)$ is $d_{ij}-1$ if the $l$th bit of $s$ is 1, and so if we query the $O(\log n)$ data-structures $(A^{(l)}\cdot \text{adj}M)_{ij}$ for each $l$, and define a sequence of bits $b_l = 1$ if  $(A^{(l)}\cdot \text{adj}M)_{ij}$ has minimum degree $d_{ij}-1$, then $(b_1, b_2\cdots  b_{O(\log n)})\equiv  s$ is exactly the bitwise description of the single witness $s$.\end{proof}

A straightforward sparsification trick allows us to reduce the case of multiple witnesses to that of a single witness.

\begin{lemma}\label{lemma::multiplesuccessor}
If there is an arbitrary number of distinct successors of $i\in V$ on any shortest path from $i$ to a given $j\in V$ of length $\leq D$, then we can find an arbitrary one of them in in $\tilde{O}(Dn^\eps)$ time.
\end{lemma} 

\begin{proof} The key construction is to maintain $O(\log^2 n)$ versions of the single-witness data-structure of Lemma \ref{lemma::singlesuccessor}. We first address the case of a known number of witnesses $c\geq 1$ for a given product $ij$ via a sparsification trick, and then show how to efficiently "guess" the number of witnesses for each product in only an additional factor of $\log n$.

If the number of witnesses $c$ of the product $(AM^{-1})_{ij}$ is known to be bounded in a range by $n/2^{w+1}\leq c\leq n/2^{w}$ for some $w = O(\log n)$, then sampling $2^w$ columns $C\subset [n]$ uniformly at random gives us a constant probability that we sample only a single witness $s\in C$. This is since
\begin{equation}
    \mathbb{P}[|W\cap C| = 1] = 2^w\frac{c}{n} (1-\frac{c}{n})^{2^w-1}\geq \frac{1}{2}(1-\frac{1}{2^{w}})^{2^w-1}\geq 1/{2e}
\end{equation}

\noindent where $W$ is the set of witnesses of the product $ij$. It follows that if we maintain $O(\log n)$ versions of the data-structure of lemma \ref{lemma::singlesuccessor}, where in each copy we sample $2^w$ columns uniformly at random and 0-out the remaining columns, then with high probability for every pair $i, j$ with $n/2^{w+1}\leq c\leq n/2^{w}$ successors there exists a copy of the data-structure with a single witness of $i$ on a shortest path to $j$.

Finally, to address the fact that the number of witnesses is unknown and can vary for each pair, we maintain $O(\log n)$ editions of the sparsification above for every $w\in \{1\cdots O(\log n)\}$, that is, for every possible bound over the number of witnesses. In this manner, whenever we query the successor for a given pair $i, j$, we query each of the $O(\log^3 n)$ data-structures of theorem \ref{thm::witnessedmm} for the entry $i, j$, and obtain a set $P$ of $O(\log^2 n)$ \textit{potential} successors ($\log n$ column samplings per witness-exponent guess). Then, we query the original distance matrix adj$M_{sj}$ for each $s\in P$, and output any one of them with distance $d_{sj}=d_{ij}-1$ (a successor) and with an edge $A_{is}=1$. By construction, we find a successor whp as we check every possible range of the number of successors (parametrized by $w$), and we always verify the output solution. This takes overall $\tilde{O}(Dn^\eps)$ to find a successor on any shortest path $i,j$ of length $\leq D$.\end{proof}

Now that we have shown how to support successor queries, we formalize our statement on short path finding.
\begin{theorem} 
[Successor Queries and Short Paths]\label{thm:restateDpathReporting} For any parameters $\eps\in (0, 1), D<n$, and an unweighted graph $G=(V, E)$ subject to edge insertions and deletions, there is a dynamic, randomized data-structure $\mathcal{P}^D_\eps$ that supports the following operations: \begin{itemize}
    \item Ins/Delete$(e)$ \textit{Inserts/Deletes edge $e\in E$ in worst case time $O\big(Dn^{\omega(1, 1, \eps)-\eps}+Dn^{1+\eps})$}.
    \item Short Distance/Successor Query$(i, j)$ \textit{Returns the distance $d \leq D$ and a successor on any short, shortest $i\rightarrow j$ path in worst case time $O(Dn^\eps)$} and is correct whp.
    \item Short Path Queries$(i, j)$ \textit{Returns a shortest $i\rightarrow j$ path of length $d \leq D$ by repeatedly finding successors in worst case $O(dDn^\eps)$ time, and is correct whp.}
\end{itemize}
with pre-processing time $O(Dn^2)$ on empty graphs and $O(Dn^\omega)$ otherwise.
\end{theorem}

\begin{proof} We maintain the polynomial matrix inverse data-structure $\mathcal{D}_\eps^D$ of Theorem \ref{thm::dqueries}, in addition to the data-structure of Lemma \ref{lemma::multiplesuccessor}. At each edge update, we update each of the polylog copies of the data-structure of Theorem \ref{thm::witnessedmm}, including at most a single entry update to $E$ and to $A$ per copy, per update. To support Short Path Queries of a given pair of vertices $i, j$ of distance $d\leq D$, we simply sequentially query the successor $s$ of $i$, and recurse on $i\leftarrow s$ while maintaining the path. This takes time $\tilde{O}(d\times Dn^\eps)$.\end{proof}

In the next sections, we apply $\mathcal{P}_\eps^D$ as a black box to the problems of Exact and Approximate Dynamic APSP. 

\subsection{Exact Fully Dynamic Unweighted APSP} \label{subsec::exactapsp}
 In this section we explain how to use our short path finding black box $\mathcal{P}_\eps^D$ of Theorem \ref{thm:DpathReporting} to construct the first subquadratic time update and path query algorithm for Exact Dynamic APSP. $\mathcal{P}_\eps^D$ allows us to query short length $d\leq D$ paths in $O(dDn^\eps)$ time, such that it only remains to show how to construct long, up to linear length paths. The construction follows that of the all pairs shortest distances algorithm of \cite{sankowski-thesis}. We can use the \textit{path decomposition technique} to sample a hitting set $H\subset V$ of size $\tilde{O}(n/D)$ that whp hits every length $D/2$ path in $G$ by Fact \ref{fact::hittingset}. In this manner, any $i$ to $j$ shortest path can be decomposed into consecutive subpaths $l_1\cdots l_k\cdots l_{O(n/D)}$ of $D/2$ vertices, s.t. whp there is a vertex $u_k\in H$ in the hitting set in each one of those subpaths $u_k\in l_k\forall k$. It follows that there exists a subpath decomposition $i\rightarrow u_1, u_1\rightarrow u_2\cdots \rightarrow j$ of the shortest path from $i$ to $j$ where whp the distance between each adjacent pair $(u_k,u_{k+1})$ is "short", that is, less than $D$, and can be queried efficiently. 
 
 This effectively reduces the question to finding the hitting set vertices in a path from arbitrary $i$ to $j.$ At every edge update, we use Corollary \ref{cor::submaintenance} to explicitly maintain the $H\times H$ submatrix of the inverse in $O(n^2/D)$ time, s.t. we can access the $H\times H$ matrix of distances $D^D_{H,H}$ between any pair of vertices in $H$ if they are $\leq D$ apart in $O(1)$ time. We can use these current distances to construct the induced, weighted graph $G'=(H, E', D^D_{H, H})$, and compute the APSP via Floyd-Warshall's algorithm on the induced $G'$ in $O(|H|^3) = O(n^3/D^3)$ time. The key idea in constructing long paths is that the hitting set lemma guarantees that on any long shortest path, each pair of hitting set vertices on said path are at most $D$ nodes away whp and thus between vertices in $H$ the shortest distance in $G'$ is also the shortest distance in $G$. Floyd-Warshall's algorithm therefore not only with high probability gives us the correct distances $D_{H,H}$ between any $p, q\in H$, as it also maintains the paths $p, u_1\cdots, u_k,\cdots q$ of hitting set vertices in $G'$ between any $p, q\in H$. 
 
 At each path query between arbitrary $i, j\in V$, we first check if the distance is short $d_{ij}\leq D$ using $\mathcal{P}_\eps^D$ of theorem \ref{thm:DpathReporting}, and if not, we compute the distance between $i$ and $j$ by performing the following minimization over the last and first hitting set vertices $p, q$ in the path $i\rightarrow p\cdots q\rightarrow j$:
 \begin{equation}
     D_{ij} = \min\{D^D_{ij}, \min_{p, q\in H} \big(D^D_{ip}+D_{H, H}(p, q)+D^D_{qj}\big)\}
 \end{equation}
 
\noindent where the terms in the inner minimization correspond to the distances $i\rightarrow p \in G$, $p\rightsquigarrow q \in G'$ and $q\rightarrow j \in G$ respectively. We note that it takes time $O(|H|^2) = \tilde{O}(n^2/D^2)$ to perform the minimization, which requires both a row and a column query to $\mathcal{D}^D_\eps$, in particular to the distances $(i, H)$ and $(H, j)$, s.t. the total runtime of the distance query is $O(|H|^2 + |H|Dn^{\eps}) = \tilde{O}(n^2/D^2+n^{1+\eps})$.

 Note additionally that finding the first and last hitting set vertices $p, q$ in a $i, j$ path also gives us \textit{all} the hitting set vertices $p, u_1, u_2\cdots q\in H$ on the path, due to the successor information in $G'$. It follows now that whp each pair of adjacent $u_k, u_{k+1}$ are less than $2\times D/2=D$ nodes apart, and thus we can perform short path queries between all the adjacent pairs $u_k, u_{k+1}\in H$ using $\mathcal{P}_\eps^D$ in $\tilde{O}(Dn^\eps\sum_k d_{u_k, u_{k+1}}) = \tilde{O}(Dn^{1+\eps})$ time; this is equivalent to the runtime of $O(n)$ successor queries. Overall, the long path query takes time 
\begin{equation}
    \tilde{O}\bigg(Dn^{1+\eps}+n^2/D^2\bigg)
\end{equation}

Putting all the ingredients together gives us a worst case update time of 
\begin{equation}
     O\big(Dn^{\omega(1, 1, \eps)-\eps}+Dn^{1+\eps}+n^3/D^3+n^2/D\big)
\end{equation}
due to the update time of theorem $\ref{thm:DpathReporting}$, the Floyd-Warshall step, and corollary \ref{cor::submaintenance} respectively. Finally, the initialization time is dominated by that of initializing Theorem $\ref{thm:DpathReporting}$. We conclude this description with a formal statement on the main result of this subsection:

\begin{theorem}
[Subquadratic, Unweighted APSP] \textit{For any $\eps\in (0, 1), D\in [n]$, and an unweighted graph $G=(V, E)$ subject to edge insertions and deletions, there is a dynamic, randomized data-structure that supports the following operations:}
\begin{itemize}
    \item Ins/Del$(e)$ \textit{Inserts/Deletes edge $e\in E$ in worst case time}
    \begin{equation}
        O\big(Dn^{\omega(1, 1, \eps)-\eps}+Dn^{1+\eps}+n^3/D^3+n^2/D\big)
    \end{equation}
    \item Distance Query$(i, j)$ \textit{Returns the distance from $i$ to $j$ whp, in worst case time}
    \begin{equation}
        O(n^2/D^2+n^{1+\eps})
    \end{equation}
    \item Path Query$(i, j)$ \textit{Returns a shortest path from $i$ to $j$ whp, in worst case time}
    \begin{equation}
        O(n^2/D^2+Dn^{1+\eps})
    \end{equation}
\end{itemize}
\end{theorem} 

The values of $D$ and $\eps$ above establish a trade-off between update and path query runtimes, and we allot the remainder of this section for a discussion on the sub and super quadratic algorithms we can construct with this trade-off. Three cases worth mentioning are that of minimal update time, minimal path query time, and minimal path query subject to subquadratic update time.

\textit{Minimal Update Time} If we optimize first over $\eps$ in the update time, picking the optimum $\eps = \eps^*\approx .529$, and then optimizing over $D$ s.t. $D=n^\mu$ and $\mu = \frac{2-\eps^*}{4}\approx .368$, then we obtain the first fully dynamic exact APSP algorithm with subquadratic time update and path query, in $O(n^{1.897})$ time for both operations. 

\textit{Minimal Path Query Time} we can establish super quadratic update time but efficient path query algorithms by first optimizing over $D$ in the path query expression. The path query runtime is optimized if we pick $D=n^{\frac{1-\eps}{3}}$, s.t. the runtime is $O(n^{\frac{4+2\eps}{3}})$. The update time for any $0<\eps\leq .313$ is now $O(n^2(n^\eps+n^{\frac{1-4\eps}{3}}))$, which itself establishes a trade-off over $\eps$, and allows for $O(n^{2.332})$ update time and $O(n^{1.334})$ path query.

\textit{Minimal Path Query subject to Subquadratic Updates} We can also establish \textit{slightly} subquadratic update time algorithms, but with efficient path query times. It suffices to first pick the minimum $D$ that allows for a subquadratic update time, $D=n^{1/3+\delta/3}\forall \delta>0$ s.t. updates are performed in $O(n^{1/3+\delta/3+\omega(1, 1, \eps)-\eps}+n^{2-\delta})$, and then finding the minimum $\eps$ s.t. the updates are still subquadratic. By numerical optimization we find $\eps'\approx .334$, s.t. the runtime of path query is $O(n^{1.667})$, while the update time is $O(n^{1.999})$.

\section{Algebraic fully dynamic spanner algorithm}
\label{subsec:AlgebraicSpanner}
%While the $m^{1+o(1)}$ bound obtained in the last section to maintain a spanner in a fully dynamic algorithm is near-optimal in terms of $m$ and tight for combinatorial algorithms, we show in this section that we can improve this bound for algebraic algorithms. The main result of this section is stated below.

The goal of this section is to prove the following theorem.

\begin{theorem}
\label{thm:AlgebraicSpanner}
For any constant $0<\epsilon \le 1$, given an undirected, unweighted fully dynamic graph, 
there is an algorithm to maintain a $(1+\epsilon, n^{o(1)})$-spanner of size $n^{1+o(1)}$ with preprocessing time $\tilde{O}(n^{\omega})$ (or $\tilde{O}(n^2)$ if the input graph is empty) and worst update time $n^{1+\kappa^*+o(1)} = O(n^{1.529})$ with high probability against an oblivious adversary.
\end{theorem}

\paragraph{Internal Data Structures.} In order to obtain this result, we use two internal data structures. We first use a data structure $\mathcal{A}$ that maintains the distance matrix and corresponding shortest paths in a graph $G$, restricted to a set of sources $S$. The algorithm is randomized and uses fast matrix multiplication internally. It is a corollary of Theorem~\ref{thm:DpathReporting}, and follows simply by querying the data structure of Theorem~\ref{thm:DpathReporting} for every pairwise distance in $S$.

\begin{restatable}{corollary}{corSankowskiPathReporting}
\label{cor:SankowskiPathReporting}
Let $\eps$ be such that $0<\eps\leq \eps_{*}< 0.529$ and let $D$ be a distance parameter in $[1,n^\eps]$.
Suppose we are given a set of vertices $S$. Then the data structure from Theorem~\ref{thm:DpathReporting} can maintain for any arbitrarily small $\delta>0$, for all pairs of nodes $s,t\in S$, a shortest path between $s$ and $t$ as long as $\dist(s,t)\leq D$ (and can check if $\dist(s,t)>D$), with high probability against an oblivious adversary, with initialization time $\tilde{O}(Dn^\omega)$ (for a nonempty initial graph) or $\tilde{O}(n^2)$ (for an empty initial graph) and worst-case update time $\tilde{O}(Dn^{\omega(1,1,\eps)-\eps}+|S|^2 D^2n^{\eps}).$

If the depth threshold $D=n^{o(1)}$, the size of the subset $|S| = O(\sqrt{n})$ and we pick the parameter $\eps = \eps^*\approx .529$ to minimize the update time as in \ref{thm:DpathReporting}, the update time becomes $n^{1+\eps^*+o(1)} = O(n^{1.529})$.
\end{restatable}

Further, we use an algorithm $\mathcal{B}$ to maintain a spanner of $G$ with high multiplicative error efficiently.

\begin{theorem}[see \cite{forster2019dynamic}, Theorem 1.4]
Given an unweighted, undirected fully dynamic graph, there exists an algorithm $\mathcal{B}$ that maintains a spanner with multiplicative stretch $\log n$ and expected size $O(n \log n)$ that has expected update time $O(\log^3 n)$ against an oblivious adversary.
\end{theorem}

\paragraph{The Algorithm.} Equipped with these two powerful data structures, let us state the algorithm that gives \Cref{thm:AlgebraicSpanner}. Let $k = \sqrt{\log n}$, let 
$\epsilon' = \frac{\epsilon}{20(k+1)}$, and let $\base = \left(\frac{\log n}{\epsilon'}\right)$. We sample sets $V = A_0 \supseteq A_1 , \dots \supseteq A_k \supseteq A_{k+1} = \emptyset$ where $A_i$ for $i \in [1,k]$ is obtained by sampling each vertex in $V$ with probability $n^{-i/k} \log n$ (and to make the sets nesting add it to all $A_j$ where $j \leq i$). 
%The algorithm composes the spanner by adding edges from two different sources. To do so 
We maintain two data structures dynamically during the sequence of edge insertions and deletions:
\begin{enumerate}
    \item For $\gamma = \lfloor \eps \cdot k\rfloor$, we run $\mathcal{A}$ on graph $G$ with fixed source set $A_{\gamma}$ and depth threshold $\frac{1}{8\log n}\base^{k+1}$. 
    \item Further, we run $\mathcal{B}$ on $G$ and let $\tilde{G}$ be the $\log n$-approximate spanner.
\end{enumerate}

\smallskip\noindent
Let $\sumC{i}{i+1}{j}  = \sum_{y = i+1}^{j} \base^{y}$.
Initially, and after an edge update we construct the spanner $H$ from scratch as follows: We say that $a \in A_{\ell} \setminus A_{\ell+1}$  is \emph{active} if for no $j > \ell$ there exists a vertex $a' \in A_j \setminus A_{j+1}$ with  $\dist_{\tilde{G}}(a,a') \leq \sumC{\ell}{\ell+1}{j} /4$.
Note that we are using distances
in $\tilde{G}$ and {\em not} in $G$ for this definition. Note that all vertices in $A_k$ are active at all times. 
In order to determine which vertices are active, we run the following process. We begin by labeling all vertices as active. 
We then run for each $i$ from $k$ down to $0$, a BFS algorithm on the spanner $\tilde{G}$ to depth $\base^{i}$ from every vertex in $A_{i} \setminus A_{i+1}$ that is still labelled active. We then deactivate each vertex $v$ in $V \setminus A_i$ for which we found a vertex $a \in A_{i} \setminus A_{i+1}$ that is close enough to establish that $v$ cannot be active. 
We then construct our spanner $H$ which is initially empty by adding edges from two sources:

\begin{enumerate}
    \item We add all edges from $\tilde{G}$ to $H$.
    \item For any two vertices $a, a' \in A_{i} \setminus A_{i+1}$ that are active, we add the shortest path $\pi_{a,a',G}$ to $H$ if $\dist_{G}(a,a') \leq \frac{1}{8\log n} \base^{i+1}$, that is, if their distance in $G$ is small. To compute these paths $\pi_{a,a',G}$, we distinguish two cases. If $i \geq \gamma$, then we pose a path reporting query to  $\mathcal{A}$. For $i < \gamma$, we run from every such active vertex a BFS on $G$ to depth $\frac{1}{8\log n}\base^{i+1}$.
\end{enumerate}

As we will show, the approximation factor of the spanner always holds and the sparsity holds with high probability. If the spanner becomes too dense, we reinitialize the algorithm.

\paragraph{Spanner Approximation.} In the following, we prove that $H$ indeed forms a $(1+\epsilon, n^{o(1)})$ spanner. The basic idea behind the proof is the following: let $s$ and $t$ be two vertices and we want to analyze $\dist_H(s,t)$ in comparison to $\dist_G(s,t)$. There are basically two cases: If $s$ 
is ``close'' to an active vertex $a$ in $A_{i} \setminus A_{i+1}$, i.e. at distance at most $d_i$ for some $d_i$, and there is a vertex $v$ that is at distance $\sim 4 d_i \log n/ \epsilon$ from $s$ on the path from $s$ to $t$, such that $v$ is ``close'' to an active vertex $a'$ in $A_{i} \setminus A_{i+1}$, then we can use the spanner $\tilde{G}$ to get from $s$ to $a$ and from $a'$ to $t$ and have that the shortest path between $a$ and $a'$ is in $H$ by part (2) of its construction. It is not hard to see that this detour only implies a $(1+\epsilon)$-multiplicative error. Otherwise, we do not have a vertex within distance $\sim 4 d_i \log n/ \epsilon$ that is close to any active vertex in $A_{i} \setminus A_{i+1}$. We can then repeat the same argument for level $i-1$ along the path segment from $s$ of length $\sim 4 d_i \log n/ \epsilon$ and we are ensured that we eventually reach a level where vertices are active and where we get a good approximation. This allows us to subsume the additive error from higher levels into  multiplicative error for a series of segments of lower levels.

%Thus, we can search for close vertices on lower levels and are guaranteed to find one eventually, since level $i$ vertices are too far to deactivate them.

To formalize this concept, fix a value $i \in [1,k]$ and
let us say a vertex
%$j$-\emph{far} if there is no vertex $a' \in A_j \setminus A_{j+1}$ for $j > i$ with $\dist_{\tilde{G}}(a,a') \leq \base^{j} - \frac{1}{2} c_{i,j}$. Observe that the distance requirement is formulated with regard to the multiplicative spanner $\tilde{G}$. 
%Further, let us say a vertex 
$a$ is $\geq\!i$-\emph{far} if
(1) 
$a \in A_{\ell} \setminus A_{\ell+1}$ for some $\ell \in[0,i-1]$ 
and
(2)  for no $j \ge i$
there is a vertex $a' \in A_j \setminus A_{j+1}$ with $\dist_{\tilde{G}}(a,a') \leq \base^{j} - \frac{1}{2} \sumC{\ell}{\ell+1}{i} $. Observe that the distance requirement is formulated with regard to the multiplicative spanner $\tilde{G}$. 
Note that as $A_{k+1} = \emptyset$, every vertex is trivially $\geq\!(k+1)$-far.

\begin{lemma}\label{c:1}
%Any vertex that is $\geq\!i$-far belongs to   $A_{0} \setminus A_i$ and 
For any $\ell \in [0,k-1]$ every vertex
$a \in A_{\ell} \setminus A_{\ell + 1}$ 
that is $\geq\!(\ell + 1)$-far is active. 
\end{lemma}
\begin{proof}
Let $a$ be $\geq\!(\ell + 1)$-far. It follows that for no $j \geq \ell + 1$ there exists a vertex $a' \in A_j \setminus A_{j+1}$ with $\dist_{\tilde{G}}(a,a') \leq \base^{j} - \frac{1}{2} \sumC{\ell}{\ell+1}{\ell+1}$.
Note that
$\base^{j} - \frac{1}{2} \sumC{\ell}{\ell+1}{\ell+1} = \base^{j} - \frac{1}{2} \base^{\ell + 1} \ge
\frac{1}{2} \base^{j} \ge
\frac{1}{4} ( \base^{j+1}-1)/
(\base - 1) \ge
\frac{1}{4} \sumC{\ell}{\ell+1}{j}$, where
the second inequality holds since 
$\base-1 \ge \base/2$. Thus, for no $j > i$,
there exists an vertex $a' \in A_j \setminus A_{j+1}$ with $\dist_{\tilde{G}}(a,a') \leq \frac{1}{4} \sumC{\ell}{\ell+1}{j}$, so $a$ is active.
\end{proof}

Equipped with this notion, we can prove the following lemma. 
It immediately implies that $H$ is a $(1+\epsilon, n^{o(1)})$-spanner since every vertex is $\geq\!(k+1)$-far.

\begin{lemma}
%Assume that $\epsilon \leq 1$. 
For any shortest $s$-$t$ path $\pi_{s,t}$ in $G$ where every vertex $v \in \pi_{s,t} \setminus \{s,t\}$ is $\geq\!i$-far for some $0 < i \leq k+1$, we have
\[
    \dist_H(s,t) \leq (1+20i\epsilon')\dist_G(s,t) + \base^{i} \le
    (1+\epsilon) \dist_G(s,t) + n^{o(1)}.
\]
\end{lemma}
\begin{proof}
As $\epsilon \le 1$ it follows that $\epsilon' \leq 1/20$.
Let us prove the claim by induction on $i$. 

\underline{Base case $i = 1$:} 
Let $v_1$ be the vertex right after $s$ and let $v_2$ be the vertex right before $t$ on $\pi_{s,t,G}$.
As every vertex $v$ on $\pi_{s,t,G} \setminus \{s,t\}$ is $\geq\!1$-far, it follows by the definition 
of $\geq\!1$-far that $v$ belongs to $A_0 \setminus A_1$.  
That is, all vertices on $\pi_{v_1, v_2, G}$ belong to $A_0 \setminus A_1.$
Furthermore, by Lemma~\ref{c:1} it follows that each such vertex $v$ is active. 
Thus the shortest path from $v$ to all vertices in $A_0 \setminus A_1$ at distance at most $\frac{1}{8\log n} \base \geq 1$ are included in $H$ and, in particular, all edges of $\pi_G(v_1, v_2)$ belong to $H$.
As $H$ contains $\tilde G$, $\dist_H(s,v_1) \leq \log n$ and $\dist_H(v_2,t) \leq \log n$.
Finally, $\dist_G(v_1, v_2) \leq \dist_G(s,t) + 2.$
Thus it follows that $\dist_{H}(s,t) \leq \dist_{G}(v_1,v_2) + 2\log n \leq  \dist_G(s,t) + 2\log n + 2 \leq\dist_G(s,t) + (\log n)/{\epsilon'},$ where the last inequality holds since  $\epsilon' \le 1/20$.

\underline{Inductive step $i \mapsto i+1$, for $i \geq 1$:}
As every vertex $v$ in $\pi_{s,t,G}\setminus \{s,t\}$ is $\geq\!(i+1)$-far,  every
such vertex must belong to $ A_{\ell} \setminus A_{\ell + 1}$ for some $\ell \le i$. 

Let the set $\mathcal{C}$ consist of the vertices on $\pi_{s,t}$ that are $\geq\!(i+1)$-far but not $\geq\!i$-far. 
In other words, a vertex $v \in \pi_{s,t}$ is in $\mathcal{C}$ if and only if there is a vertex $a \in A_{i} \setminus A_{i+1}$, with $\dist_{\tilde{G}}(v,a) \leq \base^{i} - \frac{1}{2} \sumC{\ell}{\ell+1}{i}$ (possibly $a=v$).
(Note that $\mathcal{C}$ might be empty.) We prove the following claim for vertices in $\mathcal{C}$.

\begin{claim}
\label{clm:CloseA}
For any vertex $v \in \mathcal{C}$, there exists a vertex $a \in A_i \setminus A_{i+1}$ at distance at most $\frac{1}{4}\base^{i}$ in $\tilde{G}$, such that $a$ is active.
\end{claim}
\begin{proof}
Let $\ell$ be such that $v\in A_\ell\setminus A_{\ell+1}$. If $\ell = i$, then $v$ is active by Lemma~\ref{c:1} and, thus, the claim follows. If $\ell < i$, then let $a \in A_i \setminus A_{i+1}$ be the 
 vertex closest to $v$ among all vertices in $A_i \setminus A_{i+1}$. 
As $v$ is $\geq\!(i+1)$-far, we know that 
\begin{enumerate}
    \item $\dist_{\tilde{G}}(v,a) \leq \base^{i} - \frac{1}{2} \sumC{\ell}{\ell+1}{i}$, and
    \item for no $j \ge i+1$ there exists a vertex $a' \in A_{j} \setminus A_{j+1}$ with $\dist_{\tilde{G}}(v,a') \le \base^{j} - \frac{1}{2} \sumC{\ell}{\ell+1}{i+1}$.
\end{enumerate}{}

Recall that $\base \ge 5$, which implies that $\base-1 \geq 2\base/3$. Suppose for contradiction that $a$ is not active. If $a$ is not active there must exist a vertex $a' \in A_{p} \setminus A_{p+1}$ with $p > i$ such that $\dist_{\tilde{G}}(a,a') \le \sumC{i}{i+1}{p}/4$. 
But this implies that
\begin{align*}
\dist_{\tilde{G}}(v,a') &\leq \dist_{\tilde{G}}(v,a) + \dist_{\tilde{G}}(a,a') 
\leq b^{i} - \frac{1}{2} \sumC{\ell}{\ell+1}{i} +  \frac{1}{4}  \sumC{i}{i+1}{p} \\
&= \base^{i} 
+ \frac{1}{2} \base^{i+1}
+ \frac{1}{4} \frac{\base^{p+1} - \base^{i+1}}{\base-1}
- \frac{1}{2}  \sumC{\ell}{\ell+1}{i+1} 
\le \base^{i} 
+ \frac{1}{2} \base^{i+1}
+ \frac{3}{8} \frac{\base^{p+1} - \base^{i+1}}{\base}
- \frac{1}{2}  \sumC{\ell}{\ell+1}{i+1} \\
&\le \frac{5}{8} \base^{i} 
+ \frac{7}{8} \base^{p} 
- \frac{1}{2} \sumC{\ell}{\ell+1}{i+1}
\leq \base^{p} 
- \frac{1}{2}  \sumC{\ell}{\ell+1}{i+1}.
\end{align*}
This gives a contradiction to the assumption that $v$ is $\ge(i+1)\!$-far.
\end{proof}

To prove the lemma we partition $\pi_{s,t}$ by constructing a sequence of vertices
\[
t_0 = s, s_1, t_1, s_2, t_2, \dots, s_{h-1}, t_{h-1}, s_{h} = t
\]
iteratively as follows: for each $s_g$ with $g \geq 1$, let $s_g$ be the first vertex on the path $\pi_{s,t}(t_{g-1}, t]$ that is in $\mathcal{C}$ or if there is no such vertex, we set $s_g = t$ and $h = g$ and end the sequence. If the sequence does not end with $s_g$, let
$t_g$ with $g \geq 1$ be the farthest vertex from  $s_g$ on $\pi[s_g, t] \cap B_G(s_g, \frac{1}{10\log n}\base^{i+1})$ that is in $\mathcal{C}$. Note that such a vertex always exists since
$s_g \in \pi[s_g, t] \cap B_G(s_g, \frac{1}{10\log n}\base^{i+1}) \cap \mathcal{C}$. (If $s_g$ is the only such vertex
then $t_g = s_g$.)

Clearly, the path $\pi_{s,t}$ is partitioned by the path segments $\pi_{s,t}[t_{g-1}, s_g]$ for $1 \leq g \leq h$ and the segments $\pi_{s,t}[s_g, t_g]$ for $1 \leq g < h$. Observe that for the former kind of segments, we have that the internal path vertices of a piece $\pi_{s,t}[t_{g-1}, s_g]$, i.e. the vertices in $\pi_{s,t}(t_{g-1}, s_g)$, are not in $\mathcal{C}$ by definition. Thus, we have that all such vertices are $\geq\!i$-far (since by assumption every vertex on $\pi_{s,t}$ is $\geq\! i+1$-far and every vertex in $\mathcal{C}$ is not $\geq\! i$-far but none of the vertices on $\pi_{s,t}(t_{g-1}, s_g)$ are in $\mathcal{C}$).
By the sub-path property of shortest paths it holds that $\pi_{s,t}(t_{g-1}, s_g)$
is a shortest path between $t_{g-1}$ and $s_g$.
Hence, we can invoke the induction hypothesis on $\pi_{s,t}(t_{g-1}, s_g)$, with leads to the following statement
\begin{equation}
\label{eq:IHpathSeg}
    \dist_H(t_{g-1}, s_g) \leq (1+20i\epsilon')\dist_G(t_{g-1}, s_g) + \base^{i}.
\end{equation}

For the path segments of the form $\pi_{s,t}[s_g, t_g]$,  \cref{clm:CloseA} shows that there exist vertices $a$ and $a'$ at distance at most  $\frac{1}{4}\base^{i}$ in $\tilde G$ from $s_g$ and $t_g$ respectively that are in $A_i \setminus A_{i+1}$ and are active. Then since 
\begin{align*}
\dist_G(a, a') &\leq \dist_{\tilde{G}}(a, s_g) +  \dist_{G}(s_g, t_g) + \dist_{\tilde{G}}(t_g, a')\\
&\leq \frac{1}{4}\base^{i} + \frac{1}{10\log n}\base^{i+1} + \frac{1}{4}\base^{i}\\
& \leq \frac{1}{8\log n} \cdot \base^{i+1},
\end{align*}
we have that the shortest path $\pi_{a,a'}$ from $a$ to $a'$ is in $H$, i.e., that
$\dist_H(a,a') = \dist_G(a,a')$. 
By applying the triangle inequality to $G$ and the fact that $\dist_G(x,y) \le \dist_{\tilde{G}}(x,y)$ 
(which follows from the fact that $\tilde G$ is a spanner of $G$) 
it follows that
\begin{align*}
 \dist_H(a, a') =\dist_G(a, a') &\leq \dist_{{G}}(a, s_g) +  \dist_{G}(s_g, t_g) + \dist_{{G}}(t_g, a')
\leq \frac{1}{2}\base^{i} + \dist_{G}(s_g, t_g) .
\end{align*}
Since $\tilde{G} \subseteq H$, it holds that $\dist_{H}(s_g, a) \leq  \dist_{\tilde G}(s_g, a)
\leq \frac{1}{4}\base^{i}.$ The same holds for $\dist_{H}(a', t_g).$
We, thus, have that
\begin{equation}
\label{eq:HoppathSeg}
\dist_H(s_g, t_g) \leq \dist_{H}(s_g, a) + \dist_H(a, a') + \dist_{H}(a', t_g) \leq  
\base^{i} + \dist_G(s_g, t_g).
\end{equation}

Finally, let us prove that we did not partition the path into many segments: we claim that 
\[
h \leq \left\lceil \frac{10 \log n \cdot \dist_{G}(s,t)}{\base^{i+1}} \right\rceil
\]
which can be seen from carefully studying the requirements to pick given $s_g$ the next $t_g$ and $s_{g+1}$, which stipulate that each $s_g$ and $s_{g+1}$ are at distance at least $\frac{1}{10\log n}\base^{i+1}$.

We now combine our insights and the fact that $\epsilon' \le 1/20$, and take the sum over the path segments. This gives 
\begin{align*}
\dist_H(s,t) &= \sum_{g=1}^{h} \dist_H(t_{g-1}, s_{g}) + \sum_{g=1}^{h-1} \dist_H(s_{g}, t_g) \\
&\leq \sum_{g=1}^{h} \left((1+20i\epsilon')\dist_G(t_{g-1}, s_g) + \base^{i}\right) 
+ \sum_{g=1}^{h-1} \left( \base^{i} + \dist_G(s_g, t_g) \right) \\
&\leq \sum_{g=1}^{h} (1+20i\epsilon')\dist_G(t_{g-1}, s_g) + \sum_{g=1}^{h-1} \dist_G(s_g, t_g) + h \left( \base^{i} + \base^{i}\right)\\
&\leq (1+20i\epsilon')\dist_G(s,t) + \left\lceil \frac{\dist_{G}(s,t)}{\base^{i+1}} \right\rceil \cdot 20 \log n \cdot \base^{i}\\
& \leq (1+20i\epsilon')\dist_G(s,t) + 20 \log n \cdot \epsilon' \cdot \frac{\dist_{G}(s,t)}{\log n} + 20 \log n \base^{i} \\
&\leq (1+20(i+1)\epsilon')\dist_G(s,t) + \base^{i+1}\\
&\leq  (1+\epsilon) \dist_G(s,t) + n^{o(1)}
\end{align*}
as required.
\end{proof}

\paragraph{Sparsity.} In order to prove sparsity, let us establish the following key claim. The claim is in fact slightly more powerful than necessary, however, this power will be exploited when bounding the running time.

\begin{claim}
\label{clm:fewActiveVertices}
For any vertex $v$, and for any $i \leq k$ there are at most $O(n^{1/k} \log n)$ active vertices in $A_i \setminus A_{i+1}$ in
\[
B_{\tilde{G}}\left(v, \frac{1}{8}\base^{i+1}\right)
\]
with high probability. Further, if a vertex $v \in A_i \setminus A_{i+1}$ is active, we have $|B_{\tilde{G}}\left(v, \frac{1}{8}\base^{i+1}\right)| \leq n^{(i+1)/k}$ with high probability.
\end{claim}
\begin{proof}
Let us observe that for any vertex $v \in V$, if
 $  \left|B_{\tilde{G}}(v, \frac{1}{8}\base^{i+1})\right| \geq n^{(i+1)/k}$ 
there exists at least one vertex $a$ in $A_{i+1}$ that hits $B_{\tilde{G}}(v, \frac{1}{8}\base^{i+1})$ with high probability. %For the rest of the proof let us condition on the event that every such ball is hit by a vertex from the next-higher level.

We now claim that if there is such a vertex $a$, no vertex $a'$ in $B_{\tilde{G}}(v, \frac{1}{8}\base^{j+1}) \cap A_j$ is active for any $j \leq i$. To see this let us assume for the sake of contradiction that there exists such a vertex $a'$. 
Now, let us first consider the case that the vertex $a$ is active. Then, the distance between $a$ and vertex $a'$ is $\mathbf{dist}_{\tilde{G}}(a, v) + \mathbf{dist}_{\tilde{G}}(v,a') \leq \frac{1}{8}b^{i+1}+\frac{1}{8}b^{i+1} = \frac{1}{4}b^{i+1} \leq \frac{1}{4} c_{j, i+1}$ which implies that $a$ must have deactivated $a'$ and therefore leads to a contradiction. 

Now, let us consider the case, where $a$ is not active. Then, there exists a vertex $a'' \in A_{\ell} \setminus A_{\ell +1}$ that is active and with $\ell > i+1$ such that $\dist_{\tilde{G}}(a, a'') \leq \frac{1}{4}  \sumC{i+1}{i+2}{\ell}$. Such a vertex exists by definition if $a$ is not active so we have that $a''$ is well-defined. 

But then we have by the triangle inequality that 
\begin{align*}
\dist_{\tilde{G}}(a', a'') &\leq \dist_{\tilde{G}}(a', v) + \dist_{\tilde{G}}(v, a) + \dist_{\tilde{G}}(a, a'') \\
&\leq \frac{1}{8}\base^{j+1} + \frac{1}{8}\base^{i+1} + \frac{1}{4} \sumC{i+1}{i+2}{\ell}\\ 
&\leq \frac{1}{4} \sumC{j}{j+1}{\ell}
\end{align*}
where we use in the last inequality that $\sumC{j}{j+1}{\ell} = \sum_{y=j+1}^{\ell} \base^y =  \sum_{y=j+1}^{i+1} \base^y + \sum_{y=i+2}^{\ell} \base^y = \sum_{y=j+1}^{i+1} \base^y + c_{i+1, \ell}$ and since $i \geq j$ we have that $\sum_{y=j+1}^{i+1} \base^y \geq b^{i+1}$. This implies that $a'$ cannot be active since $a''$ is close enough to $a'$ to deactivate it during the algorithm. This again leads to a contradiction and thereby completes the proof of our claim that $a'$ is not active.

Finally, we conclude that we have for any vertex $v$ and level $i$, that either $\left|B_{\tilde{G}}(v, \frac{1}{8}\base^{i+1})\right| < n^{(i+1)/k}$ in which case we have by a straightforward application of a Chernoff bound that with high probability at most $O(n^{1/k} \log n)$ of these vertices are sampled into $A_i$. Thus, there can also be only $O(n^{1/k} \log n)$ active vertices from $A_i$ in the ball. 

Otherwise  $\left|B_{\tilde{G}}(v, \frac{1}{8}\base^{i+1})\right| \geq n^{(i+1)/k}$, so by the above claim we have that some vertex $a$ in $A_{i+1}$ hits the ball $B_{\tilde{G}}(v, \frac{1}{8}\base^{i+1})$ which results in all vertices in the ball that are in $A_j \setminus A_{j+1}$ to be deactivated for any $j \leq i$, and in particular for $i$. This proves the first part of the claim.

For the second part of the claim, observe that the contrapositive of the above claim with $a' = v$ implies that there exists no active vertex $a''$ in $A_{i+1}$ in the ball $B_{\tilde{G}}(v, \frac{1}{4} c_{i, i+1}) \subseteq B_{\tilde{G}}(v,  \frac{1}{8}\base^{i+1})$, since $ c_{i, i+1} = b^{i+1}$. Further, by the contrapositive of our initial observation, we have that then
$|B(v, \frac{1}{8}\base^{i+1})| < n^{(i+1)/k}$. 
\end{proof}

A special case of the claim is that each vertex $a \in A_{i}\setminus A_{i+1}$ has $\tilde{O}(n^{1/k})$ active vertices in its ball
$B(a, \frac{1}{8}\base^{i+1})$ that are in $A_{i}\setminus A_{i+1}$. Thus, if $a$ is active itself this gives an upper bound on the number of paths that are included in $H$ due to $a$ (if $a$ is not active no paths are added). Since the maximum such path length at any level is $\frac{1}{8}\base^{k+1}$, we have that there are only $\Tilde{O}(n^{1+1/k} \base^{k+1})= n^{1+o(1)}$ edges in $H$ due to paths. The spanner $\tilde{G}$ that is additionally added to $H$ contains only $\tilde{O}(n)$ edges which is subsumed in the previous bound. 

\paragraph{Running Time.} 
To bound the running time, observe that in order to determine which vertices are active, we run for each $i \leq k$, a BFS to depth $\frac{1}{8}\base^{i+1}$ from every active vertex in $A_i \setminus A_{i+1}$ on $\tilde{G}$. By \Cref{clm:fewActiveVertices}, each vertex $v \in V$ is only explored by $\tilde{O}(n^{1/k})$ active vertices on each level. There are only $\sqrt{\log n}$ levels, so each vertex $v \in V$ is only explored by $\tilde{O}(n^{1/k})$ active vertices in total. Thus, every edge incident to each vertex $v$ in $\tilde{G}$ is only explored $\tilde{O}(n^{1/k})$ times. Since $\tilde{G}$ has only $\tilde{O}(n)$ edges, and since each BFS runs linearly in the number of edges explored, we can bound the total running time by $\tilde{O}(n^{1+1/k})$, for each $i$, and thus also for all values $i$.

After determining which vertices are active, it remains to bound the time to compute the paths between any two active vertices $a, a' \in A_i \setminus A_{i+1}$ for some $i$. For $i \geq \gamma$, we can simply check the distances between any pair of such vertices (even the ones that are non-active) by looking up their shortest paths in $\mathcal{A}$ and inserting them into $H$ if the criterion is satisfied. We will calculate the running time of $\mathcal{A}$ at the end. %This takes time at most $(n^{1-\gamma/k})^2 n^{o(1)}) = n^{0.5 + o(1)}$ using a brute-force approach. 
For any active vertex $a \in A_{i} \setminus A_{i+1}$, where $i < \gamma$, we further run a BFS on $G$ to depth $\frac{1}{8\log n}\base^{i+1}$. We observe that since $\tilde{G}$ is a $\log n$-spanner of $G$, we have for every $v \in V$ that
\[
B_{\tilde{G}}(v,\frac{1}{8\log n}\base^{i+1}) \subseteq B_G(v,\frac{1}{8}\base^{i+1}). 
\]
By \Cref{clm:fewActiveVertices} we have again that each vertex $v$ is only explored $\tilde{O}(n^{1/k})$ times during these executions of BFS. Using the second fact of \Cref{clm:fewActiveVertices}, we further have that each BFS from an active vertex $a\in A_i \setminus A_{i+1}$ only explores at most $n^{(i+1)/k}$ vertices. But this in turn implies that each vertex $v$ that is strictly contained in such a ball $B_G(a,\frac{1}{8\log n}\base^{i+1})$ has degree %$\mathbf{deg}_G(v)$
bounded by $n^{(i+1)/k}$. Since we explore only for $i < \gamma$, we therefore have that each explored vertex has degree at most $n^{\gamma/k} \leq n^{\eps}$. Putting everything together, during all BFS explorations, we scan every vertex at most $\tilde{O}(n^{1/k})$ times while the number of edges present at the vertex is at most $n^{\eps}$ and since we have $n$ vertices, the total time required for all explorations can be bounded by $O(n \cdot n^{1/k} \cdot n^{\eps}) = n^{1+\eps+o(1)}$. 

Finally, we have to account for the data structures used. While the running time of $\mathcal{B}$ is subsumed in our previous bound, we have that $\mathcal{A}$ has worst-case update time $\tilde{O}(n^{\delta}(n^{\omega(1, 1, \eps)-\eps} + (n^{1-\gamma/k})^2 n^{\eps}))$ for any constant $\eps \leq \eps^*$ and $\delta > 0$ by \Cref{cor:SankowskiPathReporting} since $|S| = n^{1-\gamma/k} = n^{1-\eps+o(1)}$. Setting $\eps = \eps^*$ and $\delta = 0.01$, the size $|S| = o(\sqrt{n})$ and thus we obtain worst-case update time $O(n^{1+\eps^*+o(1)}) = O(n^{1.529})$.

\section{Applications}\label{sec:app}
Let $\epsilon>0$ be an arbitrarily small constant.
In this section we give two applications of the data structures developed in the previous sections, namely fully dynamic $(1+\epsilon)$-approximate all-pairs shortest paths with path reporting  and $(2+\epsilon)$-approximate fully dynamic Steiner tree. Both algorithms are the first algorithms that take sub-quadratic worst-case time for these problems.

\subsection{Fully dynamic approximate all-pairs shortest paths with path reporting}
In this subsection we give a fully dynamic data structure that reports $(1+\epsilon)$-approximate
all-pairs shortest paths in sub-quadratic worst-case time. Note that all previous such data structure were only
able to report the \emph{distance}, i.e.~the \emph{length} of the paths, but not the actual paths.
\begin{theorem}\label{thm:apsp}
Let $\epsilon>0$ be an arbitrarily small constant.
There exists an algorithm that maintains $(1+\epsilon)$-approximate all-pairs shortest-path in worst-case time 
$n^{1+\eps^*  + o(1)} = O(n^{1.529})$ per edge update, in worst-case time 
 $n^{1  + o(1)}$ 
 per path reporting query and per distance reporting query with high probability against an oblivious adversary.
The preprocessing time is $n^{\omega + o(1)}$ if the initial graph is non-empty and $\tilde O(n^2)$ if the initial graph is empty. 

\end{theorem}
\begin{proof}
We maintain the following data structure:

(1) We maintain the fully dynamic path reporting data structure given in the statement of Theorem~\ref{thm:DpathReporting} on $G$ with $D = 2 n^{o(1)}/\epsilon$ and $\eps=\eps^*< 0.529$.

(2) We maintain for $G$ the fully dynamic  $(1+\epsilon/2, n^{o(1)})$-spanner from Theorem~\ref{thm:AlgebraicSpanner}.

To answer a distance reporting query  we first ask a distance query with parameter $D$ in (1). This gives us the exact answer if the distance is less than $D$. Otherwise we run a static shortest path algorithm on the spanner that we maintain in (2). As the spanner has at most $n^{1 + o(1)}$ edges this takes time
$n^{1 + o(1)}.$
Note that, by the choice of $D$, the shortest path on the spanner gives a $(1+\epsilon)$-approximation of the
shortest path in $G$ as $(1 + \epsilon/2) \dist(s,t) + n^{o(1)} \leq (1 + \epsilon) \dist(s,t)$ for
$\dist(s,t) \ge D = 2 n^{o(1)}/\epsilon$.  
Thus a distance query returns a $(1+\epsilon)$-approximate answer and takes time $n^{\eps + o(1)}/\epsilon +  n^{1+o(1)} = n^{1 + o(1)}.$

To answer a path reporting query between two nodes $s$ and $t$  we first ask a distance query with parameter $D$ in (1). If the distance is less than $D$, we ask a path reporting query in (1) in time $O(D^2n^{\eps^*}) = o(n)$. Otherwise, we execute a static shortest path algorithm on the spanner that we maintain in (2).
Thus a path reporting query returns a
$(1+\epsilon)$-approximate shortest path and takes time $n^{1 +o(1)}$ time.
\end{proof}

If we increase the preprocessing time to $O(n^{2.621})$, and the update time to $O(n^{1.843})$, we can reduce the cost of a distance reporting query even further.

\begin{corollary}
Let $\epsilon>0$ be an arbitrarily small constant.
There exists an algorithm that maintains $(1+\epsilon)$-approximate all-pairs shortest-path in worst-case time 
$O(n^{1.843 + o(1)})$
 per edge update, in worst-case time 
 $O(n^{1 + o(1)})$ 
 per  path reporting query, and in worst-case time 
 $O(n^{.45})$ 
 per distance reporting query  with high probability against an oblivious adversary.
The preprocessing time is $n^{\omega + o(1)}$ if the initial graph is non-empty and $\tilde O(n^2)$ if the initial graph is empty. 
\end{corollary}
\begin{proof}
Additionally maintain 
 a fully dynamic $(1+\epsilon)$-approximate APSP data structure from~\cite{BrandN19}, which takes worst-case time $O(n^{1.843})$ per edge update and worst-case time $O(n^{0.45})$ per distance query. It speeds
 up the distance reporting queries and, in combination with theorem~\ref{thm:apsp}, leads to the result.
\end{proof}

\subsection{Steiner trees with terminal vertex and edge updates}
In this section we give a further application of the data structure of the previous section.
Assume we are given an (unweighted) graph $G=(V,E)$ with a dynamically changing edge set and a dynamically changing terminal set $S \subseteq V.$ 
A {\em Steiner tree}  $T_S$ for a vertex set $S$ is a tree that (1) is a subgraph of $G$, (2) spans $S$, and (3) has the minimum number of edges.
A $\beta$-approximate {\em Steiner tree} is a tree which fullfils conditions (1) and (2), and the weight of the tree is at most a factor $\beta$ larger than the weight of the edges of a minimum Steiner tree.

A 2-approximate Steiner tree can be found as follows: Construct a weighted graph $\tilde G$ that consists of a clique
on the vertices of $S$ such that each edge $(u,v)$ has length $\dist_G(u,v)$, i.e., the length of the shortest path between $u$ and $v$ in $G$. Find an MST $\tilde T_S$ in $\tilde G$. Its weight $w(\tilde T_S)$ gives a 2-approximation of the value $OPT$ of the optimal
Steiner tree, i.e. $OPT \le w(\tilde T_S) \le 2 OPT.$ The reason is as follows: Consider an Eulerian tour $\mathcal{E}(T_S)$ of the optimal Steiner tree $T_S$
in $G.$ It traverses every edge twice and, thus, has length $2OPT$. Now replace the subpath between two consecutive terminal
vertices $u$ and $v$ of $\mathcal{E}(T_S)$ by an edge. Note that the length of this subpath is at least $\dist_G(u,v)$ and that there is an edge $(u,v)$ in $\tilde G$ with length $\dist_G(u,v)$. Thus $\mathcal{E}(T_S)$ induces a cycle in $\tilde G$
whose length is at most $2OPT$. As the minimum spanning tree $\tilde T_S$ in $\tilde G$ has a weight that is at most the length of this
cycle, its weight is at most $2OPT.$ Note that if the weights in $\tilde G$ are between $\dist_G(u,v)$ and $(1+\epsilon/2) \dist_G(u,v)$ for some arbitrarily small $\epsilon >0$, then the weight of $\tilde T_S$ is at most $(2+\epsilon) OPT.$
To construct the corresponding approximate Steiner tree replace each edge in $\tilde T_S$ by a shortest path in $G$ between its endpoints
and compute a tree $\tilde T'_S$ in the resulting graph. The weight of $\tilde T'_S$ is at most the weight of $\tilde T_S$. 

We consider the following dynamic changes to the input: (i) edge insertions and deletions and (ii) additions and removals from $S$.
In the \emph{fully dynamic unweighted Steiner tree problem} the goal is  to maintain a Steiner tree after each modification to the input.
In the \emph{$\beta$-approximate fully dynamic 
unweighted Steiner tree problem} the goal is to maintain an $\beta$-approximate {\em Steiner tree} after each modification of the input. Note that we want to maintain
an actual tree, \emph{not just the value} of the $\beta$-approximate Steiner tree.
We show the following result.

\begin{theorem}
Let $\epsilon>0$ be an arbitrarily small constant.
There exists an algorithm that solves the $(2+\epsilon)$-approximate fully dynamic unweighted Steiner tree problem with high probability against an oblivious adversary, in worst-case time 
$O(n^{1.529}
  + s^2 \cdot n^{1 + o(1)})$
 per edge update and in worst-case time 
 $s n^{1 + o(1)}$ 
 per vertex addition to or removal from $S,$ where $s$ is the current size of $S.$
The preprocessing time is $n^{\omega + o(1)}$ if the initial graph is non-empty and $\tilde O(n^2)$ if the initial graph is empty. 

\end{theorem}

\begin{proof}
By our discussion before the theorem it suffices to maintain the graph $\tilde G$ whose edge weights are a $(1+\epsilon/2)$-approximation of the length of their endpoints in $G$ and a minimum spanning tree in $\tilde G.$ To build the actual Steiner tree, we need to replace then each edge in this spanning tree into the corresponding shortest path in $G$.

For convenience, we denote by $A(\epsilon,n)$, the additive error of a $(1+\epsilon, n^{o(1)})$-spanner as maintained in Theorem~\ref{thm:AlgebraicSpanner}, i.e. Theorem~\ref{thm:AlgebraicSpanner} maintains a $(1+\epsilon, A(\epsilon,n))$-spanner. From this, it is also easy to see that $A(\epsilon, n) = n^{o(1)}$ for any constant $\epsilon > 0$.

Now, we describe the data structures to maintain the necessary information for updating our Steiner tree approximation:

(1) We maintain the fully dynamic path reporting data structure given in the statement 
of Theorem~\ref{thm:DpathReporting} on $G$ with $D = 4 A(\epsilon/4, n)/\epsilon$ and $\eps=\eps^*< 0.529$.

(2) We maintain the fully dynamic  $(1+\epsilon/4, A(\epsilon/4, n))$-spanner from theorem~\ref{thm:AlgebraicSpanner}.

(3) We maintain a $|S| \times |S|$ array that contain the $(1+\epsilon/4)$- approximate all-pairs shortest path lengths between any pair of terminals.

(4) We maintain the graph $\tilde G$ and a fully dynamic  minimum spanning tree data structure from $\tilde G$ that takes 
$n^{o(1)}$ deterministic worst-case update time~\cite{abs-1910-08025}.
%$O(\log^4 n \log W)$ worst case time per operation~\cite{gibb2015dynamic}, where $W$ is the maximum weight in the graph. Note that in our setting $W \le n.$

To answer a distance query  we first ask a distance query with parameter $D$ in (1). This gives us the exact answer if the distance is less than $D$. Otherwise we run a static shortest path algorithm on the spanner that we maintain in (2). As the spanner has at most $n^{1 + o(1)}$ edges this takes time
$n^{1 + o(1)}.$ Note that, by the choice of $D$, the shortest path on the spanner gives a $(1+\epsilon/2)$-approximation of the
shortest path in $G$ as $(1 + \epsilon/4) \dist(s,t) + A(\epsilon/4, n) \leq (1 + \epsilon/2) \dist(s,t)$ for
$\dist(s,t) \ge D = 4 A(\epsilon/4, n)/\epsilon$.  
Thus a distance query returns a $(1+\epsilon/2)$-approximate answer and takes time $n^{\eps^* + o(1)}/\epsilon +  n^{1+o(1)} = n^{1 + o(1)}.$

To answer a path reporting query between two nodes $s$ and $t$  we first ask a distance query with parameter $D$ in (1). If the distance is less than $D$, we ask a path reporting query in (1). Otherwise, we execute a static shortest path algorithm on the spanner that we maintain in (2).
Thus a path reporting query returns a
$(1+\epsilon/2)$-approximate shortest path and takes time $n^{1 +o(1)} +  n^{\kappa^*+o(1)}/\epsilon = n^{1 + o(1)}$ time.

After each update to either $E$ or $S$ we update data structures (1) - (4) as described below. Then we build the approximate Steiner Tree $\tilde T'_S$ (in $G$) from $\tilde T_S$ as described above, executing $s-1$ path reporting queries.
This takes time $s \cdot n^{1 + o(1)}$, which as we will see is subsumed by the runtime of updating (4).

We are left with describing how we update (1) - (4).
Each edge update leads to the corresponding update in (1) and (2). Then we recompute $\tilde G$ from scratch using $O(s^2)$ distance queries and compute its minimum spanning tree as well as the dynamic MST data structure from scratch. Afterwards we build $\tilde T'_S$ as described above. This takes total time
$n^{1+\eps^*  + o(1)}
  + s^2 \cdot n^{1 + o(1)}$.

If a vertex $v$ is added to $S$, we compute the distance in the spanner from $v$ to all nodes in $S$, add an edge from $v$ to every vertex in $S$ with the corresponding length to $\tilde G$, and update its dynamic MST data structure. This takes time $s n^{1+o(1)}.$ Afterwards we build $\tilde T'_S$ as described above.

If a vertex is removed from $S$, we remove its incident edges from $\tilde G$ and its dynamic MST data structure. This takes time $s n^{1+o(1)}.$ Afterwards we build $\tilde T'_S$ as described above.
\end{proof}

At the cost of increasing the preprocessing time to $O(n^{2.621})$ we can additionally maintain a fully dynamic $(1+\epsilon/2)$-approximate APSP data structure from~\cite{BrandN19}, which takes worst-case time $O(n^{1.843})$ per edge update and worst-case time $O(n^{0.45})$ per distance query. It speeds
 up the distance query and leads to the following result.
 
\begin{theorem}
Let $\epsilon>0$ be an arbitrarily small constant.
There exists an algorithm that solves the $(2+\epsilon)$-approximate fully dynamic unweighted Steiner tree problem with high probability against an oblivious adversary in worst-case time 
$O(n^{1.843 + o(1)}  + s^2 \cdot n^{0.45} + s \cdot n^{1 + o(1)})$
 per edge update, in worst-case time 
 $s n^{1  + o(1)}$ 
 per vertex addition to or removal from $S$, where $s$ is the current size of $S.$
The preprocessing time is $O(n^{2.621})$.
\end{theorem}

%Virginia: The new Corollary should also allow us to compute the parameters for the Steiner tree application. Before we reoptimize Max's spanner, for s nodes Steiner tree if we use the data structure with eps=0.157, then the edge update time for Steiner tree would be n^{1.843}+sn^{1.157} since edge updates in the algebraic path data structure would cost n^{2-0.157}=n^{1.843}, and the s path reporting queries and s^2 distance queries would be sn^{1.157}. For insertions and deletions for S we also spend sn^{1.157} in the algebraic structure+spanner.

\section{Combinatorial dynamic spanner algorithms}\label{sec:span}

In this section, we present our combinatorial algorithms to maintain $(1+\epsilon, \beta)$-spanners in dynamic graphs. We first present a simple algorithm to maintain a spanner in a partially dynamic graph, then show how to extend the construction to obtain near-linear update time in fully dynamic graphs. Finally, we present a deterministic algorithm to maintain a emulator with large additive error in partially dynamic graphs. 

\subsection{A Simple Partially Dynamic Spanner}

In this section, we prove the following theorem.

\begin{theorem}
\label{thm:SimpleDecrSpanner}
For any constant $0<\epsilon \leq 1$,
given an undirected, unweighted partially dynamic graph, we can maintain a $(1+\epsilon, n^{o(1)})$-spanner of size $n^{1+o(1)}$ in total time $m^{1+o(1)}$ with high probability against an oblivious adversary.
\end{theorem}

\paragraph{The Algorithm.} %Given an undirected unweighted input graph $G=(V,E)$. 
For $k = \sqrt{\log n}$, we sample sets $V = A_0 \supseteq A_1 , \dots \supseteq A_k \supseteq A_{k+1} = \emptyset$ where $A_i$ for $i \in [1,k]$ is obtained by sampling each vertex in $V$ with probability $n^{-i/k} \log n$ (and to make the sets nesting add it to all $A_j$ where $j \leq i$).  

Let $\epsilon'=\epsilon/8$. Throughout the algorithm, we say for each $a \in A_{i} \setminus A_{i+1}$ that $a$ is \emph{active} if and only if there is no vertex $a' \in A_{j}$ for $j > i$, with $\dist(a,a') \leq \epsilon'^{-(j+1)} - \epsilon'^{-(i+1)}$. Thus, all vertices in $A_k$ are always active. In the decremental algorithm, once a vertex becomes active, it remains active for the rest of the algorithm, and in the incremental algorithm, once a vertex becomes inactive it remains inactive for the rest of the algorithm. Further, for each active vertex $a \in A_{i} \setminus A_{i+1}$, we maintain the ball $\mathcal{B}_a = {B}(a, \epsilon'^{-(i+1)})$. Otherwise, we let $\mathcal{B}_a = \emptyset$. A similar definition of active nodes was introduced by Chechik~\cite{chechik2018near} for an emulator construction.

In order to obtain the spanner, it suffices to take the union of the shortest-path trees maintained from each vertex $a \in V$, truncated at depth $r_a$, i.e. the spanner is defined at all stages by
\[
H = \bigcup_{a \in V, b \in \mathcal{B}_a} \pi_{a,b}
\]
We maintain these shortest-path trees using the ES-tree data structure.

%We point out that the shortest path trees are explicitly maintained in the data structure maintaining the balls, so the spanner can be maintained explicitly and can be accessed.

\paragraph{Spanner Approximation.} %Let us assume w.l.o.g. that $\epsilon < 1/4$. 
At any stage, for an arbitrary shortest path $\pi_{s,t}$ from $s$ to $t$, we can expose a path in $H$ that approximates $\pi_{s,t}$ as follows. 

 Let $i$ be chosen to be the smallest index such that $B(s, 2\epsilon'^{-(i+1)}) \cap A_{i+1} = \emptyset$. 
First we show that $s$ is within distance at most $2\epsilon'^{-i}$ of an active vertex of $A_i$.
Then, we have $i \leq k$ since $A_{k+1} = \emptyset$. By a straightforward application of the Chernoff bound it follows that $|B(s, 2\epsilon'^{-(i+1)})| \leq n^{(i+1)/k}$ with high probability, since otherwise $A_{i+1}$ would likely hit $B(s, 2\epsilon'^{-(i+1)})$. Further, observe that by minimality of $i$, there exists some vertex $a$ in $B(s, 2\epsilon'^{-i}) \cap A_{i}\setminus A_{i+1}$, i.e. there exists some vertex $a \in A_i$ at distance at most $2\epsilon'^{-i}$ from $s$. Now, observe that $B(a, \epsilon'^{-(i+1)}) \cap A_{i+1} \subseteq B(s, 2\epsilon'^{-(i+1)}) \cap A_{i+1} = \emptyset$ by the triangle inequality. Thus $a$ is active.

Now, we distinguish two scenarios:
\begin{enumerate}
    \item if $t \in \mathcal{B}_a$: then we can use the paths $\pi_{s, a} \circ \pi_{a,t}$ with combined weight less than $2\epsilon'^{-(i+1)}\leq 2\epsilon'^{-(k+1)}$,
    \item otherwise, let $s'$ be the vertex in $\mathcal{B}_a \cap \pi_{s,t}$ that is closest to $t$. As $t \not\in \mathcal{B}_a$ it holds that $\dist_G(a, s') = \epsilon'^{-(i+1)}.$ As $\dist_G(s,a) \le 2 \epsilon'^{-i}$ and $\epsilon' \le 1/4$, it follows by the triangle inequality that $\dist_G(s, s') \ge  \epsilon'^{-(i+1)} - 2 \epsilon'^{-i} \ge \epsilon'^{-(i+1)}/2$. 
    By the triangle inequality it holds that $\dist_G(s',a) \le \dist_G(s,s') + \dist_G(s,a) \le \dist_G(s,s') + 2\epsilon'^{-i}$. Thus, we have that the path $\pi_{s, a} \circ \pi_{a,s'}$ has weight at most $\dist_G(s,s') + 4\epsilon'^{-i} =  \dist_G(s,s') + 8 \epsilon' \cdot (\epsilon'^{-(i+1)}/2)\leq (1+ 8\epsilon')\dist_G(s,s')=(1+ \epsilon)\dist_G(s,s')$.
\end{enumerate}
If we are in the first scenario, we are done. Otherwise, we can recurse on $\pi_{s',t}$. Aggregating all of the path segments from the recursion, we derive
\[
    \dist_H(s,t) \leq (1+\epsilon)\dist_G(s,t) + 2\epsilon'^{-(k+1)}.
\]
Since this worked with high probability, and we only have a polynomial number of shortest paths to consider (namely $n^2$ shortest paths after each of the at most $m$ edge deletions), we can take a union bound over all these shortest paths and obtain overall correctness with high probability. %Rescaling $\epsilon$, we obtain the claimed bounds.

\paragraph{Sparseness of the Spanner.} Pick any $i \in [0,k-1]$. Let us observe that for any vertex $v \in V$, as long as
\[
    |B(v, \epsilon'^{-(i+1)})| \geq n^{(i+1)/k} 
\]
we have that at least one vertex $a''\in A_{i+1}$ hits $B(v, \epsilon'^{-(i+1)})$ with high probability. Further, while $a''$ is in $B(v, \epsilon'^{-(i+1)})$, we claim that no vertex in $B(v, \epsilon'^{-(i+1)}) \cap A_i$ is active. For the sake of contradiction assume that there exists such a vertex in $a'$, then by the triangle inequality, we have
\[
    \dist(a', a'') \leq \dist(v, a'') + \dist(v, a') \leq 2\epsilon'^{-(i+1)} < \epsilon'^{-(i+2)} - \epsilon'^{-(i+1)}
\]
But that implies that $a'$ should not be active. This implies that vertices in $B(v, \epsilon'^{-(i+1)}) \cap A_i$ can only become active after $|B(v, \epsilon'^{-(i+1)})| < n^{(i+1)/k}$. But then with high probability $|B(v, \epsilon'^{-(i+1)}) \cap A_i| = O(n^{1/k} \log^2 n)$ (again by using a Chernoff bound). Note that vertices in $B(v, \epsilon'^{-(i+1)}) \cap A_i$ are the only nodes of $A_i$ in whose balls $v$ can lie. Thus, each vertex $v$, only ever participates in $O(n^{1/k} \log^2 n)$ many balls rooted at vertices in $A_i$. Thus, the total size of the spanner can be upper bounded by $\Tilde{O}(n^{1+1/k})= n^{1+o(1)}$.

\paragraph{Correctness of the Algorithm.} While the rest of the algorithm is rather straightforward, we show that by maintaining the balls $\mathcal{B}$ 
for all active nodes we can maintain for every vertex whether it should become active (in the case of a decremental algorithm) or inactive (in the case of an incremental algorithm). 

Consider a vertex $a$ and let $i$ be such that $a\in A_i\setminus A_{i+1}$. Consider the maximal layer $j$ such that
there exists a vertex $a' \in A_j \setminus A_{j+1}$, $i < j$, at distance $\dist(a, a') \leq \epsilon'^{-(j+1)} - \epsilon'^{-(i+1)}$. 
We will show below that $a'$ is active. Then clearly we have $a \in \mathcal{B}_{a'}$ and $a'$ can therefore inform $a$ that it cannot become active yet. This suffices for an incremental algorithm since vertices only go from active to inactive. We argue next that this also suffices for a decremental algorithm (where vertices only go from inactive to active).

When $a'$ informs $a$ that $a$ cannot be active, we say that $a'$ invokes a \emph{block} on $a$. This block stays in effect until $a'$ revokes it, which happens when $\dist(a, a')$ grows above $\epsilon'^{-(j+1)} - \epsilon'^{-(i+1)}$. By maintaining $\mathcal{B}_{a'}$, we can detect when $\dist(a, a')$ grows above $\epsilon'^{-(j+1)} - \epsilon'^{-(i+1)}$ because if $a\in \mathcal{B}_{a'}$, then $\dist(a, a')$ is explicitly maintained, and if $a$ leaves $\mathcal{B}_{a'}$ then we can detect this and conclude that $\dist(a, a')$ has grown above $\epsilon'^{-(j+1)} - \epsilon'^{-(i+1)}$. In summary, we determine the status of $a$ by setting $a$ to be active if and only if no other vertex has invoked a block on $a$ (and not revoked it).

To show that $a'$ is active, assume by contradiction that $a'$ is inactive.
It follows that there has to be a vertex $a'' \in A_{\ell} \setminus A_{\ell +1}$ for $j < \ell$ at distance at most $\epsilon'^{-(\ell+1)} - \epsilon'^{-(j+1)}$ to $a'$. But by the triangle inequality, we would have 
\[
    \dist(a, a'') \leq \dist(a, a') + \dist(a', a'') \leq \epsilon'^{-(j+1)} - \epsilon'^{-(i+1)} + \epsilon'^{-(\ell + 1)} - \epsilon'^{-(j+1)} = \epsilon'^{-(\ell + 1)} - \epsilon'^{-(i+1)} 
\]
contradicting that $a'$ was a vertex satisfying our constraint with maximal index $j$.

\paragraph{Total Time.} 
Finally, let us bound the total time. Constructing the set $A_i$ takes time $O(n^{1 + o(1)})$. 
We have already shown that each vertex is in at most $O(n^{1/k} \log^2 n)$ balls over the course of the algorithm. Since each edge only participates in balls containing both of its endpoints, we have that each edge is in at most $O(n^{1/k} \log^2 n)$ balls.

Since ES-trees run in time $O(mr)$, each edge is scanned at most $O(mr)$ times for each of the at most $O(n^{1/k} \log^2 n)$ balls it is in. Each ball has radius at most $\epsilon'^{-(i+1)}$, so an analysis over all levels and balls shows that the total update time can be bound by
\[
    \sum_{i \leq k} m \cdot O(n^{1/k} \log^2 n) \cdot \epsilon'^{-(i+1)} = \Tilde{O}(m^{1+1/\sqrt{\log n}} \epsilon'^{-(\sqrt{\log n}+1)}) = m^{1+o(1)}.
\]
%for our choice of $k$.
Note that this also bounds the total number of changes that can occur to the spanner during all deletions as we always explicitly construct the spanner.

\subsection{A Simple Determininstic Fully Dynamic Spanner}

Although the efficiency of the previous algorithm crucially depends on randomness and the non-adaptiveness of the adversary, Roditty, Thorup and Zwick showed in \cite{roditty2005deterministic} that the construction given above can be efficiently derandomized in the static setting. Thus, there exist a static algorithm, that runs in time $m^{1+o(1)}$ that computes a $(1+\epsilon, n^{o(1)})$-spanner deterministically. Rerunning this algorithm after every update yields the following theorem.
%is the best bound we can obtain for a combinatorial algorithm in the fully dynamic setting.

\begin{theorem}
For any constant $0<\epsilon \le 1$,
given an undirected, unweighted graph $G=(V,E)$ undergoing vertex updates, we can deterministically maintain a $(1+\epsilon, n^{o(1)})$-spanner of size $n^{1+o(1)}$ with preprocessing time $m^{1+o(1)} + O(n)$ and worst update time $m^{1 + o(1)}$.
\end{theorem}

We point out that a similar result can likely be obtained by using the randomized construction by Elkin and Zhang \cite{elkin2006efficient} which can be derandomized using the deterministic sparse neighborhood covers \cite{awerbuch1998near}.

\subsection{Partially Dynamic Algorithms to Maintain Emulators/Spanners with High Additive Error}

In this section, we prove the following theorems.

\begin{theorem}
Given a partially dynamic unweighted graph $G=(V,E)$, a multiplicative approximation parameter $\epsilon > 0$, and an additive approximation parameter $\alpha \geq 0$, there exists a \emph{deterministic} algorithm that can maintain an $n^{1+o(1)}$-edge $(1+\epsilon, n^{\alpha + o(1)})$-emulator $H$ of $G$ in total time $O(mn^{1-\alpha + o(1)})$.
\end{theorem}
\begin{theorem}\label{thm:adaptiveAdversarySpanner}
Given a partially dynamic unweighted graph $G=(V,E)$, a multiplicative approximation parameter $\epsilon > 0$, and an additive approximation parameter $\alpha \geq 0$, there exists a \emph{randomized} algorithm that can maintain an $n^{1+o(1)}$-edge $(1+\epsilon, n^{\alpha + o(1)})$-emulator $H$ of $G$ in total time $O(mn^{1-\alpha + o(1)})$. The algorithm works against an \emph{adaptive adversary}.
\end{theorem}

In order to implement a data structure as required by the theorem above, we heavily rely on the following result from \cite{gutenberg2020deterministic} which we extend to incremental graphs.

\begin{definition}[Refining and Coarsening]
Given a universe $U$ and a collection $\mathcal{C}$ of disjoint subsets of $U$ undergoing changes, we say that $\mathcal{C}$ is \emph{refining} if every set $C \in \mathcal{C}$ is a subset of a set $C'$ contained in an earlier version of $\mathcal{C}$. We say that $\mathcal{C}$ is \emph{coarsening} if every set $C \in \mathcal{C}$ is a superset of a set $C'$ contained in an earlier version of $\mathcal{C}$.
\end{definition}
\begin{theorem}[see \cite{gutenberg2020deterministic}, Definition 4.3, Lemma 4.1]
\label{thm:helperThmSODA}
Given a decremental/incremental unweighted graph $G=(V,E)$, a parameter $\mu$ and $0<\epsilon\leq 1$, then there exists a deterministic algorithm that maintains 
\begin{itemize}
    \item a refining/coarsening collection of pairwise-disjoint vertex sets $\mathcal{C} = \{C_1, C_2, \dots, C_k\}$ where any set $C_i$ has $\mathbf{diam}(G[C_i]) \leq 4(1/\epsilon)^{\sqrt{\log n}} \left\lceil\frac{|E(C_i)|}{\mu} \right\rceil$ and $k \leq m/ \mu$, and 
    \item a graph $H$ that forms a subgraph of $G$ such that for every shortest path $\pi_{s,t}$ in $G$ with $\pi_{s,t} \cap (\bigcup_i C_i) = \emptyset$, we have
    \[
        \dist_G(s,t) \leq \dist_H(s,t) \leq (1+\epsilon)\dist_G(s,t) + O((1/\epsilon)^{\sqrt{\log n}})
    \]
    and we have that $|H| = n^{1+o(1)}$.
\end{itemize}
The algorithm runs in total time $O(n\mu (1/\epsilon)^{1/\sqrt{\log n}} + m \log^2
n)$.
\end{theorem}
\begin{proof}
Let $\epsilon'=\epsilon/4$. We say a vertex $v \in V$, is $\mu$-\emph{heavy} if the graph $G$ induced by vertices in the ball $B(v, (1/\epsilon')^{\sqrt{\log n}})$ contains more than $\mu$ edges. Otherwise, we say $v$ is $\mu$-\emph{light}. We further maintain a graph $G^{heavy}$ that is the graph $G$ induced by the vertices $w \in V$ that are in the ball of some $\mu$-\emph{heavy} vertex $v \in V$. That is, a vertex $w$ is in $G^{heavy}$ if it is either $\mu$-heavy or has some vertex $v$ in its ball to depth $(1/\epsilon')^{\sqrt{\log n}}$ that is $\mu$-heavy. We let $\mathcal{C}$ be the set of connected components in $G^{heavy}$. Here, we point out that a slightly modified ES-tree can be used to maintain when a vertex transitions from being $\mu$-\emph{heavy} to $\mu$-\emph{light} (or vice versa) and while a vertex is $\mu$-\emph{light}, we can maintain its shortest path tree truncated at depth $(1/\epsilon')^{\sqrt{\log n}}$ explicitly. For partially dynamic graphs this modified ES-tree can be implemented in time $O(\mu (1/\epsilon')^{\sqrt{\log n}})$ per vertex, thus we can maintain the ES-trees in total update time $n^{1+o(1)}\mu$. 

Now, consider the diameter $d$ of some component $C_i \in \mathcal{C}$ at any stage in $G[C_i]$. Let $u$ and $v$ be two vertices in $C_i$ at distance $d$ in $G$. Observe that we then have at least $\lceil d/(4(1/\epsilon')^{\sqrt{\log n}})\rceil$ vertices in $C_i$ that are at distance at least $4(1/\epsilon')^{\sqrt{\log n}}$ from each other. Since each vertex in $C_i$ is at distance at most $(1/\epsilon')^{\sqrt{\log n}}$ to a $\mu$-heavy vertex, we can thus find $\lceil d/(4(1/\epsilon')^{\sqrt{\log n}})\rceil$  vertices that have disjoint balls and are $\mu$-heavy. Thus, we derive the upper bound on the diameter $d \leq 4(1/\epsilon')^{\sqrt{\log n}} \lceil\frac{E(C_i)}{\mu}\rceil$, by the pigeonhole principle. It is straightforward to see that $\mu$-heavy vertices in different components have mutually disjoint balls, so we derive the bound on the number of components by another application of the pigeonhole principle, i.e. we derive that there are at most $m/\mu$ connected components in $\mathcal{C}$ at any stage since each component contains at least one $\mu$-heavy vertex.

It is further not hard to see that when $G$ is a decremental graph, then $G^{heavy}$ is a decremental graph since the number of vertices in a ball to fixed radius can only decrease over time, as does the degree of each vertex. Thus, it is straightforward to maintain in $O(m\log^2 n)$ time the components in $G^{heavy}$ using a connectivity data structure (see \cite{holm2001poly, wulff2013faster}). For incremental graphs, $G^{heavy}$ is incremental and the same argument applies.  This also implies that $\mathcal{C}$ is refining in decremental graphs and coarsening in incremental graphs. 

To maintain $H$, we maintain greedily for each $i \leq \sqrt{\log n}$, a maximal set $A_i$ of \emph{active} vertices at level $i$ such that 
\begin{enumerate}
    \item every vertex in $A_i$ is $\mu$-\emph{light}, and
    \item every vertex $a \in A_i$ has $|B(a,(1/\epsilon')^{i+1})| \leq n^{(i+1)/\sqrt{\log n}}$ and $|B(a,(1/\epsilon')^{i})| \geq n^{i/\sqrt{\log n}}$, and
    \item any two vertices $a, a' \in A_i$ are at distance more than $(1/\epsilon')^{i}$ from each other.
\end{enumerate}
We then maintain $H$ as the union of the shortest path trees of vertices in $A_i$, for every $i$, truncated at depth $(1/\epsilon')^{i+1}$.

%Since we can use the information from the ES-trees since we only consider $\mu$-light vertices to maintain each $A_i$, and
Since each vertex joins $A_i$ at most once and leaves it again once, it is straightforward to see that each $A_i$ can be maintained in time $O(n)$ using the information from the ES-trees from $\mu$-$light$ vertices. Also, the spanner can be maintained in total time $n^{1+o(1)}\mu$ since $\mu$-\emph{light} vertices have at most $\mu$ edges in their truncated shortest-path tree.

To see that for every shortest path $\pi_{s,t}$ in $G$ with $\pi_{s,t} \cap (\bigcup_i C_i) = \emptyset$, we have
\[
    \dist_G(s,t) \leq \dist_H(s,t) \leq (1+\epsilon')\dist_G(s,t) + O((1/\epsilon')^{\sqrt{\log n}})
\]
we can construct $\pi_{s,t}$ by observing that $s$ is a $\mu$-\emph{light} vertex and by the pigeonhole principle, there exists an index $i$ such that $|B(s,(1/\epsilon')^{i+1})| \leq n^{1/\sqrt{\log n}}|B(s,(1/\epsilon')^{i})|$. But then either $s$ is in $A_i$ (in which case we set $a' = s$), or there is another vertex $a'\in A_i$ at distance at most $(1/\epsilon')^{i}$ from $s$. And, from this vertex $a'$, there is a shortest path tree up to depth $(1/\epsilon')^{i+1}$ in $H$. If $t$ is contained in the shortest path tree, then we are done since the path from $s$ to $a'$ to $t$ is only of length $n^{o(1)}$. Otherwise, let $s'$ be the vertex on $\pi_{s,t}$ that is farthest from $s$ and still in the truncated shortest path tree of $a'$. Then, $\dist_G(a',s)=(1/\epsilon')^{i+1}$. Thus, by the triangle inequality, $\dist_G(s,s')\geq(1/\epsilon')^{i+1}-(1/\epsilon')^i$. Thus, the path from $s$ to $a'$ to $s'$ in $H$ is a $\frac{(1/\epsilon')^{i+1}+(1/\epsilon')^i}{(1/\epsilon')^{i+1}-(1/\epsilon')^i}\leq(1+4\epsilon')$-multiplicative approximation of the path in $G$. We can then repeat for the path $\pi_{s',t}$. We thus derive multiplicative error of a $(1+4\epsilon')=(1+\epsilon)$ on the entire path and an additive error of $n^{o(1)}$ induced by the last segment. To see that $H$ never contains more than $n^{1+o(1)}$ edges, observe that active vertices $a'$ in $A_i$, for any level $i$, have disjoint balls $B(a',(1/\epsilon')^{i+1})$ of size at least $n^{i/\sqrt{\log n}}$ and adds at most  $n^{(i+1)/\sqrt{\log n}}$ edges to $H$. Thus, we can amortize the number of edges added over the number of vertices in its ball and obtain straight-forwardly that the total number of edges in $H$ is at most $\sum_i n^{1+1/\sqrt{\log n}} = n^{1+o(1)}$.
\end{proof}

Let us now give the description of our algorithm.

\paragraph{Algorithm.} We maintain the data structure from \Cref{thm:helperThmSODA} with $\mu = mn^{-\alpha}$. Our emulator consists of the graph $H$ from \Cref{thm:helperThmSODA} augmented with additional edges as described next.

Let us say a vertex $v \in V$ is $\mathcal{C}$-incident if there is a connected component $C_i \in \mathcal{C}$ with a vertex $w \in C_i$ such that there exists an edge $(v,w) \in E$. In particular, vertices in some $C_i$ are $\mathcal{C}$-incident. Let $\mathcal{C}_{inc}$ be the set of all vertices that are $\mathcal{C}$-incident. We maintain a connectivity data structure as described in \cite{holm2001poly, wulff2013faster} on the graph $G[\mathcal{C}_{inc}]$ and obtain connected components $\mathcal{D} = \{D_1, D_2, \dots \}$. It is not hard to see that $\mathcal{D}$ is refining in decremental graphs and coarsening in incremental graphs. Further,  each connected component $D_i$ in $G[\mathcal{C}_{inc}]$ consists of at least one connected component $C_i \in \mathcal{C}$ and its adjacent vertices from $\mathcal{C}_{inc}$, and each vertex of $\mathcal{C}_{inc}$ belongs to exactly one
connected component of $\mathcal{D}$.

Now consider some connected component $D_i$ which contains components $C_{i_1}, C_{i_2}, \dots, C_{i_k}$ for some $k \geq 1$, and let $w_i = \sum_{j \leq k} 2(1/\epsilon)^{\sqrt{\log n}} \left\lceil\frac{|E(C_{i_j})|}{\mu} \right\rceil$. We maintain for each such component $D_i$ a balanced binary spanning tree $T_i$ over the vertices (this tree is not a subgraph of $G$ but rather an arbitrary tree over the vertex set of a component), add the tree $T_i$ to the spanner $H$, and give each edge in $T_i$ weight $2^{\lceil\lg w_i\rceil}$ (that is the $w_i$ rounded up to the nearest power of $2$). This completes the description of the algorithm.

\paragraph{Emulator Approximation and Sparsity.} For the approximation ratio, consider any shortest path $\pi_{s,t}$ in $G$. Since $H \setminus \bigcup_i T_i$ forms a subgraph of $G$ and since each edge in $T_i$ for some $i$ has weight larger than the distance of its 
endpoints by \Cref{thm:helperThmSODA}, we have that the path is not underestimated in $H$.% (this is necessary to prove since $H$ is an emulator). 

Let $G / \mathcal{D}$ be the graph $G$ after contracting each $D_i$ into a supernode.
Observe that since the vertex sets of each $D_i$ are mutually disjoint, there is a one-to-one correspondence between sets $D_i$ and supernodes in $G / \mathcal{D}$. We also use $D_i$ to denote the supernode corresponding to the connected component $D_i \in \mathcal{D}$.

%Recall that each  vertex set $D_i$ in $\mathcal{D}$ induces a connected component of $G$.
Consider the shortest path $\pi'_{s,t}$ from $s$ to $t$ in $G / \mathcal{D}$.  
The path $\pi'_{s,t}$ clearly has smaller weight than $\pi_{s,t}$ in $G$. Let $d_1, d_2, \dots, d_j$ be the supernodes on $\pi'_{s,t}$ (in the order that they appear on $\pi'_{s,t}$) that correspond to a connected component $D_1, D_2, \dots, D_j$ in $\mathcal{D}$ respectively. 

Letting $d_0 = s$ and $d_{j+1} = t$, we have that every path segment $\pi'_{s,t}[d_i, d_{i+1}]$ for $0 \leq i \leq j$ is a $(1+\epsilon, n^{o(1)})$-approximate path segment even after remapping the first and last edge on the path segments again to a vertex in $G$ (i.e. to the endpoint that the edges have in $G$ instead of $G / \mathcal{D}$). 

For each supernode $d_i$ corresponding to a connected component $D_i$ in $\mathcal{D}$, we have that the two endpoints in $G$ that intersect between $\pi'_{s,t}$ and $D_i$ are connected by a path in $T_i$ consisting of $O(\log n)$ edges since $T_i$ is balanced. If the supernode $D_i$ contains connected components $C_{i_1}, C_{i_2}, \dots, C_{i_k}$, then each edge on this path has weight $w_i \leq 2 \cdot \sum_{j \leq k} 2(1/\epsilon)^{\sqrt{\log n}} \left\lceil\frac{|E(C_{i_j})|}{\mu} \right\rceil$.

We can now obtain an approximate path in $H$ for $\pi_{s,t}$ by taking the union of the path segments and by adding the exposed paths in the in the trees spanning the connected components $\mathcal{D}$. Since there are at most $m/\mu$ supernodes by \Cref{thm:helperThmSODA} and since each component $C_i \in \mathcal{C}$ is fully contained in exactly one supernode, we have that the total weight of the path from $s$ to $t$ in $H$ is at most $(1+\epsilon)\mathbf{dist}_G(s,t) + \frac{m}{\mu} n^{o(1)} \leq (1+\epsilon)\mathbf{dist}_G(s,t) + n^{\alpha + o(1)}$. The sparsity of $n^{1+o(1)}$ follows by the bound on $H$ in \Cref{thm:helperThmSODA} and the sparsity of the trees $T_i$.

\paragraph{Running time.} The running time of the connectivity data structure can be bound by $\tilde{O}(m)$. In decremental graphs, the time to maintain the balanced binary trees over the connected components $\mathcal{D}$ takes time at most $\tilde{O}(m)$ as whenever a connected component splits into two connected components, we can
remove the vertices of the smaller connected component from its tree $T$ and repair $T$ and then build a
separate tree $T'$ for the smaller component. Because each vertex that leaves $T$ has degree at most 3, both of these operations can be performed in time $\tilde O(k)$, where $k$ is the number of vertices of the smaller connected component. For incremental graphs, when we moerge two connected components, we merge the corresponding trees by destroying the smaller tree and adding each of its vertices to the larger tree.

The edge weights $w_i$ can be maintained using the binary tree and since each edge weight in a component can only decrease (increase), and then decreases (increase) by a factor $2$ the running time of these operations can be again bound by $\tilde{O}(m)$. Thus, the total update data is dominated by the data structure in  \Cref{thm:helperThmSODA} which runs in time $mn^{1-\alpha + o(1)}$, as desired.

\paragraph{A Randomized Algorithm for Spanners with High Additive Error.} Recently, Bernstein et al. \cite{bernstein2020fully} presented a fully-dynamic algorithm, with amortized update time $\tilde{O}(1)$ that maintains an $(\tilde{O}(1), 0)$-spanner of $G$ of sparsity $\tilde{O}(n)$ that is randomized but works against an adaptive adversary. It is straight-forward to see that using \Cref{thm:helperThmSODA} in conjunction with the above spanner, we can reuse the arguments from the approximation analysis of the emulator to derive \Cref{thm:adaptiveAdversarySpanner}.

%todo following section
%To do for full version!!!
% \section{Thoughts on multiplicative partially dynamic spanners and emulators}
% First, there are no spanners on truly subquadratic number of edges and multiplicative error less than $3$ -- just consider a complete bipartite graph.
% On the other hand, if we only want to preserve the distances between $p$ prespecified pairs of vertices, if $p\geq \Omega(n)$, every $n$ node undirected unweighted graph contains a spanner on $\tilde{O}((np)^{2/3})$ edges that preserves the $p$ pairwise distances exactly. Notably, if $p=O(n^{2-\delta})$ for $\delta>0$, the number of edges is truly subquadratic, $\tilde{O}(n^{2-2\delta/3})$.

% Suppose now that you wanted to maintain $p=\Theta(n^{2-\delta})$ pairs of distances under edge insertions (or only deletions), even within a factor of $2-\gamma$ for some $\gamma>0$, so that the number of edges in the pairwise spanner/emulator is $O(n^{2-\delta'})$ for some $\delta'>0$. Then I (Virginia) think we can show that under uMv the amortized update time needs to be $n^{1-\delta-o(1)}$. So for every $\eps>0$, there exists a $\delta>0$ so that if your update time is $O(n^{1-\eps})$ and you are achieving some $2-\gamma$ approximation with $O(n^{2-\delta'})$ edges for $\gamma,\delta'>0$, then you cannot be achieving the correct distance estimates for all $n^{2-\delta}$ prespecified pairs.

\section{Acknowledgments}
The authors would like to acknowledge Ivan Mikhailin for the ideas he contributed to an initial variant of the conditional lower bounds. The authors would also like to thank Ofir Geri for reminding them of these initial lower bounds.

Monika Henzinger is supported by the European Research Council under the European Community's Seventh Framework Programme (FP7/2007-2013) / ERC grant agreement No.~340506 and by the
Austrian Science Fund (FWF): Project No.~33775-N.

Maximilian Probst Gutenberg is supported by a start-up grant of Rasmus Kyng at ETH Zurich. This work was partially done while Maximilian Probst Gutenberg was at the University of Copenhagen where he was supported by Basic Algorithms Research Copenhagen (BARC), Thorup's Investigator Grant from the Villum Foundation under Grant No. 16582, and partially done while visiting MIT, supported by STIBOFONDEN’s IT Travel Grant for PhD Students.

Virginia Vassilevska Williams is supported by an NSF CAREER Award, NSF Grants CCF-1528078, CCF-1514339 and CCF-1909429, a BSF Grant BSF:2012338, a Google Research Fellowship and a Sloan Research Fellowship.

Nicole Wein is supported by NSF Grant CCF-1514339.
\bibliography{references.bib}
\end{document}